\documentclass[12pt,oneside]{amsart}
\usepackage[utf8]{inputenc}
\usepackage[a4paper]{geometry}

\usepackage{color}
\usepackage{mathtools}
\usepackage{enumitem}
\usepackage{amsthm}
\usepackage{amssymb}
\usepackage{stackrel}
\usepackage{xargs}[2008/03/08]
\usepackage[unicode=true,pdfusetitle,
 bookmarks=true,bookmarksnumbered=false,bookmarksopen=false,
 breaklinks=false,pdfborder={0 0 1},backref=false,colorlinks=false]
 {hyperref}
\usepackage{svg}

\makeatletter
\numberwithin{equation}{section}
\numberwithin{figure}{section}
\theoremstyle{plain}
\newtheorem{thm}{\protect\theoremname}[section]
\theoremstyle{definition}
\newtheorem{defn}[thm]{\protect\definitionname}
\theoremstyle{plain}
\newtheorem{lem}[thm]{\protect\lemmaname}
\theoremstyle{plain}

\theoremstyle{plain}
\newtheorem{prop}[thm]{\protect\propositionname}
\theoremstyle{plain}
\newtheorem{rem}[thm]{Remark}
\newtheorem{exa}[thm]{Example}

\@ifundefined{date}{}{\date{}}
\usepackage{amsthm}
\usepackage{longtable}

\usepackage{tikz}
\usetikzlibrary{shapes}
\usetikzlibrary{arrows.meta}

\DeclareSymbolFont{extraup}{U}{zavm}{m}{n}
\DeclareMathSymbol{\varheart}{\mathalpha}{extraup}{86}
\DeclareMathSymbol{\vardiamond}{\mathalpha}{extraup}{87} 

\makeatother

\providecommand{\corollaryname}{Corollary}
\providecommand{\definitionname}{Definition}
\providecommand{\lemmaname}{Lemma}
\providecommand{\theoremname}{Theorem}
\providecommand{\propositionname}{Proposition}

\begin{document}
\global\long\def\C{\mathbb{C}}%
\global\long\def\Cd{\C^{\delta}}%
\global\long\def\Cprim{\C^{\delta,\circ}}%
\global\long\def\Cdual{\C^{\delta,\bullet}}%
\global\long\def\Od{\Omega^{\delta}}%
\global\long\def\Oprim{\Omega^{\delta,\circ}}%
\global\long\def\Odual{\Omega^{\delta,\bullet}}%

\global\long\def\en{\mathcal{\varepsilon}}%
\global\long\def\wone{\widehat{w}}%
\global\long\def\wtwo{w}%
\global\long\def\wo{\wone}%
 
\global\long\def\wt{\wtwo}%
\global\long\def\cZ{\mathcal{Z}}%

\global\long\def\tfixed{wired}%
\global\long\def\tfree{free}%
\global\long\def\tplus{plus}%
\global\long\def\tminus{minus}%

\global\long\def\fixed{\{\mathtt{\tfixed}\}}%
\global\long\def\free{\{\mathtt{\tfree}\}}%
\global\long\def\plus{\{\mathtt{\tplus}\} }%
\global\long\def\minus{\left\{  \mathtt{\tminus}\right\}  }%

\global\long\def\Cgr{\mathcal{C}}%
\global\long\def\sCgr{\mathit{c}}%

\global\long\def\Pf{\mathrm{Pf}}%

\global\long\def\vv{v_{1},\dots,v_{n}}%
\global\long\def\uu{u_{1},\dots,u_{m}}%
\global\long\def\svv{\sigma_{v_{1}}\ldots\sigma_{v_{l}}}%
\global\long\def\muu{\mu_{v_{l+1}}\dots\mu_{v_{n}}}%

\global\long\def\pesm{\psi^{[\eta]}\!,\en,\mu,\sigma}%
\global\long\def\sfix#1{\sigma_{\mathrm{fix}}^{#1}}%

\global\long\def\any{\diamond}%
\global\long\def\anyother{\triangleright}%

\global\long\def\Opunc{\Omega^{\boxcircle}}%
\global\long\def\Onopunc{\Omega^{\boxempty}}%
\global\long\def\Oother{\Omega'}%

\global\long\def\const{\mathrm{const}}%

\global\long\def\P{\mathsf{\mathbb{P}}}%
 
\global\long\def\E{\mathsf{\mathbb{E}}}%
 
\global\long\def\sF{\mathcal{F}}%
 
\global\long\def\ind{\mathbb{I}}%

\global\long\def\R{\mathbb{R}}%
 
\global\long\def\Z{\mathbb{Z}}%
 
\global\long\def\N{\mathbb{N}}%
 
\global\long\def\Q{\mathbb{Q}}%

\global\long\def\C{\mathbb{C}}%
 
\global\long\def\Rsphere{\overline{\C}}%
 
\global\long\def\re{\Re\mathfrak{e}\,}%
 
\global\long\def\im{\Im\mathfrak{m}\,}%
 
\global\long\def\arg{\mathrm{arg}}%
 
\global\long\def\i{\mathfrak{i}}%
\global\long\def\eps{\varepsilon}%
\global\long\def\lamb{\lambda}%
\global\long\def\lambb{\bar{\lambda}}%

\global\long\def\D{\mathbb{D}}%

\global\long\def\HH{\mathbb{H}}%
 
\global\long\def\M{\mathbb{\mathcal{M}}}%
 
\global\long\def\Deps{\Omega_{\eps}}%
\global\long\def\DD{\hat{\Omega}}%

\global\long\def\dist{\mathrm{dist}}%
 
\global\long\def\reg{\mathrm{reg}}%

\global\long\def\half{\frac{1}{2}}%
 
\global\long\def\sgn{\mathrm{sgn}}%

\global\long\def\bdry{\partial}%
 
\global\long\def\cl#1{\overline{#1}}%

\global\long\def\diam{\mathrm{diam}}%
\global\long\def\corr#1{\overline{#1}}%
\global\long\def\Corr#1#2{\E_{#1}(#2)}%

\global\long\def\corr#1#2#3{\langle#1\rangle_{#2,#3}}%
\global\long\def\pa{\partial}%
\global\long\def\tto#1{\stackrel{#1}{\longrightarrow}}%

\global\long\def\res{\text{res}}%

\global\long\def\u{u}%
\global\long\def\v{v}%
 
\global\long\def\z{z}%
\global\long\def\mod{\;\mathrm{mod\;}}%
\global\long\def\wind{\mathrm{wind}}%
\global\long\def\vz{z^{\bullet}}%
\global\long\def\fz{z^{\circ}}%
\global\long\def\CF{\mathfrak{C}}%
\global\long\def\RPF{\text{I}}%

\global\long\def\FFS#1#2#3{F_{#2}^{#3}(#1)}%
\global\long\def\ds#1{\eta_{#1}}%
\global\long\def\dbar{\overline{\partial}}%
\global\long\def\dual#1{\left(#1\right)^{*}}%

\global\long\def\Eod{\Od_{+}}%
\global\long\def\GammaR{\Gamma_{\R}}%
\global\long\def\GammaiR{\Gamma_{i\R}}%

\global\long\def\ccor#1{\left\langle#1\right\rangle}%
\global\long\def\bar#1{\overline{#1}}%
\global\long\def\anypsi{\psi^{*}}%
\global\long\def\Op{\mathcal{O}}%
\global\long\def\crossing{c}%
\global\long\def\Fdual{F^{\mathrm{dual}}}%

\global\long\def\zz{z_{1},z_{2}}%
\newcommandx\norm[1][usedefault, addprefix=\global, 1=]{n_{#1}}%
\global\long\def\Ocvrc{\Omega_{\cvr}}%
\global\long\def\crad{\text{crad}}%
\global\long\def\feta{f^{[\eta]}}%
\global\long\def\T{\mathbb{T}}%
\global\long\def\reg{\sharp}%
\global\long\def\regg{*}%
\global\long\def\coefA{\mathcal{A}}%

\global\long\def\pbar#1{#1^{\star}}%
\global\long\def\fdag{f^{\star}}%
\global\long\def\CorrO#1{\langle#1\rangle_{\Omega}}%
\global\long\def\formL{\mathcal{L}}%
\global\long\def\appe{\approx_{\eps}}%
\global\long\def\jayhat{\hat{j}}%
\global\long\def\sqr#1{(#1)^{\frac{1}{2}}}%

\global\long\def\Ob{\mathcal{O}}%
\global\long\def\Dom{\Omega}%
\global\long\def\Domepsi{\Omega_{\eps_i}}%
\global\long\def\Domp{\Omega\setminus\{v_1,\dots,v_n\}}
\global\long\def\Domr{\Omega^\star}%
\global\long\def\Surf{\widehat \Omega}%
\global\long\def\Surfeps{{{\widehat \Omega}_\eps}}%
\global\long\def\Surfo{\widehat \Omega_0}%
\global\long\def\bpoints{\{b_1,\dots,b_{2k}\}}%
\global\long\def\SKDom{K_{\Omega_\eps}}
\global\long\def\SKDr{K^\star_{\Omega_\eps}}
\global\long\def\SKDomLim{K_0}
\global\long\def\SKDrLim{K^\star_{0}}
\global\long\def\SpStr{\mathfrak{c}}
\global\long\def\Obs#1{\mathcal{O}\left[#1\right]}%
\global\long\def\hOp{\hat{\mathcal{O}}}
\global\long\def\Ab{\mathcal{U}}

\global\long\def\ds{\mathrm{ds}}
\global\long\def\ns{\mathrm{ns}}
\global\long\def\cs{\mathrm{cs}}

\global\long\def\nord#1{:#1:}
\global\long\def\Nexp#1#2{:e^{#1\Phi(#2)}:}
\global\long\def\nexp#1#2{:e^{#1\varphi(#2)}:}

\global\long\def\bsigma#1{\sqrt{2}:\cos \left(\frac{\sqrt{2}}{2}\Phi(#1)\right):}
\global\long\def\bmu#1{\sqrt{2}:\sin \left(\frac{\sqrt{2}}{2}\Phi(#1)\right):}
\global\long\def\ben#1{-\frac12 :|\nabla \Phi(#1)|^2:}
\global\long\def\bpsi#1{2\sqrt{2}\i\partial \Phi(#1)}
\global\long\def\bpstar#1{-2\sqrt{2}\i\bar \partial \Phi(#1)}
\global\long\def\bpsipstar#1{e^{-\frac{1}{2}g_\Omega(#1,#1)}:\sin 2\sqrt{\pi}\Phi(#1):}

\global\long\def\setS{\mathcal{S}}
\global\long\def\setSn{\mathcal{S}_0}
\global\long\def\gap{\alpha}
\global\long\def\inv{\iota}

\title[Bosonization of the critical Ising model]{Bosonization of primary fields for the critical Ising model on multiply connected planar domains}
\author{Baran Bayraktaroglu, Konstantin Izyurov, Tuomas Virtanen And Christian
Webb}
\begin{abstract}
We prove bosonization identities for the scaling limits of the critical Ising correlations in finitely-connected planar domains, expressing those in terms of correlations of the compactified Gaussian free field. This, in particular, yields explicit expressions for the Ising correlations in terms of domain's period matrix, Green's function, harmonic measures of boundary components and arcs, or alternatively, Abelian differentials on the Schottky double. 

Our proof is based on a limiting version of a classical identity due to D.~Hejhal and J.~Fay relating Szeg\H{o} kernels and Abelian differentials on Riemann surfaces, and a systematic use of operator product expansions both for the Ising and the bosonic correlations.  
\end{abstract}
\maketitle
\tableofcontents{}

\section{Introduction}

The classical two-dimensional Ising model is one of most studied models of mathematical physics. Its importance is due to it being one of the simplest models where complicated phenomena such as phase transitions can be studied rigorously. The two-dimensional Ising model is even more special in that since the work of Onsager and Kaufman \cite{Onsager, OK} it has been regarded as \emph{exactly solvable}. 

What one means by that is that relevant quantities in the model can be calculated explicitly. Of course, the precise meaning of ``exact solvability" then depends on what these quantities are and what qualifies as ``explicit". The early work of Onsager and Kaufmann succeeded in calculating the free energy (per lattice site) of the model. Later on, similar or other combinatorial methods were used to calculate other thermodynamically relevant quantities, and eventually calculate the \emph{scaling exponents}, such as the famous $\frac18$ magnetization exponent of Onsager and Yang \cite{Yang}. 

Beyond the scaling exponents, \emph{correlation functions} 
carry even more refined information on the model; from a physicist's point of view, their computation yields a complete understanding of the model. In the case of the Ising model, a lot of progress was made on computing these correlation functions in various regimes \cite{MW,Palmer}, but the combinatorial methods were mostly limited simple geometries such as the full plane or a torus, and the critical case remained notoriously difficult. Underlying Onsager's solution of the Ising model is the fact that it combinatorially corresponds to a \emph{free fermionic theory}. However, the most natural observables in the theory are not fermions but spins. Expressing those in terms of fermions and analysing the resulting expressions proved to be challenging.

In a complete change of perspective, Belavin, Polyakov and Zamolodchikov \cite{BPZ} postulated that at criticality, the Ising model must have a scaling limit described by a conformal field theory (CFT). What's more, they identified the relevant CFT as one of the \emph{minimal models}. These are the simplest of CFTs, they are exactly solvable, in the sense that their correlation functions can be explicitly calculated. In a concrete application of their theory (a tiny one compared to the scope of the theory itself), Belavin, Polyakov and Zamolodchikov computed the 4-point spin correlation function in the full plane.

In fact, the CFT approach gives several different prescriptions on how to calculate the scaling limit of the Ising correlations, yielding \emph{different} results of varying degree of ``explicitness"; for a curious example of this phenomenon, see the discussion around \cite[Eq. 3.42--3.43]{Felder}. Part of the prediction is that the correlations must be conformally covariant, see \eqref{eq: conf_cov} below. Moreover, the relevant CFT being a minimal model implies that they satisfy second-order partial differential equations, known as the BPZ equations. In the four-point case, these equations can be reduced by conformal covariance to a (hypergeometric) ODE. In general, such reduction is not possible, and clever methods were developed for solving BPZ equations \cite{DotsenkoFateev, Felder} yielding solutions in terms of contour integrals. Whether BPZ equations can be always solved by those methods is still an active topic of research \cite{Flores-Kleban, KytolaPeltola, FloresPeltola, Sussman}.   

Another way of computing the correlation functions, that can be applied to all minimal models, is the conformal bootstrap approach. In the case of the Ising model, the method that has eventually lead to the nicest explicit formulae for the correlations in the full plane and on a torus \cite{DiFSZ} and in the half-plane \cite{BG} is \emph{bosonization}; see also \cite[Chapter 12]{DFMS}. The meaning of bosonization is that the correlation functions of the critical 2d Ising model can be related to correlation functions of a free bosonic field theory. In the case of the full plane and half-plane (with homogeneous boundary conditions), the relevant theory is just a \emph{Gaussian free field}, whose correlations are explicit in the strongest possible sense -- thus so are the Ising correlations.  

From the mathematical standpoint, most of the CFT arguments are still lacking rigorous grounds. In particular, how the Ising model (as a critical lattice model) leads to a minimal model of CFT is still not completely understood; see \cite{HonKytVik, AmKytetc, AmKytetc2} for steps in this direction. The situation is somewhat better on the level of correlation functions, as techniques of discrete complex analysis allowed to prove the existence, the conformal covariance, and universality of their scaling limits. In \cite{CHI1,CHI2}, it was proven that scaling limits of the correlations (of primary fields) in a critical Ising model on finitely connected planar domains exist and are conformally covariant. Part of the proof is a description of the limits, which, as the reader might already have anticipated, is different from the descriptions obtained by other methods. Namely, in \cite{CHI2}, the scaling limits of correlations are expressed in terms of solutions to Riemann boundary value problems. 

In nice geometries, such as a half-plane or an annulus, these boundary value problems can be reduced to solving linear systems of equations, whose coefficients are explicit algebraic (quadratic irrational) or elliptic functions, and whose size grows linearly with the number of points considered. The spin or spin-disorder correlation functions (which is the most difficult case) are then given by exponentials of integrals of the solutions to those systems. Although in some sense, this may count as an ``explicit" answer, it is not quite as explicit as the formulae in \cite{DiFSZ,BG}. In some cases, these expressions can be then simplified -- this has been done in the simply-connected case and for spin one-point function in the doubly connected case, see \cite[Theorem 1.2 and Appendix A2]{CHI1} and \cite[Section 7]{CHI2}; of course, in the simply connected case, the result is in agreement with the bosonization predictions, see \cite[Remark 7.3]{CHI2}.

As explained in \cite[Section 7]{CHI2}, once the spin-disorder correlations are computed, other correlations of primary fields can, in principle, be derived from those. However, in practice the procedure outlined in \cite[Section 7]{CHI2} quickly gets quite messy; what's more, there's more than one way to write down any given correlation, which are not obviously equivalent. As an extreme example of that, the prescription for the fermionic correlation  $\ccor{\psi_{z_1}\dots\psi_{z_{2n}}}_{\C}$ is to consider the leading term of the asymptotics of a $4n$-point spin-disorder correlation $\ccor{\sigma_{z_1}\dots\sigma_{z_{2n}}\mu_{u_1}\dots\mu_{u_{2n}}}_\C$ as $u_i\to z_i$; it is far from straightforward that the result is given simply by $\ccor{\psi_{z_1}\dots\psi_{z_{2n}}}_{\C}=\Pf\left[\frac{2}{z_i-z_j}\right]$.

The goal of this article is to provide a systematic way of writing down explicit formulae for the scaling limits of the critical Ising correlations in finitely-connected domains, as computed in \cite{CHI2}, by proving the bosonization prescriptions. Indeed, in Theorem \ref{th:main} below, we show that the \emph{squares} of correlation functions of primary fields of the Ising model on multiply connected planar domains are equal to certain correlations for the \emph{compactified free field}. Those, in their turn, are given by explicit expressions involving the domain's period matrix, Green's function, harmonic measures of boundary components and boundary arcs, and derivatives thereof. To be completely precise, the expressions involve exponentials of Green's function, polynomials in its derivatives, and multivariate theta functions of the harmonic measures (a.k.a. the Abel-Jacobi map). 

We expect that the bosonization procedure can be carried out even in greater generality, such as general Riemann surfaces and even for the off-critical Ising model, but we do not discuss these generalizations further here.

Concerning mathematically rigorous analysis of bosonization in the setting of the Ising model, we mention the results on the exact bosonization on the lattice \cite{Dubedat, Hugoetc}, see the discussion after Theorem \ref{th:main}, and \cite{JSW}, where bosonization of the spin field of the Ising model in a simply connected domain is studied from the point of view of random generalized functions. For bosonizing free massless fermions (which are also relevant for the critical Ising model), there is of course also the classical article \cite{ML}.

\subsection{The setup and the main result}
\label{sec:isingintro}
For the Ising correlations, we will follow the setup and notation of \cite{CHI2}. Let $\Omega\subset \R^2$ denote a bounded finitely connected planar domain 
 equipped with \emph{boundary conditions}, defined by a subdivision of $\pa\Omega$ into two subsets $\fixed$ and $\free$, each consisting of finitely many open arcs (the points $b_1,\dots,b_{2k}\in \pa \Omega$ separating them belong to neither of them). We will assume that the notation $\Omega$ incorporates both the domain and the boundary conditions. See Figure \ref{fig:freewired} for an illustration of this setting.

In \cite[Section 5.2]{CHI2}, the following critical Ising correlation functions on $\Omega$ were defined:
\begin{equation}
\label{eq: ccor}
\ccor{\Op_{z_1}\cdot\dots\cdot\Op_{z_N}}_{\Omega},
\end{equation}
where $z_i$ are distinct points in $\Omega$, and $\Op_{z_i}\in\{\sigma_{z_i},\mu_{z_i},\en_{z_i},\psi_{z_i},\psi^\star_{z_i}\}$ are (labels for) the \emph{primary fields} in the Ising model, called spin, disorder, energy, fermion, and conjugate fermion, respectively. (We will quite often drop the domain $\Omega$ from the notation.) Although the term ``primary field" comes from Conformal Field Theory literature, we do not claim that the correlations \eqref{eq: ccor} are actually correlations in a CFT constructed in any rigorous sense; for the purpose of this paper, this is simply a collection of functions (more precisely, two-valued functions defined up to sign) of several variables in $\Omega$. The justification for calling them ``correlation functions" stems from \cite[Theorem 1.2]{CHI2}, which shows that they are limits of suitably renormalized observables in the critical Ising model on fine mesh discretizations $\Od$ of $\Omega$, with free boundary conditions on $\free$ and locally monochromatic ones on $\fixed$, i.e., the spins are conditioned to be the same on all wired boundary arcs on the same connected component of $\pa\Omega$.\footnote{To be precise, \cite[Theorem 1.2]{CHI2}, features ``real fermions'' $\psi_z^{[\eta]}$ with $\eta\in \C$, which can be expressed in terms of $\psi$ and $\psi^*$ in the following manner (see \cite[(5.8)]{CHI2}): for any $\Op$ which is a chain of primaries as in \eqref{eq: ccor}
\[
\langle \psi_z^{[\eta]}\Op\rangle_\Omega=\frac{1}{2}\eta\langle \psi_z^* \Op\rangle_\Omega+\frac{1}{2}\bar \eta\langle \psi_z \Op\rangle_\Omega. 
\]
By varying $\eta$, one can recover the correlation functions \eqref{eq: ccor} from \cite[Theorem 1.2]{CHI2}.
}  
In \cite[Theorem 5.20]{CHI2} it was shown that the correlation functions \eqref{eq: ccor} obey the following simple covariance rule under conformal maps:
\begin{equation}
\label{eq: conf_cov}
\ccor{\Op_{z_1}\cdot\dots\cdot\Op_{z_N}}_{\Omega}=\prod^N_{i=1}\Psi'(z_i)^{\Delta_i} \prod^N_{i=1}\overline{\Psi'(z_i)}^{\Delta'_i}\ccor{\Op_{\Psi(z_1)}\cdot\dots\cdot\Op_{\Psi(z_N)}}_{\Psi(\Omega)},
\end{equation}
with the \emph{conformal weights} given by the following table:

\medskip
\renewcommand{\arraystretch}{1.2}
\begin{center}
\begin{tabular}{|c|c|c|}
\hline
$\Op_{z_i}$&$\Delta_i$&$\Delta'_i$\\
$\sigma_{z_i}$ & $\frac{1}{16}$ & $\frac{1}{16}$ \\
$\mu_{z_i}$ & $\frac{1}{16}$ & $\frac{1}{16}$ \\
$\eps_{z_i}$ & $\frac{1}{2}$ & $\frac{1}{2}$ \\
$\psi_{z_i}$ & $\frac{1}{2}$ & $0$ \\
$\psi^\star_{z_i}$ & $0$ & $\frac{1}{2}$ \\
\hline
\end{tabular}
\end{center}
\medskip

Because of \eqref{eq: conf_cov}, we may restrict our attention to domains whose boundary consists of disjoint analytic Jordan curves, or even \emph{circular} domains whose boundaries consist of disjoint circles, as any domain can be mapped to one of those, see e.g. \cite[Theorem 7.9]{Conway}. Thus, in the simply connected (respectively, doubly connected) case, one only needs to compute the correlations in the full plane $\C$ and the upper half-plane $\HH$ (respectively, annuli). In higher connectivity, it will be sometimes convenient to assume the domain to be circular, but this is not really essential for our arguments.

\begin{figure}
    \centering
\includegraphics{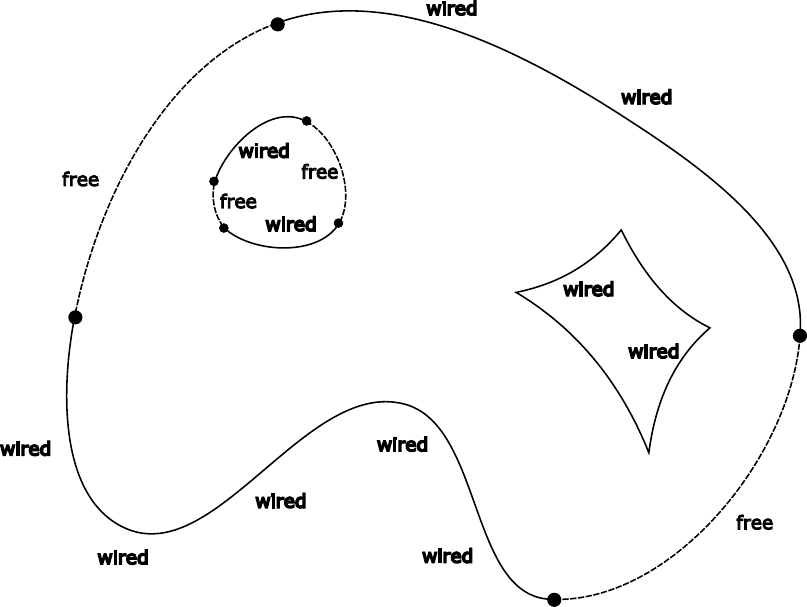}
    \caption{An illustration of a triply connected planar domain with boundary conditions. The dashed part of the boundary forms the set $\{\tt{free}\}$ and solid parts form the set $\{\tt{wired}\}$. The points where the boundary conditions changes from fixed to free are marked as dots.}
    \label{fig:freewired}
\end{figure}


As indicated earlier, the main goal of this article is to relate the correlation functions \eqref{eq: ccor} to correlations in a free bosonic theory, namely, that of a \emph{compactified free field}. The latter is defined as a sum of two \emph{independent} components, 
\[
\Phi=\varphi+\xi.
\]
Here $\varphi$ is the \emph{Gaussian free field} in $\Omega$ with Dirichlet boundary conditions, i.e., the centered Gaussian field in $\Omega$ with covariance given by the Green's function $G_\Omega(x,y)=\left(\frac{-\Delta}{2\pi}\right)^{-1}$, see e.g. \cite{Sheffield} and Section \ref{sec: bosonic}. The \emph{instanton component} $\xi$ is a harmonic function in $\Omega$ picked at random from the countable set of all harmonic functions with the following boundary conditions:
\begin{itemize}
\item $\xi(x)\in \sqrt{2}\pi\Z$ on $\fixed$, $\xi(x)\in \sqrt2\pi(\Z+\frac12)$ on $\free$;
\item $\xi$ is constant on each wired and each free arc, moreover, it assumes the same values along $\fixed$ arcs in the same boundary component. 
\item the value of $\xi$ on adjacent fixed and free arcs differ by $\pm \frac{\sqrt{2}}{2}\pi$
\item one of the boundary arcs is marked, and the value of $\xi$ is fixed there (e.g., to $0$ if the arc is wired or to $\frac{\sqrt{2}}{2}\pi$ if the arc is free).
\end{itemize}
The probability to choose a particular $\xi$ is proportional to 
\begin{equation}
\label{eq: proba_xi}
\exp\left(-\frac{1}{4\pi}\langle\nabla \xi,\nabla \xi\rangle_{\Omega,\mathrm{reg}}\right),
\end{equation}
where $\langle\nabla \xi,\nabla \xi\rangle_{\Omega,\mathrm{reg}}$ is the \emph{regularized Dirichlet energy} of $\xi$, see Section \ref{sec: bosonic} below, cf. \cite{Dubedat09}. The need for regularization comes from divergence of the Dirichlet energy at the jump points $b_1,\dots,b_{2k}$; in particular, if there are none, then it is the usual Dirichlet energy $\int_\Omega|\nabla \xi|^2.$ We note that the measure \eqref{eq: proba_xi} and all the observables we will consider will be invariant under a shift of $\xi$ by $\sqrt2 \pi n$, $n\in\Z$, which is why the choice of a marked arc above is unimportant. 

Let us motivate the term ``compactified free field", cf. \cite[Section 2]{Dubedat15}; the arguments in this paragraph are heuristic. Assume for simplicity that there are no free arcs, and recall that the Gaussian free field is a standard Gaussian over the Sobolev space $H_{0}^1(\Omega)$ with zero Dirichlet boundary conditions, i.e., it can be thought of as being sampled from all functions with probability proportional to $\exp\left(-\frac{1}{4\pi}\int_\Omega |\nabla \varphi|^2\right)$. Note that integration by parts shows that the harmonic functions $\xi$ are orthogonal to $H_{0}^1(\Omega)$, i.e., we can write 
$$
\int_\Omega |\nabla \Phi|^2=\int_\Omega |\nabla (\varphi+\xi)|^2=\int_\Omega|\nabla\varphi|^2+\int_\Omega|\nabla\xi|^2.
$$
Therefore, one can think of $\Phi$ as being sampled from all possible functions with the boundary conditions as above, with probability proportional to $\exp\left(-\frac{1}{4\pi}\int_\Omega |\nabla \Phi|^2\right)$. Since we don't care about shifts by integer multiples of $\sqrt2 \pi$, we can consider instead $\Psi=\frac{1}{\sqrt{2}}\exp(\sqrt2 \i \Phi)$, which is a ``field" with values on the circle $\frac{1}{\sqrt2}\T$, sampled with probability proportional to $\exp\left(-\frac{1}{4\pi}\int_\Omega |\nabla \Psi|^2\right)$, since $|\nabla \Psi|\equiv|\nabla \Phi|$. The boundary conditions simply become $\Psi\equiv \frac{1}{\sqrt2}$ on $\pa \Omega$.

Our main result is as follows:
\begin{thm}
We have the identity 
\label{th:main}
\begin{equation}
\label{eq: Thm1}
\ccor{\Op_{z_1}\dots\Op_{z_N}}^2_{\Omega}=\ccor{\hOp_{z_1}\dots\hOp_{z_N}}_{\Omega}.
\end{equation}
whenever both $|\{i:\Op_{z_i}\in\{\sigma_{z_i},\psi_{z_i},\psi^\star_{z_i}\}\}|$ and $|\{i:\Op_{z_i}\in\{\mu_{z_i},\psi_{z_i},\psi^\star_{z_i}\}\}|$ are even, and each pair $\Op_{z_i}$, $\hOp_{z_i}$  fits into a row of the following table:
\medskip
 
\begin{center}
\begin{tabular}{|c|c|c|}
\hline
$\Op_{z_i}$&$\hOp_{z_i}$\\
\hline
$\sigma_{z_i}$&$\bsigma{z_i}$\\
$\mu_{z_i}$&$\bmu{z_i}$\\
$\en_{z_i}$&$\ben{z_i}$\\
$\psi_{z_i}$&$\bpsi{z_i}$\\
$\psi^\star_{z_i}$&$\bpstar{z_i}$\\
\hline
\end{tabular}
\end{center}
\end{thm}
\medskip

\begin{rem}
    If any of the parity conditions is not satisfied, then the correlation on the left-hand side of \eqref{eq: Thm1} is zero. The correlation on the right-hand side may not be zero, and in that case it corresponds to an Ising correlation with different boundary conditions, see Section \ref{sec: extensions}.
\end{rem}

Several remarks are in order. First, as is the case with the Ising correlations, we \emph{do not} view the correlations in the right-hand side of \eqref{eq: Thm1} as expectations of actual random variables, or correlations in a quantum or statistical field theory constructed in any rigorous sense, although results in this direction covering some of our observables exist \cite{KM,JSW}. The issue, perhaps familiar to the reader, is that $\Phi$ is too rough for the objects like $\cos\left(\frac{\sqrt{2}}{2}\Phi(z)\right)$ or $|\nabla\Phi(z)|^2$ to make sense as random variables (in the latter case, not even when averaged over a test function). Their correlations, however, can be defined by a standard procedure of first making $\Phi$ into a smooth field (by convolution with a mollifier), and then passing to a limit. In some of the cases, for the limit to exist, we need to renormalize the field (by a multiplicative factor for $\sin$ and $\cos$, and by an additive term for $|\nabla\Phi(z)|^2$); the colons $:\cdot:$ indicate that this procedure has been applied. This corresponds to normal (or Wick) ordering in quantum field theory, although we will not use this connection in earnest. At the end of the day, the right-hand side of \eqref{eq: Thm1} is given by a concrete formula in terms of the Green's function and harmonic measures in $\Omega$, see Definition \ref{def: bos_corr}, and this formula is the only thing we need to know about it.

Second, there is a combinatorial explanation underlying Theorem \ref{th:main}. A squared correlation function as in \eqref{eq: Thm1} can be interpreted as a correlation $\ccor{\Op_{z_1}\tilde{\Op}_{z_1}\dots\Op_{z_N}\tilde{\Op}_{z_N}}_{\Omega}$, where $\Op_{z_i}$ and $\tilde{\Op}_{z_i}$ are the corresponding observables in \emph{two independent copies} of the Ising model. It has been shown \cite{Dubedat, Hugoetc} how, already at the discrete level, two such independent copies correspond to a dimer model whose height function is known to converge to the Gaussian free field; moreover, spin and disorder correlations have been shown to correspond, again on the discrete level, to so-called \emph{electric correlators} in the height function picture. In \cite{Dubedat15}, Dub\'edat proved convergence of these electric correlators to the corresponding free field observables, which, combined with \cite{Dubedat}, lead to an alternative proof of convergence of the spin and disorder Ising correlations in the case $\Omega=\C$, which manifestly gives \eqref{eq: Thm1}; see also \cite{Hugoetc}. Carrying out this program more generally could also lead to a proof of Theorem \ref{th:main}. The convergence of height functions to compactified GFF was recently proven in an independent work of Basok \cite{Basok}, although it does not seem that his result fully covers the boundary conditions considered in this paper. In addition, going from convergence of height function to the convergence of all of the required observables involves a lot of work. To sum up, even if the program just described is feasible, a purely analytic proof of Theorem \ref{th:main}, building on the already existing convergence results of \cite{CHI2}, is also of interest.

Third, the simplest case of \eqref{eq: Thm1}, corresponding to $\Omega$ being the upper half-plane $\HH$ and  $\Op_{z_i}=\psi_{z_i}$ for all $i$, is already interesting; indeed, it states that $$\ccor{\psi_{z_1}\dots\psi_{z_{2M}}}^2_{\Omega}=(-8)^M\ccor{\pa\varphi(z_1)\dots\pa
\varphi(z_{2M})}_{\Omega}.$$
The fermionic correlators satisfy the fermionic Wick's rule, i.e., they are given by a Pfaffian of two-point correlations, whereas $\pa\varphi$ are Gaussians and hence satisfy the bosonic Wick's rule, i.e., they are given by what is sometimes called a Hafnian. All in all, this leads to the well known identity \cite[Eq. 12.53]{DFMS}:
$$
\left(\sum_{p}(-1)^{i(p)}\prod_{{a,b}\in p}\frac{1}{z_a-z_{b}}\right)^2=\sum_{p}\prod_{{a,b}\in p}\frac{1}{(z_a-z_{b})^2},
$$
where the sums are over all pairings (i.e, partitions of $\{1,\dots,2M\}$ into two-element sets) and the sign $i(p)$ is the parity of the pairing (equal to the number of intersections if $p$ is drawn as a planar link pattern). This identity can be derived from the Cauchy determinant; curiously, we do not use this identity, or its analogs for domains of higher connectivity, yielding an independent proof.

\subsection{Outline of the proof}

Our proof of Theorem \ref{th:main} relies on the following particular case, which we state as a separate theorem:

\begin{thm}\label{th:ds}
Theorem \ref{th:main} holds true in the particular case when for every $i$, either $\Op_{z_i}=\sigma_{z_i}$, or $\Op_{z_i}=\mu_{z_i}$. (We are still assuming that $|\{i:\Op_{z_i}=\sigma_{z_i}\}|$ and $|\{i:\Op_{z_i}=\mu_{z_i}\}|$ are both even.)
\end{thm}
Given Theorem \ref{th:ds}, the derivation of Theorem \ref{th:main} is based on the observation that other Ising fields featured therein can be obtained from $\sigma$ and $\mu$ by fusion, i.e., by considering the asymptotic expansions as some of the points merge. Thus, the proof boils down to checking the fusion rules, or operator product expansions, on the bosonic side, and comparing them with known results on the Ising side \cite[Section 6]{CHI2}. We carry this out in Section \ref{sec:ff}.

The remainder of the article is concerned with the proof of Theorem \ref{th:ds}. The main step in the proof is a derivation of the identity 
\begin{equation}
\label{eq: for_log_der}
\left(\frac{\ccor{\psi_z\psi_w\Op_{v_1}\dots\Op_{v_n}}_\Omega}{\ccor{\Op_{v_1}\dots\Op_{v_n}}_\Omega}\right)^2=-8\frac{\ccor{\pa\Phi(z)\pa\Phi(w)\hOp_{v_1}\dots\hOp_{v_n}}_\Omega}{\ccor{\hOp_{v_1}\dots\hOp_{v_n}}_\Omega},
\end{equation}
 where each $\Op_{v_i}$ is either $\sigma_{v_i}$, or $\mu_{v_i}$, and we assume that the denominators are non-vanishing. The identity \eqref{eq: for_log_der} turns out to be a limiting form of a classical identity in the theory of Riemann surfaces, due to Hejhal \cite{Hejhal} and independently Fay \cite{Fay}, that expresses squares of Szeg\H{o} kernels of a compact Riemann surface in terms of Abelian differentials. The Riemann surface in question is the Schottky double of the domain $\Omega$, to which $n+k$ small handles are attached, where $n$ is the number of spins and disorders in $\eqref{eq: for_log_der}$ and $k$ is the number of free arcs, see Section \ref{sec: gluing}. In Section \ref{sec:hej}, we review the theory of Riemann surfaces needed to state the Hejhal--Fay identity. In Section \ref{sec: pinching}, we check that as we pinch the handles, the two sides of the Hejhal--Fay identity converge to the respective sides of \eqref{eq: for_log_der}, thus proving the latter. It is likely that one could adapt one of the existing proofs of Hejhal--Fay identity to give a direct proof of \eqref{eq: for_log_der}, however we found it easier to give a limiting argument.
 
 Using the operator product expansions again, \eqref{eq: for_log_der} implies 
 \begin{equation}
\label{eq: log_der}
 \pa_{v_1}\log\ccor{\Op_{v_1}\dots\Op_{v_n}}_\Omega=\frac12\pa_{v_1}\log\ccor{\hOp_{v_1}\dots\hOp_{v_n}}_\Omega.
 \end{equation}
This is the same as \eqref{eq: Thm1} up to a global normalization factor, which can be fixed using the operator product expansions on both sides once more. Some gimmicks are needed to work around the possible vanishing of the denominators. The route from \eqref{eq: for_log_der} to \eqref{eq: log_der} to \eqref{eq: Thm1} is closely parallel to the proof of convergence of the Ising correlations in \cite{CHI1,CHI2}. We carry out this part of the proof in Section \ref{sec: conclusion}. 
 
\subsection{Acknowledgements}

T.V. and C.W.  were supported by the Emil Aaltonen Foundation. C.W. was supported by the Academy of Finland through the grant 348452. T.V. is also grateful for the financial support from the Doctoral Network in Information Technologies and Mathematics at Åbo Akademi University. B. B. and K.I. were supported by Academy of Finland through academy project ``Critical phenomena in dimension two". We are grateful to Mikhail Basok for useful remarks.

\medskip 

\section{Bosonic correlation functions and operator product expansions: Proof of Theorem \ref{th:main} given Theorem \ref{th:ds}}\label{sec:ff}

In this section, we define the correlation functions appearing on the right-hand side of \eqref{eq: Thm1}, and study their properties, most importantly, the \emph{Operator Product Expansions (OPE)} -- these are the asymptotic expansions when two marked points collide. These, together with the Ising OPE of \cite[Section 6]{CHI2}, are key to the derivation of Theorem \ref{th:main} from Theorem \ref{th:ds}, but they are also used in the proof of the latter. 

Similar OPE were studied in the literature -- see in particular \cite[Lecture 3]{KM}. However, some of the fields we encounter (such as $:|\nabla \Phi|^2:$) differ from those studied in \cite{KM}, and the focus of \cite{KM} is on simply connected domains and the purely Gaussian case. For these reasons, we give a self-contained presentation.

\subsection{The instanton component}
\label{sec: instanton}

As we discussed in the introduction, the \emph{instanton component} $\xi$ is a random harmonic function in $\Omega$, which we now describe somewhat more concretely. We assume for simplicity that the $\fixed$ part of $\pa\Omega$ is non-empty, and $\xi$ is normalized to be zero on one of the wired boundary arcs, and let $h_1,\dots,h_g$ denote the harmonic measures of other boundary components of $\Omega$, and $\hat{h}_1,\dots,\hat{h}_{k}$ denote the harmonic measures of the free boundary arcs. (If all boundary conditions are free, the only change is to normalize $\xi$ to be $\frac{\sqrt{2}\pi}{2}$ rather than $0$ on one of the components) Any realization of $\xi$ can be uniquely written as 
\begin{equation}
\xi(z)=\frac{\gap}{2}\left(\sum_{i=1}^{g}s_ih_i(z)+\sum_{i=1}^k \hat{s}_i\hat{h}_{i}(z)\right)
\label{eq: instanton}
\end{equation}   
where $\gap=\sqrt{2}\pi$, $s_i=2m_i+N_i$ with $m=(m_1,\dots,m_g)\in\Z^g$ and $\hat{s}=(\hat{s}_1\dots \hat{s}_k)\in\{\pm 1\}^k$, and $N_i=1$ if the $i$-th boundary component is entirely free, and $N_i=0$ otherwise. Note that $s$ and $\hat{s}$ can also be viewed as random variables.

The \emph{Dirichlet energy} $\langle\nabla\xi,\nabla\xi\rangle_{\Omega}=\int_\Omega|\nabla \xi|^2$ is in general infinite, because of the divergence near the points where boundary conditions have a jump. For a harmonic function $g$ which is piecewise constant at the boundary except for jumps at $b_1,\dots,b_{r}$ of size $\alpha_1,\dots,\alpha_r$ respectively, we can define the \emph{regularized Dirichlet energy} of $g$ by subtracting the leading divergent term at each $b_1,\dots,b_{r}$, i.e., we put 
$$
\langle\nabla g,\nabla g \rangle_{\Omega,\mathrm{reg}}=\lim_{\eps\to 0}\left(\langle\nabla g,\nabla g\rangle_{\Omega\setminus\cup_{i=1}^{r}B_\eps(b_i)}+\sum_{i=1}^{r}2{\alpha_i}^2\log\eps\right).
$$
Checking the existence of the limit is left to the reader. It is important to note for $g=\xi$, the last term equals $2\pi^2 k\log\eps$, i.e., it is non-random, since all the jumps are always $\pm\frac{\gap}{2}$. We define the corresponding bilinear form by polarization: $$\langle\nabla g,\nabla f \rangle_{\Omega,\mathrm{reg}}=\frac{1}{4}\langle\nabla (g+f),\nabla (g+f) \rangle_{\Omega,\mathrm{reg}}-\frac14\langle\nabla (g-f),\nabla (g-f) \rangle_{\Omega,\mathrm{reg}};$$ this will only involve regularizing terms for those boundary points where \emph{both} $f$ and $g$ have jumps. We then have 
\begin{multline}
\langle\nabla\xi,\nabla\xi\rangle_{\Omega,\mathrm{reg}}=
\frac{\gap^2}{4}\sum_{i,j=1}^{g}s_is_j\langle\nabla h_i,\nabla h_j\rangle_{\Omega}\\ +\frac{\gap^2}{4}\sum_{i=1}^{g}\sum_{j=1}^ks_i\hat{s}_j\langle\nabla h_i,\nabla \hat{h}_j\rangle_{\Omega}+\frac{\gap^2}{4}\sum_{i=1}^{k}\sum_{j=1}^k\hat{s}_i\hat{s}_j\langle\nabla \hat{h}_i,\nabla \hat{h}_j\rangle_{\Omega,\mathrm{reg}}\\
=:-4\pi(Q_\Omega(s)+B_\Omega(s,\hat{s})+\hat{Q}_\Omega(\hat{s})), 
\end{multline}
where $Q_\Omega,\hat{Q}_\Omega$ are quadratic forms, and $B_\Omega$ is a bilinear form, depending on the conformal class of $\Omega$, we are able to drop the subscript ``$\mathrm{reg}$'' in the first two sums due to the above observation; in fact, we also have $\langle\nabla \hat{h}_i,\nabla \hat{h}_j\rangle_{\Omega,\mathrm{reg}}=\langle\nabla \hat{h}_i,\nabla \hat{h}_j\rangle_{\Omega}$ for $i\neq j$. 

Crucially, $Q_\Omega$ is strictly negative definite, which in particular ensures that $$Z=\sum_{\xi}\exp(-\frac{1}{4\pi}\langle\nabla\xi,\nabla\xi\rangle_{\Omega,\mathrm{reg}})<\infty,$$ so that $\P(\xi)=Z^{-1}\exp(-\frac{1}{4\pi}\langle\nabla\xi,\nabla\xi\rangle_{\Omega,\mathrm{reg}})$ is a well defined probability measure. Moreover, we have the following lemma: 
\begin{lem}
The expectation $\E\exp\left(\gamma_1\xi(z_1)+\dots+\gamma_n\xi(z_n)\right)$
is finite for all $\gamma_1,\dots,\gamma_n\in \C$, and defines a real-analytic function of its variables $\gamma_1,\dots,\gamma_n$, $\z_1,\dots,z_n$. Moreover, the expectation can be exchanged with derivatives of all orders acting on these variables.
\end{lem}

\begin{proof}
Observe that by positivity of $-Q_\Omega$, we have $\P(\max\{|s_i|\}>M)\leq\exp(-cM^2)$ for some $c>0$. Furthermore, note that by the maximum principle, $\max_\Omega{|\xi|}\leq\frac{\alpha}{2}(\max\{|s_i|\}+\frac12)$. It follows that $\E\exp\left(\gamma_1\xi(z_1)+\dots+\gamma_n\xi(z_n)\right)$ is a locally uniformly convergent series of real-analytic functions with radii of convergence uniformly bounded from below, and hence it is real-analytic.

Moreover, since $\xi$ is harmonic, for each disk $D\subset \Omega$, we have $\max_{z\in D}|\pa^n \xi(z)|\leq C_{D,n} (\max\{|m_i|\}+\frac12),$ by expressing $\xi(z)$ by Poisson formula in a larger disc $D'\subset \Omega$ and differentiating. Therefore, dominated convergence theorem yields the second claim.
\end{proof}
In fact, the expectations $\E\exp\left(\gamma_1\xi(z_1)+\dots+\gamma_n\xi(z_n)\right)$, which feature prominently in what follows, can be written quite explicitly in terms of \emph{theta functions}, see Section \ref{sec:hejhal-Fay}. Indeed,

    \begin{multline}
       \E\exp\left(\gamma_1\xi(z_1)+\dots+\gamma_n\xi(z_n)\right)=\\
       =\frac{1}{Z}\sum_{s, \hat{s}}\exp\left(Q_{\Dom}(s)+B_\Dom (s,\hat{s})+\hat Q_\Dom(\hat{s})+
       \frac{\gap}{2}\sum_{i=1}^n \gamma_i\left(\sum_{j=1}^g s_j h_j(z_i)+\sum_{j=1}^k\hat{s}_{j}\hat h_j(z_i)\right)\right), 
    \end{multline}
where the partition function $Z$ is given by the same sum with all $\gamma_i=0.$ Summing first over $\hat{s}$ and putting $\nu=(\frac{N_1}{2},\dots,\frac{N_g}{2})$, so that $s=2(m+\nu)$, we obtain

\begin{multline}
       \E\exp\left(\gamma_1\xi(z_1)+\dots+\gamma_n\xi(z_n)\right)=\\
       =\frac{1}{Z}\sum_{ \hat{s}\in \{\pm 1\}^k}\hat{H}(\hat{s},\{\gamma_i\},\{z_i\})
       \sum_{m\in\Z^g}\exp\left(4Q_{\Dom}(m+\nu)+2(m+\nu)\cdot v \right)= \\
      =\frac{1}{Z}\sum_{ \hat{s}\in \{\pm 1\}^k}\hat{H}(\hat{s},\{\gamma_i\},\{z_i\})\cdot\theta_{Q_\Omega}(v|H),
    \end{multline}
where $\hat{H}(\hat{s},\{\gamma_i\},\{z_i\})=\exp\left(\hat Q_\Dom(\hat s)+\frac{\gap}{2}\sum_{i=1}^n\gamma_i\sum_{j=1}^k\hat{s}_{j}\hat h_j(z_i)\right)$, the characteristic (see Section \ref{sec:hejhal-Fay} for this terminology and notation) is given by $H=\binom{0}{\nu}$, and the vector $v=v(\hat{s},{\gamma_i},{z_i})$ is given by its coordinates as
\[
v_j=\frac{\gap}{2}\sum_{i=1}^n \gamma_ih_j(z_i)-\frac{\gap^2}{16\pi}\sum_{i=1}^k \hat{s}_{i}\langle\nabla h_j,\nabla \hat{h}_i\rangle_{\Omega}.
\]

We remark here that the matrix of $Q_\Omega$ used in the definition of the theta function is in fact the period matrix $\tau$ of the Schottky double $\widehat{\Omega}$ of $\Omega$, and that $h_i,\hat{h}_j$, and the coefficients of $B_\Omega$, $\hat{Q}_\Omega$ can be in fact written in terms of Abelian integrals on $\widehat{\Omega}$, see Lemma \ref{lem: prob_instanton}.

\subsection{The correlations of \texorpdfstring{$\Phi$}{}.}
\label{sec: bosonic}
We start by setting up some notation and normalization conventions. 
\begin{defn}
For $w\in \Omega$, the Green's function $G_\Omega(\cdot,w)$ is defined to be the unique harmonic function in $\Omega\setminus\{w\}$, vanishing at $\pa\Omega$, and satisfying the expansion
\[
G_\Omega(z,w)=-\log|z-w|+O(1),\quad z\to w.
\]
We denote $g_\Omega(z,w)=G_\Omega(z,w)+\log|z-w|$.
\end{defn}

We recall that $G_\Omega(z,w)\equiv G_\Omega(w,z)$, and hence $g_\Omega(z,w)\equiv g_\Omega(w,z)$. The Gaussian free field $\varphi$ is a random distribution in $\Omega$ which is a centered Gaussian field with covariance $G_\Omega$, see \cite{Sheffield}. We will actually not use any probabilistic aspects of the theory of GFF, only the structure of the related correlation functions, which we now describe. 
\begin{defn}
\label{def: bos_corr}
Given arbitrary complex numbers $\gamma_1,\dots,\gamma_n$ and distinct $z_1,\dots,z_n\in\Omega$, the correlations of the normal-ordered exponentials of $\varphi$ are defined as
\begin{equation}
\label{eq: defnexp}
\ccor{\nexp{\gamma_1}{z_1}\dots\nexp{\gamma_n}{z_n}}_\Omega=\exp\left(\sum_{i<j}\gamma_i\gamma_j G_\Omega(z_i,z_j)+\sum_{i=1}^n\frac{\gamma^2_i}{2}g_\Omega(z_i,z_i)\right)
\end{equation}
The similar correlations of $\Phi=\varphi+\xi$ are defined as 
\begin{equation}
\label{eq: defNexp}
\ccor{\Nexp{\gamma_1}{z_1}\dots\Nexp{\gamma_n}{z_n}}_\Omega=\ccor{\nexp{\gamma_1}{z_1}\dots\nexp{\gamma_n}{z_n}}_\Omega\E(e^{\gamma_1\xi(z_1)}\dots e^{\gamma_n\xi(z_n)})
\end{equation}
Finally, for each choice of $\hOp_{z_i}\in \{\Nexp{\gamma_i}{z_i}, \pa\Phi(z_i), \bar{\pa}\Phi(z_i), \nord{|\nabla \Phi(z_i)|^2}\}$, we define 
\begin{equation}
\label{eq: defOp}
\ccor{\hOp_{z_1}\dots\hOp_{z_n}}=\mathcal{D}_{z_1}\dots\mathcal{D}_{z_n}\ccor{\Nexp{\gamma_1}{z_1}\dots\Nexp{\gamma_n}{z_n}},
\end{equation}
where the linear operators $\mathcal{D}_{z_i}$ act by combination of differentiation and substitution of $\gamma_i=0$, as follows:
\begin{align}
\mathcal{D}_{z_i}=\mathrm{Id}&\text{ if } \hOp_{z_i}=:e^{\gamma_i \Phi(z_i)}:, \\
\mathcal{D}_{z_i}=\pa_{\gamma_i}\pa_{z_i}|_{\gamma_i=0}&\text{ if } \hOp_{z_i}=\pa\Phi(z_i),\\
\mathcal{D}_{z_i}=\pa_{\gamma_i}\bar{\pa}_{z_i}|_{\gamma_i=0}&\text{ if } \hOp_{z_i}=\bar{\pa}\Phi(z_i),\\
\mathcal{D}_{z_i}=2\pa^2_{\gamma_i}\pa_{z_i}\bar{\pa}_{z_i}|_{\gamma_i=0}&\text{ if } \hOp_{z_i}=\nord{|\nabla \Phi(z_i)|^2}.
\end{align} 
Finally, the correlations involving $\nord{\sin(\gamma\Phi)}$ or $\nord{\cos(\gamma\Phi)}$ are defined using correlations of exponentials and linearity, e.g. 
\[
\ccor{\nord{\sin(\gamma\Phi(z_1))}\hOp}=\frac{1}{2\i}\left(\ccor{\Nexp{\i\gamma}{z_1}\hOp}-\ccor{\Nexp{-\i\gamma}{z_1}\hOp}\right)
\]
for any string $\hOp=\hOp_{z_2}\dots\hOp_{z_n}$ as above.
\end{defn}

We record the following simple, but important observation:
\begin{lem}
The function $\ccor{\Nexp{\gamma_1}{z_1}\dots\Nexp{\gamma_n}{z_n}}_\Omega$ is a real-analytic function of its $2n$ variables as long as $z_i$ are distinct, i.e., it is real-analytic in the region 
\begin{equation}
\label{eq: distinct}
\{\gamma_1,\dots,\gamma_n\in\C,\quad z_1,\dots,z_n\in\Omega:z_i\neq z_j \text{ if } i\neq j\}.
\end{equation}
Moreover, the function 
\begin{equation}
\label{eq: def_F}
F(\gamma_1,\dots\gamma_n,z_1,\dots,z_n):=|z_1-z_2|^{\gamma_1\gamma_2}\ccor{\Nexp{\gamma_1}{z_1}\dots\Nexp{\gamma_n}{z_n}}
\end{equation} extends real-analytically to $z_1=z_2$, i.e., to the set 
\[
\{\gamma_1,\dots,\gamma_n\in\C,\quad z_1,\dots,z_n\in\Omega:z_i\neq z_j \text{ if } i\neq j\text{ and }\{i,j\}\neq\{1,2\}\}.
\] 
\end{lem}
\begin{proof}
The real analyticity of $\E(\exp(\gamma_1\xi(z_1)+\dots+\gamma_n\xi(z_n)))$ is already noted above. The real analyticity of $\ccor{\nexp{\gamma_1}{z_1}\dots\nexp{\gamma_n}{z_n}}_\Omega$ follows from the harmonicity, and hence real analyticity, of $G_\Omega$ and $g_\Omega$ in both variables, the latter one including on the diagonal. The only singular term at $z_1=z_2$ comes from $\exp(\gamma_1\gamma_2G_\Omega(z_1,z_2))$, however, we have $|z_1-z_2|^{\gamma_1\gamma_2}\exp(\gamma_1\gamma_2G_\Omega(z_1,z_2))=\exp(\gamma_1\gamma_2g_{\Omega}(z_1,z_2))$, which is real analytic.
\end{proof}
\begin{rem}
\label{rem: analyticity}
In view of the above lemma, we can view the operators $\mathcal{D}_{z_i}$ as simply picking out some of the coefficients in the Taylor expansion of $\ccor{\Nexp{\gamma_1}{z_1}\dots\Nexp{\gamma_n}{z_n}}_\Omega$. In particular, the general correlations $\ccor{\hOp_{z_1}\dots\hOp_{z_n}}$ are also real analytic over \eqref{eq: distinct}.
\end{rem}

Although we will not need it, Definition \ref{def: bos_corr} can be motivated by the following regularization procedure. This is the only part of the article where we view $\Phi$ as truly being a random generalized function, and in particular, that its mollification produces a random smooth function. Let $\rho:\C\to\R_{\geq0}$ be a smooth, rotationally symmetric, compactly
supported function with $\int_{\C}\rho=1.$ For $\eps>0,$ we put
$\rho_{\eps}(z)=\eps^{-2}\rho(\frac{z}{\eps})$ and define the regularized bosonic field to be $\Phi_{\eps}=\Phi\star\rho_{\eps}=\varphi_{\eps}+\xi_{\eps},$
where $\xi_{\eps}=(\xi\star\rho_{\eps})$ and $\varphi_{\eps}=\varphi\star\rho_\eps$. Moreover, we define the regularized counterparts of the fields $\hOp_{z_i}$ by
\[
\hOp^\eps_z:=\begin{cases}\pa\Phi_\eps(z),&\hOp_z=\pa\Phi(z)\\ 
\bar{\pa}\Phi_\eps(z),&\hOp_z=\bar{\pa}\Phi(z)\\
\left(c_\rho\eps\right)^\frac{\gamma^2}{2}e^{\gamma\Phi_\eps(z)},&\hOp_z=\Nexp{\gamma}{z}\\
|\nabla\Phi_\eps(z)|^2-\eps^{-2}d_\rho&\hOp_{z}=\nord{|\nabla\Phi(z)|^2}\\
\end{cases}
\] 
where 
\[
c_{\rho}=\exp\left(\int_{\C^{2}}\rho(w_{1})\rho(w_{2})\log(|w_{1}-w_{2}|)dw_{1}dw_{2}\right);\quad d_{\rho}=2\pi\int_{\C^{2}}\rho^{2}.
\]
Note that on one hand, $\xi_\eps(z)=(\xi\star\rho_\eps)(z)=\xi(z)$ for $\eps$ small enough (depending on $z$), since $\xi$ is almost surely harmonic, on the other hand, $\varphi_{\eps}$ is a smooth random Gaussian field. Hence $\hOp_z^\eps$ are actually \emph{random variables} defined on the same probability space (they are measurable functions of $\Phi$). 
\begin{lem}
As $\eps\to 0,$ we have for any distinct $z_1,\dots,z_n\in\Omega$,
\[
\E(\hOp^\eps_{z_1}\dots\hOp^\eps_{z_n})\to\ccor{\hOp_{z_1}\dots\hOp_{z_n}}.
\]
\end{lem}

\begin{proof}

Observe that $\varphi_{\eps}$ is a centered
Gaussian field with covariance 
\begin{equation}
G_{\Omega,\eps}(z_{1},z_{2})=\int_{\C^{2}}\rho_{\eps}(w_{1}-z_{1})\rho_{\eps}(w_{2}-z_{2})G_{\Omega}(w_{1},w_{2})dw_{1}dw_{2}\label{eq: Green_eps}
\end{equation}
 satisfying the following two properties, as easily seen using harmonicity of $G_\Omega$:
\begin{itemize}
\item For any $z_{1}\neq z_{2}$, one has $G_{\Omega,\eps}(z_{1},z_{2})=G_{\Omega}(z_{1},z_{2})$
for $\eps$ small enough;
\item For a given $z\in\Omega$, one has $G_{\Omega,\eps}(z,z)=g_{\Omega}(z,z)-\log(c_\rho\eps)$
for $\eps$ small enough.
\end{itemize}
This immediately implies the statement
of the lemma in the case $\hat{\Op}_{z_{i}}=:\exp(\gamma_{i}z_{i}):$
for all $i$, since for $\eps$ small enough (depending on $z_{i}$), writing $:\exp \gamma \Phi_\eps(z):=(c_\rho \eps)^{\frac{\gamma^2}{2}}e^{\gamma \Phi_\eps(z)}$,
we have 
\begin{multline*}
\ccor{:\exp\gamma_{1}\Phi_{\eps}(z_{1}):\dots:\exp\gamma_{n}\Phi_{\eps}(z_{n}):}\\
=(c_\rho\eps)^{\frac{\gamma_{1}^{2}+\dots+\gamma_{n}^{2}}{2}}\E\exp(\gamma_{1}\varphi_{\eps}(z_{1})+...+\gamma_{n}\varphi_{\eps}(z_{n}))\E\exp(\gamma_{1}\xi_{\eps}(z_{1})+...+\gamma_{n}\xi_{\eps}(z_{n}))\\
=(c_\rho\eps)^{\frac{\gamma_{1}^{2}+\dots+\gamma_{n}^{2}}{2}}\exp\left(\frac{1}{2}\sum_{i,j}\gamma_{i}\gamma_{j}G_{\Omega,\eps}(z_{i},z_{j})\right)\E\exp(\gamma_{1}\xi(z_{1})+...+\gamma_{n}\xi(z_{n}))\\
=\exp\left(\sum_{i<j}\gamma_{i}\gamma_{j}G_{\Omega}(z_{i},z_{j})+\sum_{i}\frac{\gamma_{i}^{2}}{2}g_{\Omega}(z_{i},z_{i})\right)\E\exp(\gamma_{1}\xi(z_{1})+...+\gamma_{n}\xi(z_{n})).
\end{multline*}
Therefore, to prove the lemma, it suffices to show that (in the notation of Definition \ref{def: bos_corr})
\begin{equation}
\ccor{\hat{\Op}_{z_{1}}^{\eps}\dots\hat{\Op}_{z_{n}}^{\eps}}=\mathcal{D}_{z_{1}}\dots\mathcal{D}_{z_{n}}\ccor{:\exp\gamma_{1}\Phi_{\eps}(z_{1}):\dots:\exp\gamma_{n}\Phi_{\eps}(z_{n}):}\label{eq: Op_eps_via_L}
\end{equation}
This is immediate if for each $i,$ either $\hat{\Op}_{z_{i}}=\pa\Phi(z_{i})$, $\hat{\Op}_{z_{i}}=\bar{\pa}\Phi(z_{i})$ or $\hat{\Op}_{z_{i}}=:\exp(\gamma_{i}\Phi(z_{i})):$,
taking into account that the correlations $\E\exp(\gamma_{1}\varphi_{\eps}(z_{1})+...+\gamma_{n}\varphi_{\eps}(z_{n}))$
can be readily differentiated under the expectation sign, e.g., by a standard application of Borell--TIS inequality,
see \cite[Chapter 2]{Adler}. 

Thus, it remains to handle $:|\nabla\Phi(z_{i})|^{2}:$. We start
with the computation for the case $\hat{\Op}_{z_{1}}=:|\nabla\Phi(z_{1})|^{2}:$
and $\hat{\Op}_{z_{i}}=:\exp(\gamma_{i}\Phi(z_{i})):$ for $i\geq 2$ and ignore
the instanton component for the moment. Recalling that $\mathcal{D}_{z_{1}}=2\pa_{\gamma_{1}}^{2}\pa_{z_{1}}\pa_{\bar z_{1}}|_{\gamma_{1}=0}$,
a straightforward computation yields (again for small enough $\eps$)
\begin{multline*}
\mathcal{D}_{z_{1}}\ccor{:\exp\gamma_{1}\varphi_{\eps}(z_{1}):\dots:\exp\gamma_{n}\varphi_{\eps}(z_{n}):}=\mathcal{D}_{z_{1}}\ccor{:\exp\gamma_{1}\varphi(z_{1}):\dots:\exp\gamma_{n}\varphi(z_{n}):}\\
=\left(4\left(\sum_{i=2}^{n}\gamma_{i}\pa_{z_{1}}G(z_{1},z_{i})\right)\left(\sum_{i=2}^{n}\gamma_{i}\pa_{\bar z_{1}}G(z_{1},z_{i})\right)+2\pa_{z_{1}}\pa_{\bar z_{1}}g_{\Omega}(z_{1},z_{1})\right)\\
\times\ccor{:\exp\gamma_{2}\varphi(z_{2}):\dots:\exp\gamma_{n}\varphi(z_{n}):}
\end{multline*}

For $z_{0},z_{1},z_{2},\dots,z_{n}$ distinct,
we can calculate 
\begin{multline*}
\ccor{\pa_{z_0}\varphi_{\eps}(z_{0})\pa_{\bar z_{1}}\varphi_{\eps}(z_{1}):\exp\gamma_{2}\varphi_{\eps}(z_{2}):\dots:\exp\gamma_{n}\varphi_{\eps}(z_{n}):}\\
=\left.\pa_{\gamma_{0}}\pa_{\gamma_{1}}\pa_{z_{0}}\pa_{z_{1}}\right|_{\gamma_{0}=\gamma_{1}=0}\ccor{:\exp\gamma_{0}\varphi_{\eps}(z_{0}):\dots:\exp\gamma_{n}\varphi_{\eps}(z_{n}):}\\
=\left(\left(\sum_{i\neq0,1}\gamma_{i}\pa_{z_{0}}G_{\Omega,\eps}(z_{0},z_{i})\right)\left(\sum_{i\neq0,1}\gamma_{i}\pa_{\bar z_{1}}G_{\Omega,\eps}(z_{1},z_{i})\right)+\pa_{z_{0}}\pa_{\bar z_{1}}G_{\Omega,\eps}(z_{0},z_{1})\right)\\
\times\ccor{:\exp\gamma_{2}\varphi_{\eps}(z_{2}):\dots:\exp\gamma_{n}\varphi_{\eps}(z_{n}):}.
\end{multline*}
If we now put $z_{0}=z_{1},$ the left-hand side becomes 
\[
\frac{1}{4}\ccor{|\nabla\varphi_{\eps}(z_{1})|^{2}:\exp\gamma_{2}\varphi_{\eps}(z_{2}):\dots:\exp\gamma_{n}\varphi_{\eps}(z_{n}):}.
\]
For the right-hand side, compute, using (\ref{eq: Green_eps}) and
integrating by parts: 
\begin{multline*}
\pa_{z_{0}}\pa_{\bar z_{1}}G_{\Omega,\eps}(z_{0},z_{1})=\int_{\C^{2}}\pa_{z_{0}}\rho_{\eps}(w_{0}-z_{0})\pa_{\bar z_{1}}\rho_{\eps}(w_{1}-z_{1})G_{\Omega}(w_{0},w_{1})dw_{0}dw_{1}\\
=\int_{\C^{2}}\rho_{\eps}(w_{0}-z_{0})\rho_{\eps}(w_{1}-z_{1})\pa_{w_{0}}\pa_{\bar w_{1}}G_{\Omega}(w_{0},w_{1})dw_{0}dw_{1}\\
=\int_{\C^{2}}\rho_{\eps}(w_{0}-z_{0})\rho_{\eps}(w_{1}-z_{1})\left(\pa_{\bar w_{1}}\frac{1}{2(w_{1}-w_{0})}+\pa_{w_{0}}\pa_{\bar w_{1}}g_{\Omega}(w_{0},w_{1})\right)dw_{0}dw_{1}\\
=\pa_{z_{0}}\pa_{\bar z_{1}}g_{\Omega}(z_{0},z_{1})+\frac{\pi}{2}\int_{\C^{2}}\rho_{\eps}(w_{0}-z_{0})\rho_{\eps}(w_{0}-z_{1})dw_{0},
\end{multline*}
where in the last equality we have used $\pa_{\bar w_{1}}\frac{1}{(w_{1}-w_{0})}=\pi\delta(w_{1}-w_{0})$
and the fact that all the derivatives of $g$ are harmonic. Plugging
in $z_{0}=z_{1}$ and noticing that $\pa_{z_{0}}\pa_{\bar z_{1}}g_{\Omega}(z_{0},z_{1})|_{z_{0}=z_{1}}=\frac{1}{2}\pa_{z_{1}}\pa_{\bar z_{1}}g_{\Omega}(z_{1},z_{1}),$
we conclude that 
\begin{multline*}
\mathcal{D}_{z_{1}}\ccor{:\exp\gamma_{1}\varphi_{\eps}(z_{1}):\dots:\exp\gamma_{n}\varphi_{\eps}(z_{n}):}\\
=\ccor{\left(|\nabla\varphi_{\eps}(z_{1})|^{2}-\eps^{-2}d_{\rho}\right):\exp\gamma_{2}\varphi_{\eps}(z_{2}):\dots:\exp\gamma_{n}\varphi_{\eps}(z_{n}):}.
\end{multline*}
To extend this to an arbitrary string $\hat{\Op}_{z_{2}},\dots,\hat{\Op}_{z_{n}}$
not including $\hat{\Op}_{z_{i}}=:|\nabla\varphi_{\eps}(z_{i})|^{2}:$,
simply apply $\mathcal{D}_{z_{2}}\dots\mathcal{D}_{z_{n}}$ to the
above identity. To allow for $\hat{\Op}_{z_{i}}=:|\nabla\varphi_{\eps}(z_{i})|^{2}:,$
notice that, renaming $z_{0}$ into $z_{1}'$, the above computation
can be re-interpreted as stating that 
\[
\mathcal{\tilde{D}}_{z_{1}}\ccor{:\exp\gamma'_{1}\varphi_{\eps}(z_{1}')::\exp\gamma_{1}\varphi_{\eps}(z_{1}):\dots:\exp\gamma_{n}\varphi_{\eps}(z_{n}):}\equiv0,
\]
where $\mathcal{\tilde{D}}_{z_{1}}$ is a linear map taking as
input a real-analytic function of $z_{1}',z_{1}\dots,z_{n},$ $\gamma_{1}',\gamma_{1}\dots,\gamma_{n}$
and outputting a real-analytic function of $z_{1},\dots,z_{n},\gamma_{2},\dots,\gamma_{n}$;
$\mathcal{\tilde{D}}_{z_{1}}$ is constructed out of derivatives with
respect to $\gamma_{0},\gamma_{1},z_{1}',z_{1}$, multiplication by
functions of $\gamma_{0},\gamma_{1}$, and specializing $\gamma_{1}'$
or $\gamma_{1}$ to zero and $z_{0}$ to $z_{1}$. Crucially, if we
now similarly ``double'' each of the points $z_{i}$ such that $\hat{\Op}_{z_{i}}=:|\nabla\varphi_{\eps}(z_{i})|^{2}:$,
then each of the operators comprising $\tilde{\mathcal{D}}_{z_{i}}$
commutes with ones comprising $\mathcal{\tilde{D}}_{z_{j}}$ for $i\neq j$.
After unwinding the algebra, this observation implies (\ref{eq: Op_eps_via_L})
whenever, for each $i$, $\hat{\Op}_{z_{i}}=:|\nabla\varphi(z_{i})|^{2}:$
and $\hat{\Op}_{z_{i}}=:\exp(\gamma_{i}\varphi(z_{i})):.$ The fully
general case readily follows
as above.

To include the instanton component, note that $\mathcal{D}_{z_{1}}\ccor{:\exp\gamma_{1}\Phi_{\eps}(z_{1}):\dots:\exp\gamma_{n}\Phi_{\eps}(z_{n}):}$
will include the term $\left(\mathcal{D}_{z_{1}}\ccor{:\exp\gamma_{1}\varphi_{\eps}(z_{1}):\dots:\exp\gamma_{n}\varphi_{\eps}(z_{n}):}\right)\E e^{\gamma_{1}\xi(z_{1})+\dots+\gamma_{n}\xi(z_{n})},$
which is handled above, and the terms containing at most first-order
derivatives acting on the GFF correlations, which are readily seen
to be equal to their counterparts in the left-hand side of \eqref{eq: Op_eps_via_L},
with no ``renormalization'' needed. We leave further details to the
reader.
\end{proof}

\subsection{Operator product expansions for the bosonic correlations}

In this subsection, we derive the asymptotics of the correlations $\ccor{\hOp_{z_1}\dots\hOp_{z_n}}$ as $z_1\to z_2$. We start with the case of exponentials, on which all the rest will be based.
\begin{lem}
\label{lem: fuse_exps}
For any fixed distinct $z_2,\dots,z_n\in\Omega$, any $\Op=\hOp_{z_3}\dots\hOp_{z_n}$, we have the following expansion as $z_1\to z_2$, as long as $\gamma_1\neq-\gamma_2$:
\begin{multline}
\label{eq: fuse_exps}
\ccor{\Nexp{\gamma_1}{z_1}\Nexp{\gamma_2}{z_2}\Op}=|z_1-z_2|^{-\gamma_1\gamma_2}\bigg(\ccor{\Nexp{(\gamma_1+\gamma_2)}{z_2}\Op}\\
+\frac{\gamma_1}{\gamma_1+\gamma_2}(z_1-z_2)\pa_{z_2}\ccor{\Nexp{(\gamma_1+\gamma_2)}{z_2}\Op}\\
+\frac{\gamma_1}{\gamma_1+\gamma_2}(\bar{z}_1-\bar{z}_2)\bar{\pa}_{z_2}\ccor{\Nexp{(\gamma_1+\gamma_2)}{z_2}\Op}
+O(z_1-z_2)^2\bigg).
\end{multline}
If $\gamma_1=-\gamma_2$, the expansion reads 
\begin{multline}
\label{eq: fuse_exp_bis}
\ccor{\Nexp{\gamma_1}{z_1}\Nexp{\gamma_2}{z_2}\Op}=|z_1-z_2|^{-\gamma_1\gamma_2}\\
\times\left(\ccor{\Op}+\gamma_1(z_1-z_2)\ccor{\pa\Phi(z_2)\Op}+\gamma_1(\bar{z}_1-\bar{z}_2)\ccor{\bar{\pa}\Phi(z_2)\Op}
+O(z_1-z_2)^2\right).
\end{multline}
Moreover, all further Taylor coefficients of $|z_1-z_2|^{\gamma_1\gamma_2}\ccor{\Nexp{\gamma_1}{z_1}\Nexp{\gamma_2}{z_2}\Op}$ at $z_1=z_2$ vanish at $\gamma_1=0$, while the mixed second derivative 
\[
\bar{\pa}_{z_1}\pa_{z_1}\left(|z_1-z_2|^{\gamma_1\gamma_2}\ccor{\Nexp{\gamma_1}{z_1}\Nexp{\gamma_2}{z_2}\Op}\right)|_{z_1=z_2}
\]
vanishes to the second order at $\gamma_1=0$.
\end{lem}

\begin{proof}
We first check \eqref{eq: fuse_exps} in the particular case $\hOp_{z_i}=\Nexp{\gamma_i}{z_i}$ for $i=3,\dots,n$, when it is a statement about the beginning of a Taylor expansion of the real-analytic function $F(z_1)$ given by \eqref{eq: def_F} at $z_1=z_2$; we view $\gamma_1,\dots,\gamma_n,z_2,\dots,z_n$ as fixed parameters. We can write by definition
\begin{equation}
F(z_1)=F_1(z_1)\cdot F_2(z_1)
\end{equation}
where 
\[
F_1(z_1)=C\exp\left(\gamma_{1}\gamma_2 g_\Omega(z_1,z_2)+\sum_{i=3}^n \gamma_1\gamma_iG_\Omega(z_1,z_i) +\frac{\gamma_1^2}{2}g_\Omega(z_1,z_1)\right),
\]
\[
F_2(z_1)=\E (\exp(\gamma_1\xi(z_1)+\gamma_2\xi(z_2)+\dots\gamma_n\xi(z_n))),
\]
and $C=\exp\left(\sum_{2\leq i<j\leq n}\gamma_i\gamma_j G_\Omega(z_i,z_j)+\sum_{i=2}^n\frac{\gamma^2_i}{2}g_\Omega(z_i,z_i)\right)$ does not depend on $z_1$. Denote also 
\begin{align*}
H_1(z_2,\dots,z_n)&=\ccor{\nexp{(\gamma_1+\gamma_2)}{z_2}\nexp{\gamma_3}{z_3}\dots\nexp{\gamma_n}{z_n}},\\
H_2(z_2,\dots,z_n)&=\E (\exp((\gamma_1+\gamma_2)\xi(z_2)+\gamma_3\xi(z_3)+\dots\gamma_n\xi(z_n))).
\end{align*}
By putting $z_1=z_2$ and comparing with \eqref{eq: defNexp}, one readily checks that $F_1(z_2)=H_1$ and $F_2(z_2)=H_2$, that is, that 
\[
F(z_2)=H_1H_2=\ccor{\Nexp{(\gamma_1+\gamma_2)}{z_2}\Nexp{\gamma_3}{z_3}\dots\Nexp{\gamma_n}{z_n}}.
\]
Moreover, we have 
\begin{multline}
\label{eq: pa_z_1}
\pa_{z_1}F(z_1)=\left(\gamma_1\gamma_2\pa_{z_1}g_\Omega(z_1,z_2)+\frac{\gamma^2_1}{2}\pa_{z_1}g_\Omega(z_1,z_1)+\sum_{i=3}^n\gamma_1\gamma_i\pa_{z_1}G_\Omega(z_1,z_i)\right) F(z_1)\\
+F_1(z_1)\pa_{z_1}F_2(z_1),
\end{multline}
while by definition 
\begin{multline}
\label{eq: pa_z_1_bis}
\pa_{z_2}(H_1H_2)=\left(\frac{(\gamma_1+\gamma_2)^2}{2}\pa_{z_2}g_\Omega(z_2,z_2)+\sum_{i=3}^n(\gamma_1+\gamma_2)\gamma_i\pa_{z_2}G_\Omega(z_2,z_i)\right)H_1H_2\\+
H_1\pa_{z_2}H_2.
\end{multline}
Note that $g_\Omega(z_1,z_2)$ is a real-analytic function which is symmetric in $z_1$ and $z_2$, therefore $\pa_{z_2}g_\Omega(z_2,z_2)=2\pa_{z_1}g_\Omega(z_1,z_2)|_{z_1=z_2}$. Also, we have \begin{align}
\pa_{z_1}F_2(z_1)|_{z_1=z_2}&=\gamma_1\E(\pa \xi(z_2)e^{(\gamma_1+\gamma_2)\xi(z_2)+\dots+\gamma_n\xi(z_n)}),\\
\pa_{z_2}H_2&=(\gamma_1+\gamma_2)\E(\pa\xi(z_2)e^{(\gamma_1+\gamma_2)\xi(z_2)+\dots+\gamma_n\xi(z_n)}).
\end{align}
Putting these observations together and comparing \eqref{eq: pa_z_1} with \eqref{eq: pa_z_1_bis}, we conclude 
\[
\pa_{z_1}F(z_1)|_{z_1=z_2}=\frac{\gamma_1}{\gamma_1+\gamma_2}\pa_{z_2}\ccor{\Nexp{(\gamma_1+\gamma_2)}{z_2}\Nexp{\gamma_3}{z_3}\dots\Nexp{\gamma_n}{z_n}}.
\]
Using that $\overline{\pa_{z_1}F(z_1)}=\bar{\pa}_{z_1}\overline{F(z_1)}$, we see that the same identity holds with $\bar{\pa}_{{z}_1}$ and $\bar{\pa}_{z_2}$ instead of $\pa_{z_1}$ and $\pa_{z_2}$ respectively. This proves \eqref{eq: fuse_exps} in the case $\hOp_{z_i}=\Nexp{\gamma_i}{z_i}$ for $i=3,\dots,n$. The general case follows immediately from the observation that $\mathcal{D}_{z_3},\dots,\mathcal{D}_{z_n}$ commute both with multiplication by $|z_1-z_2|^{\gamma_1\gamma_2}$ and taking derivatives with respect to $z_1$. For $\hOp_{z_i}=\nord{\cos(\gamma\Phi(z_i))}$ or $\hOp_{z_i}=\nord{\sin(\gamma\Phi(z_i))}$, $i=3,\dots,n$, we simply use linearity.

The expansion \eqref{eq: fuse_exp_bis} follows from \eqref{eq: fuse_exps} by taking the limit $\gamma_1\to -\gamma_2$ and recalling the definition of correlations involving $\pa \Phi$ and $\bar{\pa}\Phi$; again this is fully justified since $F$ (and $\mathcal{D}_{z_3}\dots\mathcal{D}_{z_n}F$) is a real-analytic function in all its parameters, particularly its Taylor coefficients with respect to $z_1$ are also real-analytic functions of $\gamma_1$. 

For the final claim, since we are talking about convergent Taylor series, we can first plug $\gamma_1=0$ and then evaluate the derivatives with respect to $z_1$; they vanish identically since the function does not depend on $z_1$ when $\gamma_1=0$. The stronger assertion about the mixed second derivative requires more work. We go back to \eqref{eq: pa_z_1} and apply $\bar{\pa}_{z_1}$ to both sides, taking into account that both $G_\Omega(z_1,z_2)$ and $g_\Omega(z_1,z_2)$ are harmonic and therefore annihilated by $\bar{\pa}_{z_1}\pa_{z_1}$. This gives
\begin{multline*}
\bar{\pa}_{z_1}\pa_{z_1}F(z_1)=\frac{\gamma_1^2}{2}\bar{\pa}_{z_1}\pa_{z_1}g_{\Omega}(z_1,z_1)F_1(z_1)F_2(z_1)+\frac{\pa_{z_1}F_1(z_1)}{F_1(z_1)}(\bar{\pa}_{z_1}F_1(z_1))F_2(z_1)\\
+\pa_{z_1}F_1(z_1)\bar{\pa}_{z_1}F_2(z_1)+\bar{\pa}_{z_1}F_1(z_1)\pa_{z_1}F_2(z_1)+F_1(z_1)\bar{\pa}_{z_1}\pa_{z_1}F_2(z_1).
\end{multline*}
As we have seen above, each of the first derivatives in the above expression (at $z_1=z_2$) vanish to the first order at $\gamma_1=0$, and the same is true of the $\frac{\pa_{z_1}F_1}{F_1}$ (which is of course just the prefactor in \eqref{eq: pa_z_1}). It remains to treat $\bar{\pa}_{z_1}\pa_{z_1}F_2|_{z_1=z_2}$. We have 
\begin{multline*}
\bar{\pa}_{z_1}\pa_{z_1}F_2=\bar{\pa}_{z_1}\pa_{z_1}\E (\exp(\gamma_1\xi(z_1)+\dots\gamma_n\xi(z_n)))\\
=\gamma_1^2\E (\pa_{z_1}\xi(z_1)\bar{\pa}_{z_1}\xi(z_1)\exp(\gamma_1\xi(z_1)+\dots\gamma_n\xi(z_n)))\\+\gamma_1\E (\pa_{z_1}\bar{\pa}_{z_1}\xi(z_1)\exp(\gamma_1\xi(z_1)+\dots\gamma_n\xi(z_n))),
\end{multline*}
but the last term vanishes because $\xi$ is almost surely harmonic. This concludes the proof that $\bar{\pa}_{z_1}\pa_{z_1}F=\gamma_1^2R(z_1),$ where $R(z_1)$ is real-analytic at $z_1=z_2$ and $\gamma_1=0$, i.e., we have proven the required assertion in the case  $\Op=\Nexp{\gamma_3}{z_3}\dots\Nexp{\gamma_n}{z_n}$. It is then extended to the general case exactly as above.
\end{proof}

\begin{rem}
For a reader familiar with Conformal Field Theory, let us mention that in general, the $n$-th order terms in the expansion \eqref{eq: fuse_exps} can be expressed in terms of Virasoro descendants of $\Nexp{(\gamma_1+\gamma_2)}{z_2}$. In particular, the lemma above handles the case of the level 2 descendant $\bar{L}_{-1}L_{-1}\Nexp{(\gamma_1+\gamma_2)}{z_2}$; the reader can check that explicitly, 
\[
\bar{\pa}_{z_1}\pa_{z_1}F(z_1)|_{z_1=z_2}=\frac{\gamma_1^2}{(\gamma_1+\gamma_2)^2}\bar{\pa}_{z_2}\pa_{z_2}\ccor{\Nexp{(\gamma_1+\gamma_2)}{z_2}\Nexp{\gamma_3}{z_3}\dots\Nexp{\gamma_n}{z_n}}.
\] 
\end{rem}

The identities \eqref{eq: fuse_exps}--\eqref{eq: fuse_exp_bis} are examples of what is known as \emph{Operator Product Expansions} (OPE). It will be convenient to abbreviate the expansions of this type by suppressing $\ccor{\,\cdot\,\hOp}$, e.g. with this convention \eqref{eq: fuse_exps} can be written as
\begin{multline}
\label{eq: fuse_exps_short}
\Nexp{\gamma_1}{z_1}\Nexp{\gamma_2}{z_2}=|z_1-z_2|^{-\gamma_1\gamma_2}\bigg(\Nexp{(\gamma_1+\gamma_2)}{z_2}\\
+\frac{\gamma_1}{\gamma_1+\gamma_2}(z_1-z_2)\pa_{z_2}\Nexp{(\gamma_1+\gamma_2)}{z_2}\\
+\frac{\gamma_1}{\gamma_1+\gamma_2}(\bar{z}_1-\bar{z}_2)\bar{\pa}_{z_2}\Nexp{(\gamma_1+\gamma_2)}{z_2}
+O(z_1-z_2)^2\bigg).
\end{multline}
The meaning of such identity is that they become valid expansions if we take a correlation of both sides with any string $\hOp=\hOp_{z_3}\dots\hOp_{z_n}$, with $z_2,\dots,z_n$ distinct, and use linearity to make sense of the resulting expression (in particular, taking derivatives outside the correlations). 

\begin{lem}
\label{lem: fuse_pa_exp}
We have the following expansions as $z_1\to z_2$:
\begin{equation}
\pa\Phi(z_1)\Nexp{\gamma_2}{z_2}=-\frac{\gamma_2}{2(z_1-z_2)}\Nexp{\gamma_2}{z_2}+\frac{1}{\gamma_2}\pa_{z_2}\Nexp{\gamma_2}{z_2}+O(z_1-z_2)\label{eq: fuse_pa_exp}
\end{equation}
\begin{equation}
\label{eq: fuse_pa_pa_bar}
\pa\Phi(z_1)\bar{\pa}\Phi(z_2)=\frac14\nord{|\nabla\Phi(z_2)|^2}+O(z_1-z_2).
\end{equation}
\end{lem}
\begin{proof}
The expansion \eqref{eq: fuse_pa_exp} is obtained by applying $\mathcal{D}_{z_1}=\pa_{\gamma_1}\pa_{z_1}|_{\gamma_1=0}$ to \eqref{eq: fuse_exps}; since the right-hand side (after the prefactor) is a convergent power series, it is clear that we can just differentiate term by term. We have, taking into account the ``moreover" claim of Lemma \ref{lem: fuse_exps},  
\begin{multline}
\label{eq: for_fusions}
\pa_{z_1}\Nexp{\gamma_1}{z_1}\Nexp{\gamma_2}{z_2}=\\-\frac{\gamma_1\gamma_2}{2(z_1-z_2)}|z_1-z_2|^{-\gamma_1\gamma_2}\left(\Nexp{(\gamma_1+\gamma_2)}{z_2}+\gamma_1\cdot(\dots)\right)\\
+|z_1-z_2|^{-\gamma_1\gamma_2}\left(\frac{\gamma_1}{\gamma_1+\gamma_2}\pa_{z_2}\Nexp{(\gamma_1+\gamma_2)}{z_2}+O(z_1-z_2)\right).
\end{multline}
Taking the derivative with respect to $\gamma_1$ at $\gamma_1=0$ (which will kill the term abbreviated as $(\dots)$ in the second line above) yields \eqref{eq: fuse_pa_exp}. In fact, we have a bit more information: inside a correlation with $\hOp$, the $O(z_1-z_2)$ error term in \eqref{eq: for_fusions} is given by  
\[
2(z_1-z_2)\pa_{z_1}^2\hat{F}(z_1)|_{z_1=z_2}+2(\bar{z}_1-\bar{z_2})\pa_{z_1}\bar{\pa}_{z_1}\hat{F}(z_1)|_{z_1=z_2}+O(z_1-z_2)^2,
\]
where $\hat{F}(z_1)=|z_1-z_2|^{\gamma_1\gamma_2}\ccor{\Nexp{\gamma_1}{z_1}\Nexp{\gamma_2}{z_2}\hOp}$. Lemma \ref{lem: fuse_exps} guarantees that when taking the derivative with respect to $\gamma_1$ at $\gamma_1=0$ this becomes 
\[
2(z_1-z_2)\pa_{\gamma_1}\pa_{z_1}^2\hat{F}(z_1)|_{z_1=z_2,\gamma_1=0}+O(z_1-z_2)^2,
\]
which is a refined form of the \eqref{eq: fuse_pa_exp} error term. Crucially, there is no $(\bar{z}_1-\bar{z}_2)$ part. 

To prove \eqref{eq: fuse_pa_pa_bar} we apply $\mathcal{D}_{z_2}=\pa_{\gamma_2}\bar{\pa}_{z_2}|_{\gamma_2=0}$ to \eqref{eq: fuse_pa_exp}, with this refined error term. The first term vanishes and we get, as required, 
\begin{multline}
\pa\Phi(z_1)\bar{\pa}\Phi(z_2)=\pa_{\gamma_2}\left.\left(\frac{1}{\gamma_2}\bar{\pa}_{z_2}\pa_{z_2}\Nexp{\gamma_2}{z_2}\right)\right|_{\gamma_2=0}+O(z_1-z_2)\\=\frac12\pa^2_{\gamma_2}\left.\bar{\pa}_{z_2}\pa_{z_2}\Nexp{\gamma_2}{z_2}\right|_{\gamma_2=0}+O(z_1-z_2),
\end{multline}
where in the last equality we have used the observation that $\bar{\pa}_{z_2}\pa_{z_2}\ccor{\Nexp{\gamma_2}{z_2}\Op}$ vanishes to the second order at $\gamma_2=0$, which can be readily checked from the definition when $\Op=\Nexp{\gamma_3}{z_3}\dots\Nexp{\gamma_n}{z_n}$ and then extended to the general case in the familiar way.
\end{proof}

\begin{lem} For all $\gamma\in\C$, we have the following operator product expansions:
\begin{multline}
\pa\Phi(z_1)\nord{\cos(\gamma\Phi(z_2))}\\=\frac{\gamma}{2(z_1-z_2)}\nord{\sin(\gamma\Phi(z_2))}+\frac{1}{\gamma}\pa_{z_2}\nord{\sin(\gamma\Phi(z_2))}+O(z_1-z_2)\label{eq: fuse_pa_cos}
\end{multline}
\begin{multline}
\pa\Phi(z_1)\nord{\sin(\gamma\Phi(z_2))}\\=-\frac{\gamma}{2(z_1-z_2)}\nord{\cos(\gamma\Phi(z_2))}-\frac{1}{\gamma}\pa_{z_2}\nord{\cos(\gamma\Phi(z_2))}+O(z_1-z_2)\label{eq: fuse_pa_sin}
\end{multline}
\begin{multline}
\label{eq: fuse_sin_cos}
\nord{\sin(\gamma\Phi(z_1))}\nord{\cos(\gamma\Phi(z_2))}=\frac{1}{2}|z_1-z_2|^{\gamma^2}(\nord{\sin({2}\gamma\Phi(z_2))}+O(z_1-z_2))\\
+\frac{\gamma}{2}|z_1-z_2|^{-\gamma^2}\left((z_1-z_2)\pa_{z_2}\Phi(z_2)+(\bar{z}_1-\bar{z}_2)\bar{\pa}_{z_2}\Phi(z_2)+O(z_1-z_2)^2\right)\\
\end{multline}
\begin{multline}
\label{eq: fuse_cos_cos}
\nord{\cos(\gamma\Phi(z_1))}\nord{\cos(\gamma\Phi(z_2))}=\frac12|z_1-z_2|^{\gamma^2}(\nord{\cos(2\gamma\Phi(z_2))}+O(z_1-z_2))\\
+|z_1-z_2|^{-\gamma^2}\left(\frac12+O(z_1-z_2)^2\right).
\end{multline}
\begin{multline}
\label{eq: fuse_sin_sin}
\nord{\sin(\gamma\Phi(z_1))}\nord{\sin(\gamma\Phi(z_2))}=-\frac12|z_1-z_2|^{\gamma^2}(\nord{\cos(2\gamma\Phi(z_2))}+O(z_1-z_2))\\
+|z_1-z_2|^{-\gamma^2}\left(\frac12+O(z_1-z_2)^2\right).
\end{multline}
\end{lem}
\begin{proof}
The expansions \eqref{eq: fuse_pa_cos}--\eqref{eq: fuse_pa_sin} are readily obtained by adding (respectively subtracting) \eqref{eq: fuse_pa_exp} with $\gamma_2=\pm \i\gamma$ and multiplying by $\tfrac{1}{2}$ (respectively $\tfrac{1}{2\i}$). To prove \ref{eq: fuse_sin_cos}, we write by linearity
\begin{multline}
\nord{\sin(\gamma\Phi(z_1))}\nord{\cos(\gamma\Phi(z_2))}=\frac{1}{4\i}\left(\Nexp{\i\gamma}{z_1}\Nexp{\i\gamma}{z_2}-\Nexp{- \i\gamma}{z_1}\Nexp{-\i\gamma}{z_2}\right)\\
+\frac{1}{4\i}\left(\Nexp{ \i\gamma}{z_1}\Nexp{-\i\gamma}{z_2}-\Nexp{-\i\gamma}{z_1}\Nexp{\i\gamma}{z_2}\right)
\end{multline}
and apply \eqref{eq: fuse_exps} to the first term and \eqref{eq: fuse_exp_bis} to the second one. The expansions \eqref{eq: fuse_cos_cos} and \eqref{eq: fuse_sin_sin} are proven in exactly the same way; only some signs and constants differ.
\end{proof}

We mention for completeness the remaining operator product expansions (that we do not need in the remainder of the article):
\begin{align*}
\pa\Phi(z_1)\pa\Phi(z_2)=&-\frac{1}{2(z_1-z_2)^2}+O(1),\\
\nord{|\nabla\Phi(z_1)|^2}\Nexp{\gamma_2}{z_2}=&\frac{1}{|z_1-z_2|^2}\Nexp{\gamma_2}{z_2}+O(|z_1-z_2|^{-1}),\\
\pa\Phi(z_1)\nord{|\nabla\Phi(z_2)|^2}=&-\frac{2}{(z_1-z_2)^2}\pa_{\bar{z}_2}\Phi(z_2)+O(1), \\
\nord{|\nabla\Phi(z_1)|^2}\nord{|\nabla\Phi(z_2)|^2}=&4|z_1-z_2|^{-4}+O(|z_1-z_2|^{-3}).
\end{align*}
These can be derived from \eqref{eq: fuse_exps} similarly to the proof of Lemma \ref{lem: fuse_pa_exp}.

\subsection{Proof of Theorem \ref{th:main} given Theorem \ref{th:ds}.}
We are ready to derive Theorem \ref{th:main} from Theorem \ref{th:ds}.
\begin{proof}
We start by extending the identity \eqref{eq: Thm1} to the case when some of the $\Op_{z_i}$ are $\psi_{z_i}$ or $\psi_{z_i}^\star$, by induction on the number of such $z_i$. Assume on the one hand that we already know the identity 
\[
\ccor{\mu_{z_1}\sigma_{z_2}\Op}^2=\ccor{\bmu{z_1}\bsigma{z_2}\hOp}.
\]
We compute the leading term of the asymptotics of both sides as $z_1\to z_2$. The asymptotics of the right-hand side is given by \eqref{eq: fuse_sin_cos}, specifically
\begin{multline}
\ccor{\bmu{z_1}\bsigma{z_2}\hOp}\\=|z_1-z_2|^\frac{1}{2}\ccor{\nord{\sin(\sqrt{2}\Phi(z_2))}\hOp}+\frac{\sqrt{2}}{2}\frac{z_1-z_2}{|z_1-z_2|^\frac12}\ccor{\pa\Phi(z_2)\hOp}\\+\frac{\sqrt{2}}{2}\frac{\bar{z}_1-\bar{z}_2}{|z_1-z_2|^\frac12}\ccor{\bar{\pa}\Phi(z_2)\hOp}+O(z_1-z_2)^\frac32.
\end{multline}
On the other hand \cite[Eq. (6.11) together with (6.10) and (5.8)]{CHI2} gives
\begin{multline}
\ccor{\mu_{z_1}\sigma_{z_2}\Op}^2=|z_1-z_2|^\frac12\left(\frac{e^{-\frac{ \i\pi}{4}}}{2}\frac{(z_1-z_2)^\frac12}{|z_1-z_2|^\frac12}\ccor{\psi_{z_2}\Op}+\frac{e^{\frac{ \i\pi}{4}}}{2}\frac{(\bar{z}_1-\bar{z}_2)^\frac12}{|z_1-z_2|^\frac12}\ccor{\psi^\star_{z_2}\Op}+O(z_1-z_2)\right)^2\\
=\frac12|z_1-z_2|^\frac12\ccor{\psi_{z_2}\Op}\ccor{\psi^\star_{z_2}\Op}-\frac{\i}{4}\frac{z_1-z_2}{|z_1-z_2|^\frac12}\ccor{\psi_{z_2}\Op}^2+\frac{\i}{4}\frac{\bar{z}_1-\bar{z}_2}{|z_1-z_2|^\frac12}\ccor{\psi^\star_{z_2}\Op}^2+O(z_1-z_2)^\frac32.
\end{multline} 
Comparing the coefficients yields 
\[
\ccor{\psi_{z_2}\Op}^2=\i 2\sqrt{2}\ccor{\pa\Phi(z_2)\hOp} \quad\text{and}\quad \ccor{\psi^\star_{z_2}\Op}^2=-\i 2\sqrt{2}\ccor{\bar{\pa}\Phi(z_2)\hOp},
\]
that is, we get the bosonization identity with one more fermion than the string $\Op$ already contained, and with the parity conditions preserved. This enables the inductive argument.

In a similar manner we can extend the result to energies, by passing to the limit $z_1\to z_2$ in the identity $\ccor{\psi_{z_1}\psi^\star_{z_2}\Op}^2=8\ccor{\pa\Phi(z_1)\bar{\pa}\Phi(z_2)\hOp}$, using \eqref{eq: fuse_pa_pa_bar} and comparing with \cite[Eq. 6.3]{CHI2} yields
\[
\ccor{\en_{z_2}\Op}^2=-\frac{1}{2}\ccor{\nord{|\nabla\Phi(z_2)|^2}\hOp}.
\]
\end{proof}

\begin{rem} \label{rem: prod_boson}
As a by-product of the above proof we get another bosonization prescription not covered directly by Theorem \ref{th:main}; namely we see that, for $\Op=\Op_{z_1}\dots\Op_{z_N}$ 
\[
\ccor{\psi_{z}\Op}\ccor{\psi^\star_{z}\Op}=\ccor{\nord{2\sin(\sqrt{2}\Phi(z))}\hOp}
\]
when $|\{i:\Op_{z_i}\in\{\sigma_{z_i},\psi_{z_i},\psi^\star_{z_i}\}\}|$ and $|\{i:\Op_{z_i}\in\{\mu_{z_i},\psi_{z_i},\psi^\star_{z_i}\}\}|$ are both odd.
Also, putting $\gamma=\frac{1}{\sqrt{2}}$ in \eqref{eq: fuse_cos_cos} and comparing with \cite[Eq. 6.12]{CHI2} gives 
\[
\ccor{\en_{z}\Op}\ccor{\Op}=\ccor{\nord{2\cos(\sqrt{2}\Phi(z))}\hOp}
\]
when $|\{i:\Op_{z_i}\in\{\sigma_{z_i},\psi_{z_i},\psi^\star_{z_i}\}\}|$ and $|\{i:\Op_{z_i}\in\{\mu_{z_i},\psi_{z_i},\psi^\star_{z_i}\}\}|$ are both even. These identities can be naturally interpreted as claiming that $\nord{2\sin(\sqrt{2}\Phi(z))}$ (respectively, $\nord{2\cos(\sqrt{2}\Phi(z))}$) corresponds to $\psi_z\tilde{\psi}^\star_z$ (respectively, $\frac12(\en_z+\tilde{\en}_z)$) in the two independent copies of the Ising model, as discussed after the statement of Theorem \ref{th:main}.
\end{rem}

\subsection{Extensions}
\label{sec: extensions}
In this section, we describe an extension of Theorem \ref{th:main} to the boundary conditions $\mathcal B$ defined by dividing the boundary into finitely many arcs, with boundary conditions on each arc set to $\plus$, $\minus$, or $\free$. This is the setup of \cite[Theorem 1.3]{CHI2}, where the scaling limits $\langle \Op_{z_1}\dots\Op_{z_n}\rangle_{\Omega,\mathcal B}$ of such correlations are computed, where $\Op_{z_i}\in\{\sigma_{z_i},\en_{z_i}\}$ for all $i$. Following \cite{CHI2}, we focus on correlation functions involving spin and energy only, although as discussed after \cite[Theorem 1.3]{CHI2}, it is in principle possible to include fermions and disorders, as well as boundary fields. We state the result somewhat loosely:

\begin{prop}\label{pr:pmfreeboso}
The (continuous) correlation functions $\langle \Op_{z_1}\dots\Op_{z_n}\rangle_{\Omega,\mathcal B}$, where $\Op_{z_i}\in\{\sigma_{z_i},\en_{z_i}\}$, featuring in \cite[Theorem 1.3]{CHI2}, can be expressed in terms of the bosonic correlation functions.
\end{prop}

\begin{proof}
The proof is based on the fact that $\langle \mathcal \Op_{z_1}\dots\Op_{z_n}\rangle_{\Omega,\mathcal B}$ can be expressed in terms of the correlations $\langle \Op_{w_1}\dots\Op_{w_N}\rangle_\Omega$ with locally monochromatic boundary condtions, $\Op_{w_i}\in\{\en_{w_i},\psi_{w_i},\sigma_{w_i}\}$ by a series of passages to a limit, see \cite[Section 5]{CHI2}. We now recall the procedure. 

First, one defines the correlation $\langle \Op_{w_1}\dots\Op_{w_N}\rangle_\Omega$, still with the locally monochromatic boundary conditions, where some spins are allowed to be on the $\fixed$ part of the the boundary, see \cite[Proposition 5.4, Remark 5.5, Definition 5.6, Definition 5.7, Definition 5.9, Remark 5.11, and Definition 5.17]{CHI2}. The bottom line of these definitions is that if, say, $w_1,\dots,w_l\in\Omega$ and $w_{l+1},\dots,w_n\in\fixed$, then
\begin{multline}\label{eq:bdry}
\langle \Op_{w_1}\dots\Op_{w_l}\sigma_{w_{l+1}}\dots\sigma_{w_N}\rangle_{\Omega}\\=\lim_{\hat{w}_i\to w_i}\prod_i2^{-\frac14}\mathrm{crad}_\Omega(\hat{w}_i)^{\frac18}\langle \Op_{w_1}\dots\Op_{w_N}\sigma_{\hat{w}_{l+1}}\dots\sigma_{\hat{w}_N}\rangle_{\Omega}.
\end{multline}
(The factor of $2^{-\frac14}$ is missing in \cite{CHI2} by mistake.) The conformal radius in a multiply connected domain $\Omega$ is defined by replacing it with a simply-connected domain $\Omega'\subset\hat{\C}$, removing all boundary components except the one containing $w_i$; it is easy to see that then as $\hat w_i\to w_i$, $\log\crad_{\Omega'}(\hat{w}_i)=g_{\Omega'}(\hat{w}_i,\hat{w}_i)=g_{\Omega}(\hat{w}_i,\hat{w}_i)+o(1)$. We note that these correlations are a scaling limit of the discrete correlations with locally monochromatic boundary conditions, and they in fact do not depend on the positions of $w_{l+1},\dots,w_N$ within the $\fixed$ part of their boundary components. This is valid for arbitrary fields $\Op_{w_i}\in \{\sigma_{w_i},\mu_{w_i},\psi_{w_i},\psi^\star_{w_i},\en_{w_i}\}$, but we do require $|\{i\leq l:\Op_{w_i}\in\{\sigma_{w_i},\psi_{w_i},\psi^\star_{w_i}\}\}|+N-l$ to be even, else the correlation is defined to be zero.

 The next step is to define the correlations with  \emph{auxiliary boundary conditions} $\widetilde {\mathcal B}$, see \cite[Section 5]{CHI2}. These are specified, in addition to the partition of $\pa\Omega$ into $\fixed$ and $\free$, by a choice of distinct points $b_1,...,b_q$ on $\partial \Omega$, such that there's an even number of them on each connected component of $\partial \Omega$, and the points either belong to $\fixed$, or are endpoints of free arcs. One then defines (see \cite[Definition 5.24]{CHI2} and the discussion thereafter)
 \begin{equation}\label{eq:aux}
\langle \Op_{w_1}\dots\Op_{w_N}\rangle_{\Omega,\widetilde{\mathcal B}}:=\lim_{z_1\to b_1,...,z_q\to b_q}\frac{\langle \psi_{z_1}\cdots \psi_{z_q}\Op_{w_1}\dots\Op_{w_N}\rangle_\Omega}{\langle \psi_{z_1}\cdots \psi_{z_q}\rangle_\Omega}.
\end{equation}

In the discrete picture, this corresponds to the situation where the boundary spins are still random and ``locally monochromatic" away from $b_i$, but forced to change sign at each $b_i\in\fixed$ and across each free arc that has exactly one of its endpoints in $\{b_1,\dots,b_q\}$.

Finally, given a boundary condition $\mathcal{B}$ specified by a partition of $\pa\Omega$ into $\plus$, $\minus$ and $\free$ arcs, we define the auxilliary boundary condition $\widetilde{\mathcal B}$ by including into $\{b_1,\dots,b_q\}$ all points separating a $\plus$ arc from a $\minus$ arc, and one endpoint of each free arc separating $\plus$ from $\minus$. Further, for each boundary component that has at least one $\plus$ or $\minus$ arc, we pick a point on one of those arcs; denote by $w_1,\dots,w_d$ the resulting collection of points. Note that to recover $\mathcal{B}$ from $\widetilde{\mathcal{B}}$, we need to specify the spin at each $w_i$. Define $\tau\in \{-1,1\}^d$ as follows: $\tau_j=1$ if the arc $j$ is a subset of $\{$\texttt{plus}$\}$ and $\tau_j=-1$ if the arc $j$ is a subset of $\{$\texttt{minus}$\}$. We then define

\begin{equation}\label{eq:pmfree}
\langle \Op_{z_1}\dots\Op_{z_n}\rangle_{\Omega,\mathcal B}=\frac{\sum_{S\subset \{1,...,d\}} \alpha_S(\tau)\langle \prod_{j\in S}\sigma_{w_j}\Op_{z_1}\dots\Op_{z_n}\rangle_{\Omega,\widetilde{\mathcal B}} }{\sum_{S\subset \{1,...,d\}}\alpha_S(\tau)\langle \prod_{j\in S}\sigma_{w_j}\rangle_{\Omega,\widetilde{\mathcal B}}},
\end{equation}
where the coefficients $\alpha_S(\tau)$ are defined as follows. Consider the function $I_\tau:\{\pm 1\}^d \to \R$,
\[
I_\tau(\sigma)=\begin{cases}
1, & \sigma=\tau,\\
0, & \text{else},
\end{cases}
\]and let $\alpha_S(\tau)$ be its Fourier--Walsh coefficients (see also \cite[(2.18)]{CHI2}), i.e., the unique coefficients $(\alpha_S(\tau))_{S\subset \{1,...,d\}}$ such that for each $\sigma\in \{\pm 1\}^d$,
\[
I_\tau(\sigma)=\sum_{S\subset \{1,...,d\}}\alpha_S(\tau)\prod_{j\in S}\sigma_j.
\]
By \cite[Lemma 5.25]{CHI2}, the correlation \eqref{eq:pmfree} is well defined, i.e., the denominator does not vanish. The fact that $\langle \Op_{z_1}\dots\Op_{z_n}\rangle_{\Omega,\mathcal B}$ thus defined do correspond to a scaling limit of correlation functions of the discrete Ising model with the $+/-/$\texttt{free} boundary conditions is the content of \cite[Theorem 1.3]{CHI2}. Note that in the left-hand side of \eqref{eq:pmfree}, we allow arbitrary parity of the spin insertions -- some terms simply vanish if they contain an odd amount of bulk and boundary spins in total.

Starting with the bosonization identity $$\langle \Op_{w_1}\dots\Op_{w_N}\rangle_{\Omega}^2=\langle \hOp_{w_1}\dots\hOp_{w_N}\rangle_{\Omega}$$ and applying all the above step, we can express $\langle \Op_{z_1}\dots\Op_{z_n}\rangle_{\Omega,\mathcal B}$ in terms of the (limits of) bosonic correlations.
\end{proof}

Let us remark that in the simply-connected case, the third step of the above procedure is trivial, and boundary conditions $\widetilde{\mathcal B}$ are essentially equivalent to $\mathcal{B}$. Indeed, let $z_1,\dots,z_n\in \Omega$. If $|\{i:\Op_{z_i}\in\{\sigma_{z_i},\psi_{z_i},\psi^\star_{z_i}\}\}|$ is even, then $\langle \Op_{z_1}\dots\Op_{z_n}\rangle_{\Omega,\widetilde{\mathcal B}}=\langle \Op_{z_1}\dots\Op_{z_n}\rangle_{\Omega,\mathcal B},$
by spin flip symmetry. If $|\{i:\Op_{z_i}\in\{\sigma_{z_i},\psi_{z_i},\psi^\star_{z_i}\}\}|$ is odd, we need to insert one boundary spin, say at $w_1$, and so, we have $d=1$, $I_\tau(\sigma)=\frac12(1+\tau_1\sigma_1)=\alpha_\emptyset(\tau_1)+\alpha_{\{1\}}(\tau_1)\sigma_1$, where $\tau_1=\pm1$ is the boundary condition at $w_1$, and
$$
\langle \Op_{z_1}\dots\Op_{z_n}\rangle_{\Omega,\mathcal B}=
\frac{\frac{1}{2}\cdot 0+\frac12\cdot\tau_1\langle \sigma_{w_1}\Op_{z_1}\dots\Op_{z_n}\rangle_{\Omega,\widetilde{\mathcal B}}}{\frac{1}{2}+\frac12\cdot 0}=\tau_1\langle \sigma_{w_1}\Op_{z_1}\dots\Op_{z_n}\rangle_{\Omega,\widetilde{\mathcal B}}.$$

Coming back the the general case, let us give a more concrete description of the first two steps in the above procedure; we do not expect the third step to get any more concrete in general.

For the first step, we have the following interesting extension of the main theorem:
\begin{prop}
\label{prop: boson_bdry_spin}
We have the following bosonization prescription for $w_1,\dots w_l\in\Omega$  and $w_{l+1},\dots,w_N \in \fixed$:
$$
\langle \Op_{w_1}\dots\Op_{w_l}\sigma_{w_{l+1}}\dots\sigma_{w_N}\rangle^2_{\Omega}=\langle \hOp_{w_1}\dots\hOp_{w_N}\rangle_{\Omega},
$$
where for $i=1,\dots,l$, each pair $\Op_{w_i},\hOp_{w_i}$ is as in Theorem \ref{th:main} while  for $i=l+1,\dots,N$, we have $\hOp_{w_i}=(-1)^{\frac{\Phi(w_i)}{\sqrt{2}\pi}}=(-1)^{\frac{\xi(w_i)}{\sqrt{2}\pi}}$. 
\end{prop}
\begin{proof}
We will square \eqref{eq:bdry} and apply the bosonization formula of Theorem \ref{th:main} inside the limit on the right-hand side. To compute the limit, we first study the limiting behavior of correlations of exponentials. Let $\check{\Op}_{w_{i}}=\nexp{\gamma_i}{w_{i}}$ for $i=1,\dots,l,$ and $\check{\Op}_{\hat{w}_{i}}=\nexp{\gamma_i}{\hat{w}_{i}}$ for $i=l+1,\dots,N$, then 
\begin{equation}
\label{eq: exp_to_bdry}
\lim_{\hat{w_i}\to w_i}\prod^N_{i=l+1}\crad_\Omega(\hat{w}_i)^{-\frac{\gamma^2_i}{2}} \ccor{\check{\Op}_{w_{1}}\dots\check{\Op}_{w_{l}}\check{\Op}_{\hat{w}_{l+1}}\dots\check{\Op}_{\hat{w}_{N}}}_\Omega=\ccor{\check{\Op}_{w_{1}}\dots\check{\Op}_{w_{l}}}_\Omega,
\end{equation}
using the definition \eqref{eq: defnexp} and observing that $g_\Omega(\hat{w}_i,\hat{w}_i)=\log\crad_\Omega(\hat{w}_i)+o(1)$ while $G_\Omega(w_i,\hat{w}_j)\to 0$ (respectively, $G_\Omega(\hat{w}_i,\hat{w}_j)\to 0$) for each $i=1,\dots,l$ and $j=l+1,\dots,N$  (respectively, $l+l\leq i\neq j\leq N$). We specialize this to $\gamma_{l+1}=\dots=\gamma_N=\i\frac{\sqrt 2}{2},$ and observe that by linearity, \eqref{eq: exp_to_bdry} also holds when $\check{\Op}_{\hat{w}_i}=\nord{\cos\left(\frac{\sqrt2}{2}\varphi(\hat{w}_i)\right)}$ for all $i=l+1,\dots,N$, while if some of them are $\nord{\sin\left(\frac{\sqrt2}{2}\varphi(\hat{w}_i)\right)}$ and others $\nord{\cos\left(\frac{\sqrt2}{2}\varphi(\hat{w}_i)\right)}$, the limit vanishes.   

Now, let's include the instanton component: let $\check{\Op}_{w_{i}}=\Nexp{\gamma_i}{w_{i}}$ for $i=1,\dots,l,$ and $\check{\Op}_{w_{i}}=\nord{\cos\left(\frac{\sqrt2}{2}\Phi(\hat{w}_i)\right)}$ for $i=l+1,\dots,N$; we then just expand each cosine into exponentials, pass to the limit term by term, and then regroup terms. To help with the bookkeeping, we note that inside any correlation, 
\begin{multline}
\label{eq: cos_sum}
\nord{\cos\left(\frac{\sqrt2}{2}\Phi(\hat{w}_i)\right)}=\nord{\cos\left(\frac{\sqrt2}{2}\varphi(\hat{w}_i)\right)}\cos\left(\frac{\sqrt2}{2}\xi(\hat{w}_i)\right)\\-\nord{\sin\left(\frac{\sqrt2}{2}\varphi(\hat{w}_i)\right)}\sin\left(\frac{\sqrt2}{2}\xi(\hat{w}_i)\right),
\end{multline}
since by definition, the usual Euler identities are valid for normal ordered sines and cosines. Plugging \eqref{eq: cos_sum} into \eqref{eq: exp_to_bdry} and expanding, by the above observation, only the term with all cosines survives in the limit, i.e., we get  
\begin{multline*}
\lim_{\hat{w_i}\to w_i}\prod^N_{i=l+1}\crad_\Omega(w_i)^{-\frac{\gamma^2_i}{2}} \ccor{\check{\Op}_{w_{1}}\dots\check{\Op}_{w_{l}}\check{\Op}_{\hat{w}_{l+1}}\dots\check{\Op}_{\hat{w}_{N}}}_\Omega\\
=\ccor{\nord{e^{\gamma_1\varphi(w_1)}}\dots \nord{e^{\gamma_1\varphi(w_l)}}}_\Omega \E\left(\prod_{i=1}^le^{\gamma_i\xi(w_i)}\prod_{i=l+1}^N\cos\left(\frac{\sqrt2}{2}\xi(w_i)\right)\right).
\end{multline*}
It remains to note that  $\cos\left(\frac{\sqrt2}{2}\xi(w_i)\right)=(-1)^{\frac{\xi(w_i)}{\sqrt 2 \pi}}$ since $\xi\in\sqrt{2}\pi \Z$ on the fixed part of the boundary. The same result holds if we replace $\check{\Op}_{w_i}=\Nexp{\gamma_i}{w_{i}}$ with the bosonic counterparts of $\hOp_{w_i}$ of any fields $\Op_{w_i}\in\{\sigma_{w_i},\mu_{w_i},\psi_{w_i},\psi^\star_{w_i},\en_{w_i}\}$, either by linearity (for $\sigma_{w_i},\mu_{w_i}$) or by checking that the above limiting procedure commutes with corresponding $\mathcal{D}_{w_i}$ (for $\psi_{w_i},\psi^\star_{w_i},\en_{w_i}$). We leave further details to the reader.
\end{proof}

\begin{rem}
Note that, in the context of Theorem \ref{th:main}, if $|\{i:\Op_{z_i}\in\{\sigma_{z_i},\psi_{z_i},\psi^\star_{z_i}\}\}|$ is odd, we can make it even by adding a spin on the boundary arc where $\xi$ is normalized to be zero. This has no effect on the right-hand side of \eqref{eq: Thm1}. By the discussion before Proposition \ref{prop: boson_bdry_spin}, the effect of spin insertion on the left-hand side it tantamount to setting boundary conditions to $+$ on that boundary component. In particular, in the simply-connected case, Theorem \eqref{th:main} can be formulated without the restriction on the parity of $|\{i:\Op_{z_i}\in\{\sigma_{z_i},\psi_{z_i},\psi^\star_{z_i}\}\}|$ and with boundary conditions set to $+$ on all fixed arcs. This type of remarks, for purely spin correlations, have been already made in \cite{CHI1,CHI2,Dubedat,JSW}.
\end{rem}

For the second step of the above procedure, namely, the boundary conditions $\widetilde{\mathcal B}$, we can use two alternative bosonization prescriptions. The first one is to simply square \eqref{eq:aux}; this leads to 
\[
\langle \Op_{w_1}\dots\Op_{w_N}\rangle_{\Omega,\widetilde{\mathcal B}}^2=\lim_{z_1\to b_1,...,z_q\to b_q}\frac{\langle \partial \Phi(z_1)\cdots \partial \Phi(z_q) \hOp_{w_1}\dots\hOp_{w_N}\rangle_\Omega }{\langle \partial \Phi(z_1)\cdots \partial \Phi(z_q) \rangle_\Omega}.
\]

Alternatively, if $\Op_{w_i}\in\{\sigma_{v_i},\en_{v_i}\}$, then they are real fields, and due to \cite[Proposition 5.18]{CHI2} we can alternatively use $\psi^\star_{z_i}$ in \eqref{eq:aux}. This gives by Remark \ref{rem: prod_boson}
\begin{multline}
\langle \Op_{w_1}\dots\Op_{w_N}\rangle_{\Omega,\widetilde{\mathcal B}}^2=\lim_{z_i\to b_i}\frac{\langle \psi_{z_1}\cdots \psi_{z_q}\Op_{w_1}\dots\Op_{w_N}\rangle_\Omega\langle \psi^\star_{z_1}\cdots \psi^\star_{z_q}\Op_{w_1}\dots\Op_{w_N}\rangle_\Omega}{\langle \psi_{z_1}\cdots \psi_{z_q}\rangle_\Omega\langle \psi^\star_{z_1}\cdots \psi^\star_{z_q}\rangle_\Omega}\\
=\lim_{z_i\to b_i}\frac{\langle \nord{2\sin(\sqrt{2}\Phi(z_1))}\cdots \nord{2\sin(\sqrt{2}\Phi(z_q))}\hOp_{w_1}\dots\hOp_{w_N}\rangle_\Omega }{\langle \nord{2\sin(\sqrt{2}\Phi(z_1))}\cdots \nord{2\sin(\sqrt{2}\Phi(z_q))} \rangle_\Omega}.
\end{multline}
In both of the above representations, the effect of forcing a $\pm1$ boundary change is akin to tilting the bosonic field measure by a boundary field; note however that these fields are not positive. 

\section{Hejhal--Fay identity on compact Riemann surfaces}\label{sec:hej}


Our proof of Theorem \ref{th:ds} is based on an identity, discovered by Hejhal \cite{Hejhal}  and independently by Fay \cite{Fay}, which expresses the square of a Szeg\H{o} kernel on a compact Riemann surface in terms of Abelian differentials. In this section we review concepts needed to state this identity (Theorem \ref{thm: Hejhal}): Riemann surfaces, Abelian differentials, spinors, Szeg\H{o} kernels, and theta functions. No material here is original; the only bit for which we have not found an explicit reference in the literature is Lemma \ref{lem: wind}. 

\subsection{Riemann surfaces and Abelian differentials.}
\label{sec: Riemann surfaces}
We start with the basics of the theory of Riemann surfaces that we will need in what follows. The material is standard and can be found in many sources, see e.g. \cite{Ahlfors, FK}. 

\begin{defn}
A Riemann surface $\M$ is a connected two (real) dimensional smooth manifold equipped with an atlas $(\mathcal U_\alpha,z_\alpha)_{\alpha\in \mathcal A}$, where 
\begin{itemize}
    \item the \emph{charts} $(\mathcal U_\alpha)_{\alpha\in \mathcal A}$ form an open cover of $\M$,
    \item the \emph{local coordinate map} $z_\alpha$ is a homeomorphism between $\mathcal U_\alpha$ and some open subset of $\C$,
    \item the \emph{transition functions} $\varphi_{\alpha,\beta}=z_\beta \circ z_\alpha^{-1}: z_\alpha(\mathcal U_\alpha \cap \mathcal U_\beta)\to z_\beta(\mathcal U_\alpha\cap \mathcal U_\beta)$ are holomorphic functions when $\mathcal U_\alpha\cap \mathcal U_\beta\neq \emptyset$.
\end{itemize}
Two atlases are equivalent if their union is also an atlas, and Riemann surfaces with equivalent atlases are not distinguished. 
\end{defn}

A Riemann surface is automatically oriented. The \emph{genus} of a compact Riemann surface is defined to be 1/2 the rank of its first homology group $H_1(\M, \Z)$, see e.g. \cite[Chapter I.2.5]{FK}). Topologically, a compact Riemann surface of genus $g$ is a sphere with $g$ handles attached. A \emph{Torelli marking} of a Riemann surface is a choice of generators $\{A_j,B_j\}_{j=1}^g$ of $H_1(\M, \Z)$ so that their intersection numbers (see \cite[Chapter III.1]{FK}) satisfy:
\[
A_{j}\sharp B_{k}=\delta_{jk}, \qquad A_{j}\sharp A_{k}=0, \quad \text{and} \quad B_{j}\sharp B_{k}=0.
\]

To understand this condition, the reader may think that the generators $\{A_j,B_j\}_{j=1}^g$ can be represented by simple closed oriented smooth loops, mutually disjoint except for $A_j$ and $B_j$, which intersect transversally exactly once so that their tangent vectors form a positive frame, see Figure \ref{fig:hombas}. 

\begin{figure}
    \centering
        \begin{tikzpicture}[scale=2]
        \draw (0,0) ellipse (3.2 and 1.8);
        \draw[rounded corners=28pt] (-2.6,.1)--(-1.5,-0.6)--(-.4,.1);
        \draw[rounded corners=24pt] (-2.43,0)--(-1.5,0.6)--(-.58,0);
        
        \draw[rounded corners=28pt] (.4,.1)--(1.5,-0.6)--(2.6,.1);
        \draw[rounded corners=24pt] (0.58,0)--(1.5,0.6)--(2.43,0);

       \draw[blue,very thick] (-1.5,0) ellipse (1.4 and 0.7);
       \draw[blue,very thick] (1.5,0) ellipse (1.4 and 0.7);
      
       \draw[red, very thick] (1.95,0.3) to[out=60, in=180] (2.85,0.8);
       \draw[red, very thick, dashed] (1.95,0.3) to[out=0, in=270] (2.85,0.8);
       
       \draw[red, very thick] (-1.95,0.3) to[out=120, in=0] (-2.85,0.8);
       \draw[red, very thick, dashed] (-1.95,0.3) to[out=180, in=270] (-2.85,0.8);

       \draw[-{Latex[length=5mm, width=2mm, blue]}] (1.5,-0.7) -- (1.7,-0.7);
       \draw[-{Latex[length=5mm, width=2mm, blue]}] (-1.5,-0.7) -- (-1.3,-0.7);

      \draw[-{Latex[length=5mm, width=2mm, red]}] (2.53,0.75) -- (2.58,0.76);
      \draw[-{Latex[length=5mm, width=2mm, red]}] (-2.53,0.75) -- (-2.58,0.76);

      \node at (1.5,-1.1) {${\color{blue}B_2}$};
      \node at (-1.5,-1.1) {${\color{blue}B_1}$};

      \node at (-2.2,1) {${\color{red}A_1}$};
      \node at (2.2,1) {${\color{red}A_2}$};

      \node at (0,-1.1) {$\mathcal M$};
    \end{tikzpicture}
    \caption{A genus $2$ Riemann surface with a canonical homology basis $\{A_j,B_j\}_{j=1}^2$.}
    \label{fig:hombas}
\end{figure}
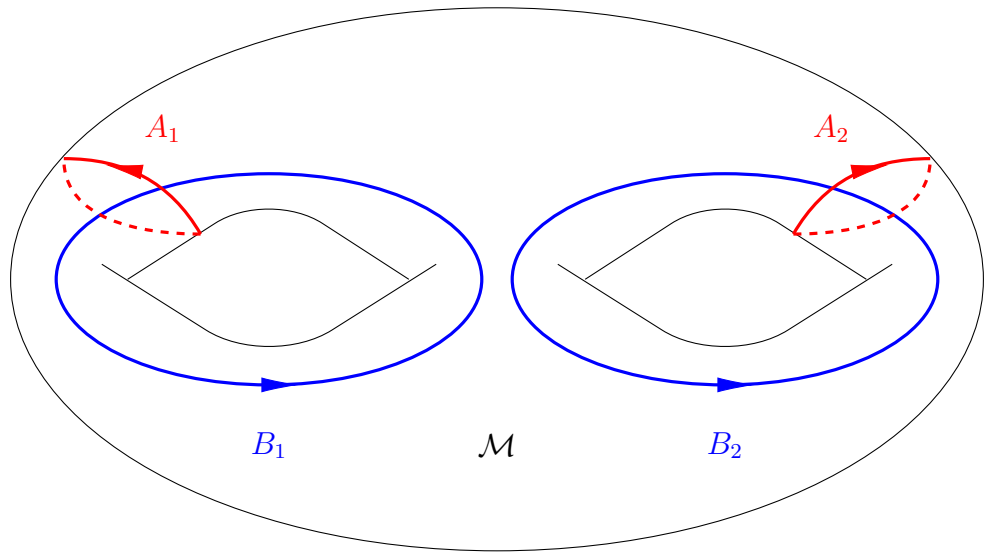

We will consider both compact Riemann surfaces as well as Riemann surfaces with boundaries, such that $\M\cup\pa \M$ is compact. For surfaces with boundary, the transition functions are assumed to be holomorphic up to the boundary.

\begin{defn}
A function $f$ from a Riemann surface $\M$ to the Riemann sphere $\hat{\C}$ is called meromorphic if $f\circ z^{-1}_{\alpha}$ is meromorphic for any coordinate map $\alpha$. A 1-form $\omega$ is meromorphic (holomorphic) if in any coordinate chart $z_\alpha:\mathcal{U}_\alpha\to U$, it is represented as $f_\alpha\circ z_\alpha\,dz_\alpha$, where $f_\alpha:U\to \hat{\C}$ is meromorphic (holomorphic). Meromorphic 1-forms are also called \emph{Abelian differentials}.
\end{defn}

In other words, a meromorphic 1-form is specified by a collection of its \emph{elements} $f_\alpha$, assigned to each coordinate chart and satisfying $f_{\alpha}(z)=f_{\beta}(\varphi_{\alpha,\beta}(z))\cdot \varphi'_{\alpha,\beta}(z)$ on the overlaps. This point of view will be useful later when we discuss spinors. While in planar domains, there is a bijective correspondence $f(z)\mapsto f(z)\,dz$ between meromorphic functions and 1-forms, on Riemann surfaces they are genuinely different objects.

Being 1-forms, Abelian differentials can be integrated along contours; in particular, if $P\in \M$ is a pole of an Abelian differential $u$, then one can define a residue of $u$ at $P$ by $\res_P u:=\frac{1}{2\pi \i}\oint_P u$, where the integral is over a small circle around $P$. Other coefficients of the Laurent expansion of a differential at a point are not well defined, i.e., depend on the
coordinate chart. However, the \emph{order} of a pole (or zero) is well defined. On a compact Riemann surface, any Abelian differential has only finitely many poles, and the sum of its residues is zero.

A (non-trivial) fundamental result concerns existence of Abelian differentials on a compact Riemann surface \cite[Chapter III]{FK}:
\begin{itemize}
\item The space of holomorphic Abelian differentials is $g$-dimensional; the map $u\mapsto (\oint_{A_1}u,\dots,\oint_{A_g}u)$ is an isomorphism of that space to $\C^g$.
\item For any $P\in \M$, there exists an Abelian differential $\beta_P$ with a double pole at $P$ (and no other poles). 
\item For any $P,Q\in \M$, there exists an Abelian differential $\omega_{P,Q}$ with simple poles at $P,Q$ (and no other poles) and $\res_P \omega_{P,Q}=-\res_Q \omega_{P,Q} =1$.
\end{itemize}
These differentials are sometimes called Abelian differentials of the first, second, and third kind respectively\footnote{More generally, Abelian differentials of the second kind are those with all residues vanishing, and of the third kind, those with only simple poles}. It is clear that $\omega_{P,Q}$ is uniquely defined up to an addition of a  holomoprhic differential, while $\beta_{P}$ is defined up to scaling and an addition of a holomorphic differential. We now fix this normalization.

\begin{defn}\label{def:normdiffs}
We denote by $u_1,\dots,u_g$ the holomorphic Abelian differentials uniquely specified by the normalization condition 
    \begin{equation}\label{eq:ab1norm}
        \oint_{A_j}u_k=\pi \i \delta_{j,k}.
    \end{equation}
    The \emph{Abel map} (with basepoint $P_0$) is a map from the universal cover of $\M$ to $\C^g$, given by 
    $$
    \Ab(P)=\left(\int^P_{P_0}u_1,\dots,\int^P_{P_0}u_g\right)
    $$
For $P,Q\in\M$, we denote by $\omega_{P,Q}$ the unique meromorphic differential with simple poles at $P,Q$ (no other poles) satisfying\footnote{The condition \eqref{eq: omega_A_loops} depends not only on the Torelli marking but also on their concrete representation by loops, which must avoid $P,Q$. In what follows, we will specify such a representation.}
\begin{align}
 \res_P \omega_{P,Q}=-\res_Q \omega_{P,Q} &=1,\\
 \oint_{A_1}\omega_{P,Q}=\dots=\oint_{A_g}\omega_{P,Q}&=0.\label{eq: omega_A_loops}
\end{align}Finally, given $Q\in\M$ and a coordinate chart $\mathcal{U}_\alpha\ni Q$, we define $\beta_{Q,\alpha}$ to be the unique Abelian differential with a double pole at $Q$ and no other poles, whose element in that chart has the expansion 
\begin{equation}
 (\beta_{Q,\alpha})_\alpha(z)=\frac{1}{(z-z_\alpha(Q))^2}+O(1)
 \end{equation}
 and satisfies the normalization
 \begin{equation}
 \label{eq: beta_A_loops}
 \oint_{A_1}\beta_{Q,\alpha}=\dots=\oint_{A_g}\beta_{Q,\alpha}=0.
\end{equation}
\end{defn}

Once the normalization (\ref{eq:ab1norm}, \ref{eq: omega_A_loops}, \ref{eq: beta_A_loops}) is fixed, we have no freedom to choose the integrals over the $B$-loops. We move on to some important notation:

\begin{defn}
The \emph{period matrix} of $\M$ is the matrix $\tau=(\tau_{ij})_{i,j=1}^g$ with entries given by 
\begin{equation}\label{eq:tau}
    \tau_{ij}=\oint_{B_i}u_j.
\end{equation}
\end{defn}

The period matrix determines uniquely a marked Riemann surface, up to conformal equivalence. Other properties of the Abelian differentials are summarised in the following proposition: 

\begin{prop}
\label{prop: Riemann_relations}
    (Riemann bilinear relations) We have the following properties of the period matrix and Abelian differentials:
    \begin{enumerate}

    \item The period matrix $\tau$ is symmetric, i.e.,  $\tau_{ij}=\tau_{ji}$ for all $i,j$, and its real part is strictly negative definite, i.e., $\sum_{i,j=1}^g \alpha_i \alpha_j \re \tau_{i,j}<0$ for any $(\alpha_1,...,\alpha_g)\in \R^g\setminus \{0\}$.
    \item The differential $\beta$ is symmetric, that is, for any choice of coordinate charts $ \mathcal{U}_\alpha\ni Q$, $\mathcal{U}_{\alpha'}\ni P$, one has
    \begin{equation}\label{eq: beta_symmetry}
    (\beta_{Q,\alpha})_{\alpha'}(z_{\alpha'}(P))=(\beta_{P,\alpha'})_{\alpha}(z_\alpha(Q)).
    \end{equation}
    \item One has, in any coordinate chart $\mathcal{U}_\alpha\ni Q$, 
    \begin{equation}\label{eq: omega_beta}
    \int_Q^P\beta_{R,\alpha}=(\omega_{P,Q})_\alpha(z_\alpha(R))   
    \end{equation}
     where the path of integration is chosen to avoid the loops $A_i,B_i$.
\end{enumerate}    
\end{prop}

\begin{proof}
These results are standard and can be derived by applying the Stokes theorem or the residue formula to a suitable Abelian differential in the simply connected region $\M'$ obtained by cutting $\M$ along the $A$ and $B$ loops, see e.g. \cite[Chapter III.2 and III.3]{FK} or \cite[Section 18]{Ahlfors}. (Note that our normalization involves an extra factor of $\i$ in \eqref{eq:ab1norm} compared to \cite{FK}). For convenience of the reader, we sketch a proof of the identities. If $\nu,\xi$ are two Abelian differentials and $\Xi(z)=\int^z\xi$, then $\Xi$ is a well-defined function on $\M'$, and if the loops representing $A_i$ and $B_i$ avoid the residues of both $\nu$ and $\xi$, we have 
\begin{equation}
\label{eq: bilinear}
2\pi \i \sum_p\res_p (\nu\Xi)=\int_{\pa\M'}\nu\Xi=\sum_i\left(\int_{A_i}\nu\int_{B_i}\xi-\int_{A_i}\xi\int_{B_i}\nu\right).
\end{equation}
To derive the symmetry $\tau_{ij}=\tau_{ji}$, plug in $\nu=u_i$ and $\xi=u_j$ and note that the left-hand side vanishes. For \eqref{eq: beta_symmetry}, use $\nu=\beta_{Q,\alpha}$ and $\xi=\beta_{P,\alpha'}$, notice that now the \emph{right-hand side} vanishes, and $\nu\Xi$ has two poles at $P,Q$, whose residues are conveniently evaluated in charts $\alpha'$ and $\alpha$ as $(\beta_{Q,\alpha})_{\alpha'}(z_{\alpha'}(P))$ and $-(\beta_{P,\alpha'})_{\alpha}(z_\alpha(Q))$ respectively. For \eqref{eq:  omega_beta}, use $\nu=\omega_{P,Q}$ and $\xi=\beta_{R,\alpha}$; again the right-hand side of \eqref{eq: bilinear} vanishes, and $\nu\Xi$ now has three  poles: at $P$ with residue $\Xi(P)$, at $Q$ with residue $-\Xi(Q)$, and at $R$ with residue $-(\omega_{P,Q})_\alpha(z_\alpha(R))$ (since $\Xi$ has a simple pole with residue $-1$ in the chart $\alpha$). 
\end{proof}

We note that (\ref{eq: beta_symmetry}) means that $\beta$ also behaves as an Abelian differential in its second argument, namely, the location of the pole. It is sometimes called the \emph{fundamental normalized bi-differential}, or the \emph{Bergman kernel} on $\M$. 

To lighten the notation, it is customary to write points of the Riemann surface as arguments of Abelian differentials, even though the latter have no well-defined values at points. It is to be understood that if the same point appears several times in an equation, then the same coordinate chart is used in every instance. Using such notation (\ref{eq: beta_symmetry}) and (\ref{eq: omega_beta}) become $\beta_Q(P)=\beta_P(Q)$ and $\int_Q^P\beta_{R}=\omega_{P,Q}(R)$, respectively. At this point it will be convenient to re-denote $\beta(P,Q):=\beta_P(Q)$.

\begin{exa}
\label{exa: riemann_sphere}
The Riemann sphere $\hat{\C}=\C\cup \{\infty\}$ has genus $0$, thus it has no Abelian differentials of the first kind. Other Abelian differentials are given (in the  coordinate chart $\C$) by 
\[
\omega_{P,Q}(z)=\left(\frac{1}{z-P}-\frac{1}{z-Q}\right)\,dz\quad\text{and}\quad \beta(P,Q)=\frac{1}{(P-Q)^2}\,dPdQ.
\]
On a torus $\C/(\Z+\nu\Z),$ where $\im \nu>0$, we have the unique Abelian differential of the first kind given by $\pi \i \,dz$. The only entry of the period matrix is $\tau_{11}=\pi \i \nu$. The Abelian differentials $\omega_{P,Q}$ and $\beta$ can be written in terms of Jacobi and Weierstrass elliptic functions, or Jacobi theta functions.
\end{exa}

When necessary, we will incorporate the underlying surface $\M$ into the notation, as in $\u_{\M,i}(P),\beta_{\M}(P,Q), \tau_\M$ etc.

\subsection{Spinors and Szeg\H{o} kernels on a Riemann surface} 
\label{sec: spinors_Szego}
To give an elementary introduction to spinors, or half-order differentials, we will introduce a class of finite atlases suited for our purpose. Given a triangulation $T$ of $\M$ we construct, in an obvious way, an atlas with \emph{simply-connected} charts $\mathcal{U}_\alpha$ labeled by vertices $\alpha$ of $T$ such that:
\begin{itemize}
\item $\alpha \in \mathcal{U}_\alpha$
\item The overlap $\mathcal{U}_{\alpha_1}\cap\mathcal{U}_{\alpha_2}$ is simply connected and non-empty if and only if $(\alpha_1\alpha_2)$ is an edge of $T$.
\item The overlap $\mathcal{U}_{\alpha_1}\cap\mathcal{U}_{\alpha_2}\cap \mathcal{U}_{\alpha_3}$ is non-empty if and only if $(\alpha_1\alpha_2\alpha_3)$ is a triangle of $T$.
\item The four-fold overlaps are all empty.
\end{itemize}
\begin{defn}
Assume that a holomorphic branch of $(\varphi_{\alpha_1\alpha_2}')^\frac12$ is chosen for each transition map $\varphi_{\alpha_1\alpha_2}$ between overlapping charts. The choice is called \emph{coherent} if 
\begin{equation}
    (\varphi_{\alpha_1\alpha_2}')^\frac12\cdot\left((\varphi_{\alpha_2\alpha_1}')^\frac12\circ\varphi_{\alpha_{1}\alpha_2}\right)\equiv 1
\end{equation}
and
\begin{equation}
((\varphi_{\alpha_3\alpha_1}')^\frac12\circ \varphi_{\alpha_2\alpha_3}\circ\varphi_{\alpha_1\alpha_2})\cdot((\varphi_{\alpha_2\alpha_3}')^\frac12\circ \varphi_{\alpha_1\alpha_2}) \cdot(\varphi_{\alpha_1\alpha_2}')^\frac12\equiv 1
\end{equation}
on every overlap and triple overlap, respectively. 
\end{defn}

Given \emph{two} coherent choices we define a $\Z_2$-valued function on the edges by $s(\alpha_1\alpha_2)=0$ or $s(\alpha_1\alpha_2)=1$ if, in the two choices, the branches of $(\varphi_{\alpha_1\alpha_2}')^\frac12$ agree or disagree, respectively. Extending $s$ to $\Z_2$-chains in $T$ by linearity, we see that the coherency condition ensures that $s$ vanishes on boundaries, i.e., it is a cocycle. We say that two coherent choices are \emph{equivalent} if $[s]=0$ in $H^1(T,\Z_2)$, i.e., if $s$ vanishes on all cycles. 

\begin{defn}
A \emph{spin line bundle} on $\M$ is a coherent choice of branches of $(\varphi_{\alpha_1\alpha_2}')^\frac12$, modulo the above equivalence. A \emph{meromorphic (respectively, holomorphic) spinor} on $\M$ is a meromorphic (holomorphic) section of a spin line bundle, that is, given a coherent choice of branches $\varphi'_{\alpha\beta}(z)^\frac12$, it is a collection of meromorphic (holomorphic) functions (``elements") $f_\alpha:U_\alpha\to\C$, where $U_\alpha=z_\alpha(\mathcal{U}_\alpha)$, satisfying on the overlaps
\begin{equation}
f_\alpha(z)=\varphi'_{\alpha\beta}(z)^\frac12f_{\beta}(\varphi_{\alpha\beta}(z)).\label{eq: 1/2form}    
\end{equation}
\end{defn}

If two coherent choices are equivalent, then the map $f_{\alpha}\mapsto s(\gamma_{\alpha_0\alpha})f_{\alpha}$, where $\alpha_0$ is any fixed vertex and $\gamma_{\alpha_0\alpha}$ is any path in $T$ from $\alpha_0$ to $\alpha$, is well defined and turns a collection satisfying (\ref{eq: 1/2form}) with one choice to a collection satisfying (\ref{eq: 1/2form}) with the other one. By construction, $H^1(T,\Z_2)\cong H^1(\M,\Z_2)$ acts freely and transitively on the set of spin line bundles. That is, there are $|H^1(T,\Z_2)|=2^{2g}$ spin line bundles, provided that there's at least one coherent choice. A simple cohomological argument for the existence of a coherent choice is given in \cite[Section 7]{HS}.

If $f$ is a meromorphic spinor, then the collection $\{f^2_\alpha\}_\alpha$ defines an Abelian differential whose poles and zeros are all of even order. Conversely, given such an Abelian differential $\omega$, taking element-wise square roots defines a spinor, and in particular, \emph{a spin line bundle}. Namely, in each chart $\mathcal{U}_\alpha$, choose arbitrarily a branch of the square root $(\omega_{\alpha})^\frac12$ of the element of $\omega$. Then, we can \emph{define} the branches of $\varphi'_{\alpha\beta}(z)^\frac12$ by $(\omega_{\alpha})^\frac12=(\varphi'_{\alpha\beta})^\frac12\cdot(\omega_{\beta})^\frac12\circ\varphi_{\alpha\beta}$; this choice is manifestly coherent. We note that meromorphic sections of a given spin line bundle form a vector space, and a product of two sections of \emph{the same} spin line bundle is an Abelian differential.

As an example, consider the twice punctured Riemann sphere $\C\setminus\{0\}$ and choose $\HH$, $e^{2\pi\i/3}\HH$ and $e^{4\pi\i/3}\HH$ to be coordinate charts, with identity coordinate maps\footnote{The atlas here does not really correspond to a triangulation, as the Riemann surface here is not compact, but otherwise all of the above considerations are valid.}. For the transition maps $(\varphi_{\alpha\beta}')^\frac12\equiv\sqrt{1}=\pm 1$ there are two spin line bundles $\SpStr$ and $\tilde\SpStr$, corresponding to choosing the $(-1)$ sign on an even and odd number of overlaps, respectively. Holomorphic sections of these bundles are given for example by $dz^\frac12$ and $\sqrt{z}\, dz^\frac12$, i.e., the elements are given by $\pm 1$ and $\pm\sqrt{z}$ with appropriate choice of signs.

Even though the set of spin line bundles can be parametrized by $H^1(\M,\Z_2)\cong \Z_2^{2g}$, this parametrization is not canonical, in that there's no distinguished spin line bundle that would correspond to the origin. (To see that this is indeed the case, consider the above example with coordinate maps $z\mapsto \log z$ instead.) A topological datum \emph{naturally} associated to a spin line bundle is instead a \emph{spin structure}. A spin structure is a way of assigning a winding (of the tangent vector) modulo $4\pi$ to any smooth simple closed loop, invariant under isotopies and satisfying certain other natural conditions that we will not go into, see \cite{Johnson}. To fix a spin structure, it is enough to specify the $2g$ bits $\wind(A_1)/2\pi \mod 2,\dots,\wind(B_g)/2\pi \mod 2$, where the generators are represented by simple loops. For a proof of equivalence of spin structures and spin line bundles, see \cite[Section 3]{Atiyah}; here we will explain how to go from the latter to the former. Given a meromorphic section $f$ of the bundle in question, $f^2$ is a 1-form,  which we can identify with a vector field, by choosing a Riemannian metric compatible with the conformal structure on $\M$. The vector field depends on the metric, but its direction does  not; in the chart $\mathcal{U}_\alpha$, the direction of the vector field is given by the direction of $\overline{f^2_{\alpha}(z_\alpha)}$. Now, if $\gamma$ avoids zeros and poles of $f$, we simply compute the winding of this vector field with respect to the tangent vector field along $\gamma$, modulo $4\pi$. 

\medskip 

\begin{defn}
 \cite{HS}, \cite[Section 2]{Fay},  Let $\M$ be a compact Riemann surface and let $\SpStr$ be a spin line bundle on $\M$. Let $Q\in \M$, and pick a coordinate chart $\mathcal{U}_\alpha \ni Q$. A \emph{Szeg\H{o} kernel} $\Lambda_{\M,\SpStr,\alpha}(\cdot,Q)$ is a meromorphic section of $\SpStr$, with a simple pole at $Q$ of residue $1$ in the chart $\mathcal{U}_\alpha$ and no other poles; in other words, its element $(\Lambda_{\M,\SpStr,\alpha}(\cdot,Q))_\alpha$ satisfies 
 $$
 (\Lambda_{\M,\SpStr,\alpha}(\cdot,Q))_\alpha(z_\alpha(P)) =\frac{1}{z_\alpha(P)-z_\alpha(Q)}+O(1),\quad P\to Q.
 $$
\end{defn}

The Szeg\H{o} kernel does not always exist, but it is possible to show \cite[Theorem 23]{Hejhal} that given $\M,\SpStr$, the following alternative holds: either 
\begin{enumerate}
\item the Szeg\H{o} kernel exists for all $Q,\alpha$, or
\item $\SpStr$ admits a non-trivial holomorphic section.
\end{enumerate}
 Let us explain the easy part of the alternative. If $h$ is a holomorphic section of $\SpStr$, then $\Lambda_{\M,\SpStr,\alpha} (\cdot,Q)h(\cdot)$ is an Abelian differential with the only pole at $Q$. Since the residues must sum up to zero, we see that $h(Q)=0$. In other words, if (2) holds, then a Szeg\H{o} kernel can exist only when $Q$ is a zero of $h$, i.e., for at most finitely many $Q$. A difference of two Szeg\H{o} kernels with the same $Q,\alpha$ is a holomorphic section, hence, in the case (1), the Szeg\H{o} kernel is unique. Cases (1) and (2) are determined by vanishing of a certain $\theta$-constant \cite[Theorem 23]{Hejhal}. We remark here that there is an invariant of a spin line bundle, called \emph{parity}; for \emph{odd} spin line bundles the case (2) above always holds, while for \emph{even} ones, case (1) holds generically, but case (2) may hold for special moduli, see \cite{Atiyah}, \cite{Hejhal}, \cite{Johnson}. 

Now, assuming we are in case $1$, consider the Abelian differential $\Lambda_{\M,\SpStr,\alpha} (\cdot,Q)\Lambda_{\M,\SpStr,\alpha'} (\cdot,Q')$. It has two simple poles at $Q$ and $Q'$, and the residues must be negatives of each other.  This leads to the relation, in any coordinate charts,  
$$
(\Lambda_{\M,\SpStr,\alpha} (\cdot,Q))_{\alpha'}(z_{\alpha'}(Q'))=-(\Lambda_{\M,\SpStr,\alpha'} (\cdot,Q'))_{\alpha}(z_\alpha(Q)).
$$
In particular, it follows that the Szeg\H{o} kernel also behaves as a $\SpStr$-spinor with respect to its second variable, the position of the pole $Q$. Thus, abusing notation as explained in the end of Section \ref{sec: Riemann surfaces}, we will write the Szeg\H{o} kernel simply as $\Lambda_{\M,\SpStr}(P,Q)$. The above anti-symmetry relation then becomes 
$$
\Lambda_{\M,\SpStr}(P,Q)=-\Lambda_{\M,\SpStr}(Q,P).
$$
\begin{exa}
    On the Riemann sphere $\hat{\C}$, there's just one spin line bundle, and its Szeg\H{o} kernel is given by $\frac{1}{P-Q}dP^\frac12 dQ^\frac12$, meaning that, its element in the chart $\C$ is given by $\frac{1}{z-w}$. On the torus $\C/(\Z+\nu\Z)$, there are four spin line bundles. One of them is odd and has a non-trivial holomorphic section, given by $dP^\frac12$, and no Szeg\H{o} kernel. The other three are even, and the elements of their Szeg\H{o} kernels are, in terms of Jacobi elliptic functions,
    $$
    2K\cdot\cs(2K(z-w),k),\quad 2K\cdot\ds(2K(z-w),k),\quad 2K\cdot\ns(2K(z-w),k),
    $$
    where $k$ is the elliptic modulus and $K$ is the complete elliptic integral of the first kind, see \cite[Chapters 19, 22]{DLMF}. 
\end{exa}

\subsection{Riemann theta functions and Hejhal--Fay identity} 

\label{sec:hejhal-Fay}
As the final ingredient of Hejhal--Fay identity, we need to recall some basic facts involving Riemann theta functions. Given a symmetric $g\times g$ matrix $\tau$, we denote by $Q_\tau$ the associated quadratic form multiplied by $\frac14$ (a choice that will prove convenient later), i.e., 
\begin{equation}
\label{eq: qtau}
   Q_\tau(m)=\frac14\sum_{i,j=1}^g\tau_{ij}m_{i}m_{j},\quad m\in \C^g. 
\end{equation}

We also use the dot product $m\cdot n=\sum_{i=1}^gm_i n_i$.
\begin{defn}
The theta function of $g$ complex variables associated with any symmetric matrix $\tau$ with negative definite real part is given by
\begin{equation}\label{eq:thetadef}
\theta_\tau:\C^{g}\rightarrow\C,\quad\theta_\tau(z)=\sum_{m\in\Z^{g}}\exp\left(4Q_\tau(m)+2m\cdot z\right), \quad z=(z_1,\dots,z_g).
\end{equation}
\end{defn}
One readily checks that $\theta$ is an entire function of $z$. For us, $\tau=\tau_\M$ will be the period matrix of a marked Riemann surface \eqref{eq:tau}, which is symmetric and has a strictly negative real part by Proposition \ref{prop: Riemann_relations}.

Furthermore $\theta$ has the following quasiperiodical properties (for a proof, see e.g. \cite[Chapter VI.1.2]{FK}, though note the difference in the conventions): for $j,k=1,...,g$
\begin{equation}\label{eq:quasi}
    \theta_\tau(z+\pi \i e_j)=\theta_\tau(z) \qquad \text{and} \qquad \theta_\tau(z+\tau_k)=e^{-2z_k-\tau_{kk}}\theta_\tau(z),
\end{equation}
where $e_j$ denotes the $j$th column of the $g\times g$ identity matrix and $\tau_k$ denotes the $k$th column of the period matrix $\tau$. 

\begin{defn}
The theta function with characteristic 
$$
H=\begin{pmatrix}
    \mu\\
    \nu 
\end{pmatrix}=\begin{pmatrix}
    \mu_1 & \dots & \mu_g\\
    \nu_1 & \dots & \nu_g
\end{pmatrix}\in \R^{2g}
$$ is given by
\begin{equation}
\label{eq: thetachar}
\theta_\tau(z;H)=\sum_{m\in\Z^{g}}\exp\left(4Q_\tau\left(m+\nu\right)+\left(m+\nu\right)\cdot (2z+2\pi \i \mu)\right).
\end{equation}
\end{defn}

We will be concerned with half-integer characteristics, $\mu=\frac12 M$ and $\nu=\frac12 N$ with $M_i,N_i\in \{0, 1\}$. Such characteristics are in a natural correspondence $H\mapsto\SpStr(H)$ with spin line bundles on a marked Riemann surface $\M$, see \cite[Section 1]{Fay} and Lemma \ref{lem: wind} below, which we now explain in the case of \emph{non-singular} characteristics $H$, i.e., such that $\theta_{\tau_\M}(\mathcal{U}(\cdot);H)$, where $\mathcal{U}$ is the Abel map as in Definition \ref{def:normdiffs}, and is not identically zero. Recall that a \emph{divisor} on a Riemann surface $\M$ is a formal linear combination of finitely many points of $\M$ with integer coefficients. In particular, the divisor $\mathrm{Div}(h)$ of a meromorphic function (or section of a line bundle, e.g., a spinor, etc.) $h$ on $\mathcal{M}$ is a formal sum of its zeros minus the formal sum of its poles, with multiplicitites. For two divisors $\Delta_{1,2},$ we write $\Delta_1\simeq\Delta_2$ if they differ by a divisor of a meromorphic function on $\mathcal{M}$. In these terms, the relation between $H$ and $\mathfrak{c}(H)$ goes as follows \cite[Section 1, especially Theorem 1.1.]{Fay}: there exists a divisor $\text{\ensuremath{\Delta}}_{H}$ (of
degree $g-1$) such that 
\begin{itemize}
\item A meromorphic spinor $h$ is a section of $\mathfrak{c}(H)$ if and only if $\mathrm{Div}(h)\simeq \Delta_H$,
\item One has, for any $Q\in\mathcal{M}$, $\Delta_H\simeq \mathrm{Div}(\theta_{\tau_M}(\Ab(\cdot)-\Ab(Q);H))-Q$. 
\end{itemize}
We note that $H$ is said to be even or odd depending on whether $\sum_{k=1}^{g}M_{k}N_{k}$
is even or odd, respectively, and this corresponds to odd and even spin line bundles mentioned above.

We now have all the relevant concepts to state the Hejhal--Fay bosonization identity \cite[Theorem 32 and Section XI]{Hejhal} or \cite[Corollary 2.12]{Fay}:

\begin{thm}
\label{thm: Hejhal}
Let $\M$ be a marked compact Riemann surface with genus
$g\geq1$, $H$ a half-integer characteristic, and $\SpStr=\SpStr(H)$ the corresponding spin line bundle. Assume that $\theta_\tau(0;H)\neq 0$. Then, the Szeg\H{o} kernel $\Lambda_{\M,\SpStr}$ exists, and

\begin{equation}
\Lambda_{\M,\SpStr}(P,Q)^{2}=\beta_\M(P,Q)+\sum_{i,j=1}^{g}\frac{\partial_{z_i}\partial_{z_j}\theta_{\tau_\M}(0;H)}{\theta_{\tau_\M}(0;H)}u_{\M,i}(P)u_{\M,j}(Q)
\label{eq: bosonization}.
\end{equation}
\end{thm}
We remark here that the existence of such an identity with \emph{some coefficients} $a_{ij}=a_{ji}$ in the sum (depending only on $\SpStr$ and $\tau_\M$), follows easily from the fact that $\Lambda_{\M,\SpStr}(P,Q)^{2}-\beta(P,Q)$ is a holomorphic differential in each variable, symmetric in the exchange of the variables. The fact that correspondence between $H$ and $\SpStr(H)$ as described above indeed holds for Szeg\H{o} kernels featuring in Theorem \ref{thm: Hejhal} follows from \cite[eq. 103]{Hejhal} (or \cite[page 12]{Fay}), or from examining the proofs. 

We will need the following explicit correspondence between $H$ and $\SpStr(H)$, whose proof we postpone to Section \ref{sec: limit_right}:
\begin{lem}
\label{lem: wind}
The spin structure $\SpStr(H)$ associated to a non-singular characteristics $H$ assigns the following winding to the simple loops representing the classes $A_i$ and $B_i$: 
$$
\wind(A_i)/2\pi=1-N_i\mod 2,\quad \wind(B_i)/2\pi=1-M_i\mod 2.
$$
\end{lem}

\section{The surface \texorpdfstring{$\Surfeps$}{} and pinching the handles}
\label{sec: pinching}
Throughout this section, we fix a domain $\Omega$, its boundary conditions, distinct points $v_1,\dots,v_n\in\Omega$, and a choice $\Op_{v_i}\in\{\sigma_{v_i},\mu_{v_i}\}$. From now on, we assume the domain $\Omega$ to be a circular domain, which does not lose generality because of conformal covariance of the correlations.

\subsection{Gluing Riemann surfaces and the surface \texorpdfstring{$\Surfeps$}.}
\label{sec: gluing}
A key step in the proof of Theorem \ref{th:ds} for a $(g+1)$-connected domain $\Omega$ will be provided by taking a limit of (\ref{eq: bosonization}) on a compact Riemann surface $\Surfeps$ of genus $g+n+k$ as $n+k$ handles degenerate at a rate controlled by 
a small parameter $\eps$. 
Here $n$ is the number of spins or disorders and $2k$ is the number of points where boundary conditions change from free to fixed. In this section, we describe the construction of this surface.

We start with a standard procedure of gluing together two boundary components of a (not necessarily connected) Riemann surface with a boundary, see \cite[Chapter 2.2]{SS} or \cite[Chapter 11-3]{Cohn} or \cite[Chapter II.4.5]{FK} or \cite[Chapter 3]{Fay}. Let $\gamma_1$ and $\gamma_2$ be two distinct connected components of $\pa \M$ homeomorphic to circles. Let $\mathcal{U}_{1,2}$ be (small, disjoint) neighborhoods of $\gamma_{1,2}$ respectively, and let $z_{i}:\mathcal{U}_{i}\to \C,$ $i=1,2$ be coordinate maps such that $z_1(\mathcal{U}_1\setminus \gamma_1)\subset \D$, $z_2(\mathcal{U}_2\setminus \gamma_2)\subset \C\setminus \D$, and $z_i(\gamma_i)=\pa \D$. Such coordinate maps can be constructed e.g. by uniformizing annular neighborhoods of $\gamma_{1,2}$; then the uniformization maps $z_{1,2}$ extend continuously to a bijection of $\gamma_{1,2}$ to $\pa\D$ respectively. We construct the topological surface $\M_{\gamma_1\leftrightarrow\gamma_2}$ by identifying each point $w\in \gamma_1$ with the unique point $w'\in\gamma_2$ such that $z_1(w)=z_2(w')$, and endow it with a conformal structure by adding to the original atlas the chart $\mathcal{U}=\mathcal{U}_1\cup\mathcal{U}_2$ endowed with the coordinate map 
$$z(w)=\begin{cases}z_1(w),&w\in \mathcal{U}_1;\\
z_2(w),&w\in \mathcal{U}_2.
\end{cases}
$$
The resulting Riemann surface $\M_{\gamma_1\leftrightarrow\gamma_2}$ depends on the choice of the local coordinates\footnote{Actually, only on the welding homeomorphism $z_2\circ z_1^{-1}:\gamma_1\to\gamma_2$.} $z_{1,2}$, but this will not be an issue as we will always describe this choice concretely. It is clear that several pairs of boundary components can be glued in that way, in any order.

A convenient version of the above construction is \emph{gluing along straight cuts}. Let $\hat{z}_1,\hat{z}_2$ be local coordinates in neighborhoods $\mathcal{U}_{1,2}$ of $\gamma_{1,2}$, respectively, that map $\gamma_{1,2}$ to the segment $[-1,1]$. Then, we can glue $\gamma_1$ with $\gamma_2$ by identifying the upper sides of the two cuts $[-1,1]$ together using the map $x\mapsto -x$, and similarly for the lower sides. Formally, let $\mathcal{Z}_{\mathrm{in},\mathrm{out}}(w)=w\pm\sqrt{w^2-1}$ be the two branches of the inverse Zhoukowski map, mapping $\C\setminus [-1,1]$ conformally to $\D$ and $\C\setminus\overline{\D}$ respectively, and apply the above construction to $z_1=\mathcal{Z}_{\mathrm{out}}\circ \hat{z}_1$ and $z_2=-\mathcal{Z}_{\mathrm{in}}\circ \hat{z}_2$. 

A particular case of the gluing construction is the \emph{Schottky double} $\Surf$ of a circular domain $\Dom$. Let $\Domr$ be a copy of $\Dom$, and take $\M$ to be the disjoint union $\Dom\sqcup\Domr$. We endow $\M$ with a conformal structure by taking the coordinate maps to be all conformal maps from subsets of $\Dom$ to $\C$ and all \emph{anti-holomorphic} maps from subsets of $\Domr$ to $\C$, i.e., maps $z:\Domr\to \C$ such that $\overline{z}$ is holomorphic. The Schottky double $\Surf$ is now obtained by applying the above construction to each pair of the corresponding components in $\pa \Dom$ and $\pa \Domr$, taking $z_1$ and $z_2$, respectively, to be the identity map and the inversion with respect to the corresponding component (i.e., $z_2(w)=\frac{r^2}{\bar{w}-\bar{c}}+c$ for a component that is a circle of radius $r$ with centre $c$). This way, $\Surf$ is a Riemann surface (without boundary), and the mapping $u\mapsto u^\star$ that maps corresponding points on the two copies to one another is an anti-holomorphic involution. Also as mentioned, the genus of $\Surf$ is $g$ if $\Omega$ is $(g+1)$-connected; see Figure~\ref{fig:schottky}.

\begin{figure}
    \centering
\includegraphics[width=0.9\textwidth]{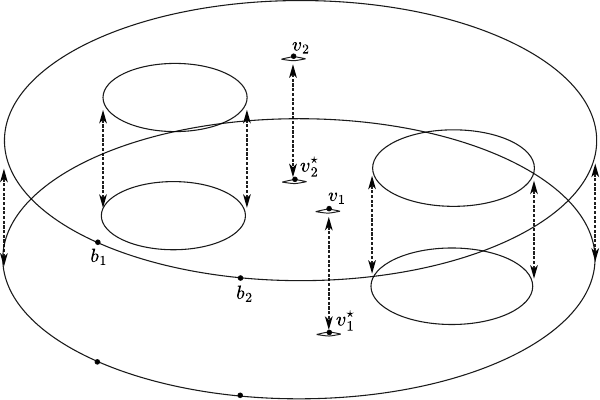}

\vspace{0.5cm}

\includegraphics[width=0.9\textwidth]{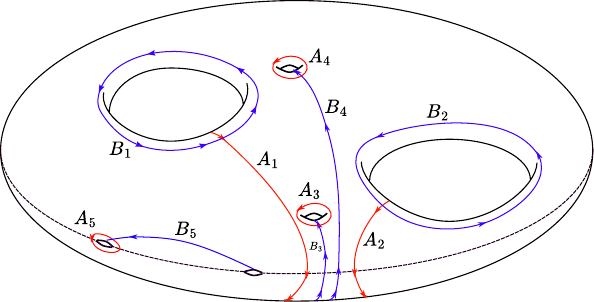}
\caption{A construction of the surface $\Surfeps$ in the case $g=2$, $n=2$, $k=1$. Two copies $\Dom_\eps, \Domr_\eps$ with small cuts near $v_{1,2}$ and $v^\star_{1,2}$ are glued together, forming small ``holes". The Torelli marking is drawn with $A$-loops in red and $B$-loops in blue. In the final stage, small cuts near $b_1,b_2$ are glued together, so that the curve $B_5$ becomes a closed loop.}   
    \label{fig:schottky}
\end{figure}

\begin{defn} We construct the surface $\Surfeps$ as follows:
\label{def: surfeps}
\begin{itemize}
    \item Starting from a circular domain $\Dom$, construct its Schottky double $\Surf$; 
    \item For each marked point $v_1,\dots,v_n\in\Omega$, make straight cuts $[v_i-\eps,v_i+\eps]$ and $[v^\star_i-\eps,v^\star_i+\eps]$ on $\Surf$ and glue them as described above, using the coordinates $\hat{z}_1(u)=\eps^{-1}(u-v_i)$ on $\Dom$ and $\hat{z}_2(u)=\eps^{-1}({u}^\star-v^\star_i)$ on $\Domr$;
    \item Let $(b_{1},b_{2})\subset \pa \Dom$ be a free boundary arc. We can choose a local coordinate $(\mathcal{U},z)$ in which $(b_{1},b_{2})$ becomes an interval in $\R$, i.e., that $(b_{1},b_{2})\subset \mathcal{U}$, $z([b_1,b_2])\subset \R$, and $z(u^\star)=\bar{z(u)}$ on $\mathcal{U}$. In this coordinate chart, cut the surface along the segments $[z(b_1)-\eps,z(b_1)+\eps]$ and $[z(b_2)-\eps,z(b_2)+\eps]$, and glue the cuts together as described above.
    \item Repeat for other free arcs $(b_3,b_4),\dots,(b_{2k-1},b_{2k})$
\end{itemize}
\end{defn}
The details of the above construction are not important; for example, the first two steps could be replaced by taking $\Dom\setminus\cup B_{\eps}(v_i)$ and then constructing its Schottky double. Of course, $\eps$ is assumed to be small enough that the cuts remain disjoint and fit into the corresponding coordinate charts.

In what follows, we assume that $\fixed \neq \emptyset$. The case when all boundary conditions are free can be treated by Kramers-Wannier duality, or by a simple modification of the choices made below. We mark a point $w_0$ of $\fixed$, and assume without loss of generality that it is on the outer boundary component of $\Omega$. We also mark points $w_1,\dots w_g\notin\{b_1,\dots b_{2k}\}$, one on each of the inner boundary components; we choose them on $\fixed$ whenever possible. 

\begin{defn}\label{def:markchar}We choose the following marking of $\Surfeps$:  
\label{def: marking}
\begin{enumerate}
\item the loop $A_i$, $i=1,\dots,g$, runs from $w_i$ to $w_0$ in $\Dom$ and then back in $\Domr$. The loop $B_i$ runs counterclockwise along the boundary of $i$-th inner boundary component of $\Omega$.
\item the loop $A_{g+i}$, $i=1,\dots,n$, is a simple loop in $\Dom$ surrounding the cut at $v_i$ counterclockwise. The loop $B_{g+i}$, $i=1,\dots,n$ runs from $w_0$ to the cut at $v_i$ in $\Omega$, and then back to $w_0$ in $\Domr$.
\item the loop $A_{g+n+i}$, $i=1,\dots,k$, is a (small) simple loop on $\Surfeps$ surrounding the cut at $b_{2i-1}$ in counterclockwise order, or, equivalently, surrounding the cut at $b_{2i}$ in clockwise order. The loop $B_{g+n+i}$ is represented by a simple curve connecting $b_{2i}$ and $b_{2i-1}$ in $\Omega$.
\end{enumerate}
To this marked surface, we associate a half-integer theta characteristic $H$ as follows: 
\begin{enumerate}
\item for $i=1,\dots,g$, we set $M_i=N_i=0$, unless $w_i\in\free$, in which case we set $M_i=0$, $N_i=1$.
\item for $i=g+1,\dots,g+n$, we set $N_i=1$ and $M_i=0$ (respectively, $M_i=1$) if $\Op_{v_{i-g}}=\sigma_{v_{i-g}}$ (respectively, $\Op_{v_{i-g}}=\mu_{v_{i-g}}$).
\item for $i=g+n+1,\dots,g+n+k$, we set $M_i=0$ and $N_i=1$.
\end{enumerate}
\end{defn}

\subsection{The Szeg\H{o} kernel on \texorpdfstring{$\Surfeps$}{}}
\label{sec: Szego on surfeps}

It will be convenient to describe the Szeg\H{o} kernel on $\Surfeps$ using its ``elements'' in the coordinate charts $\Dom_\eps:=\Omega\setminus \cup_{i=1}^n[v_i-\eps;v_i+\eps]$ and $\Domr_\eps=\Domr\setminus \cup_{i=1}^n[v^\star_i-\eps;v^\star_i+\eps]$. Since these charts are not simply connected, the ``elements'' themselves will be $(\frac12,\frac12)$--forms, which however can be identified with two-valued \emph{functions} by ``dividing by $dz^\frac12$ or $d\bar{z}^\frac12$". Also, since the charts $\Dom_\eps,$ $\Domr_\eps$ do not form an open cover of $\Surfeps$, there will be \emph{boundary conditions} that ensure that the two functions together indeed give rise to a $(\frac12,\frac12)$--form. The discussion is similar to \cite[Section I]{Hejhal}.

Thus, let $\SKDom(z,w)$ be the element $(\Lambda_{\Surfeps,\SpStr,\alpha})_\alpha$ of the Szeg\H{o} kernel for some simply connected chart $\mathcal{U_\alpha}$, where $\mathcal{U}_\alpha\subset\Omega_\eps$ and $z_\alpha$ is the identity map; recall that the spin structure $\SpStr=\SpStr(H)$ is the one associated via Lemma \ref{lem: wind} to the characteristic $H$ as per Definition \ref{def: marking}. We continue it analytically to a function of two variables on a double cover of $\Omega_\eps$ whose two values differ by a sign between sheets. The resulting function is analytic for $z,w$ not in a fiber of the same point, anti-symmetric, and has a simple pole of residue $1$ at $z=w$. 

Similarly, let $\SKDr(z,w)=(\Lambda_{\Surfeps,\SpStr,\alpha'})_{\alpha}$ be the element of $\Lambda_{\Surfeps,\SpStr}$ for two symmetric simply connected charts  $\mathcal{U}_\alpha\subset\Dom_\eps$ and $\mathcal{U}_{\alpha'}\subset \Domr_\eps$, with the coordinate maps $z_\alpha: z\mapsto z$ and $z_{\alpha'}: w\mapsto \bar{w}$. It can be continued analytically in $z$ and anti-analytically in $w$ to a two-valued function on $\Omega_\eps$. Similar constructions can be made when $z,w\in\Domr_\eps$ and $z\in \Domr_\eps$, $w\in\Dom_\eps$; the symmetries of our setup and the uniqueness of Szeg\H{o} kernel ensure that they result in $\overline{\SKDom}$ and $\overline{\SKDr}$, respectively. We note that while $\SKDom$ is uniquely determined by the expansion as $z\to w$, $\SKDr$ is only defined up to a global sign. The whole construction can be informally summarized as follows:
$$
\Lambda_{\Surfeps,\SpStr}(z,w)=\begin{cases}
    \SKDom(z,w)\,dz^\frac12 dw^\frac12,& z\in \Dom_\eps,w\in\Dom_\eps;\\
    \SKDr(z,w)\,dz^\frac12 d\bar{w}^\frac12,& z\in \Dom_\eps,w\in\Domr_\eps;\\
    \overline{\SKDr(z,w)}\,d\bar{z}^\frac12 dw^\frac12,& z\in \Domr_\eps, w\in \Dom_\eps;\\
    \overline{\SKDom(z,w)}\,d\bar{z}^\frac12 d\bar{w}^\frac12,& z\in\Domr_\eps,w\in\Domr_\eps.
\end{cases}
$$

We now note these four elements must analytically continue into each other in an appropriate sense: again informally, we must have $\SKDom(z,w)dw^\frac12=\SKDr(z,w)d\bar{w}^\frac12$ when $w\in \pa\Dom_\eps$, and $\SKDom(z,w)dz^\frac12=\overline{\SKDr(z,w)}d\bar{z}^\frac12$ when $z\in \pa\Dom_\eps$. A precise statement is given by the following lemma. Let the set $\fixed^\dagger\subset\Dom_\eps$ (respectively $\free^\dagger\subset\Dom_\eps$) comprise the original $\fixed$ (respectively, $\free$) part of the boundary of $\Omega$ with the cuts around endpoints \emph{excluded}, and the cuts $[v_i-\eps,v_i+\eps]$ with $\Op_{v_i}=\sigma_{v_i}$ (respectively, $\Op_{v_i}=\mu_{v_i}$) \emph{included}. 
\begin{lem}
\label{lem: Szego_elements}
   With the choice of the half-integer characteristic $H$ as in Definition \ref{def:markchar} and the corresponding spin line bundle $\SpStr(H)$, the two-valued functions $\SKDom(z,w)$ and $\SKDr(z,w)$ pick up a $-1$ sign along the loops $A_{g+i}$, $i=1,\dots,n$ but not along $B_i$, $i=1,\dots,g$. Moreover, we have the following boundary conditions for $\SKDom,\SKDr$:
    \begin{eqnarray}
        \tau_w\SKDom(z,w)&=\SKDr(z,w),\quad w\in\fixed^\dagger\label{eq: Szego_bc1} \\ 
        \tau_w\SKDom(z,w)&=-\SKDr(z,w),\quad w\in\free^\dagger\label{eq: Szego_bc2}\\ 
        \tau_z\SKDom(z,w)&=\overline{\SKDr(z,w)},\quad z\in\fixed^\dagger\label{eq: Szego_bc3}\\
        \tau_z\SKDom(z,w)&=-\overline{\SKDr(z,w)},\quad z\in\free^\dagger\label{eq: Szego_bc4}
    \end{eqnarray}
 where $\tau_z\in\C$ denotes the positively oriented unit tangent to $\Omega$ at $z$.
\end{lem}
\begin{proof}
We first explain why (\ref{eq: Szego_bc1}--\ref{eq: Szego_bc4}) hold up to signs (which are of course constant along each free or fixed boundary arc). Let $w_0\in\fixed^\dagger$, and let $z_{1,2}$ be coordinate maps near $w_0$ in the gluing construction described in Section \ref{sec: gluing}. Let $\tilde{K}(z,w)$ denote the element of $\Lambda_{\Surfeps,\SpStr}$ in the $w$-chart $\mathcal{U}$ (coming from the gluing), and some arbitrary $z$-chart. We have
$$(z_1'(w_0))^{-\frac12}\SKDom(z,w_0)=\tilde{K}(z,z_1(w_0))=\tilde{K}(z,z_2(w_0))=\left(\overline{z_2'(w_0)}\right)^{-\frac12}\SKDr(z,w_0),$$ since the coordinate maps between $z_1$ (resp. $\overline{z_2}$) are the relevant coordinate change maps. Hence $\SKDom(z,w_0)=\left(z'_1(w_0)/\overline{z_2'(w_0)}\right)^\frac12\SKDr(z,w_0).$ It remains to note that the pre-factor is $\pm \tau^{-1}_z$, proving (\ref{eq: Szego_bc1}) up to sign. The proof of (\ref{eq: Szego_bc2}--\ref{eq: Szego_bc4}) up to sign is similar. 

The rest of the proof amounts to showing that the spin line bundle $\SpStr(H)$, with our choice of $H$, indeed corresponds to the choice of signs as above.  By Lemma \ref{lem: wind}, we have $N_{g+i}=1$ implies $\wind_{\SpStr}{A_{g+i}}=0$, $i=1,\dots,n$, and working in the chart $\Dom_\eps$, we see that it means that $\SKDom(z,w)$ must pick up a $-1$ sign as $z$ traces $A_i$, and respectively for $B_i$, $i=1,\dots,g$. We can ensure that \eqref{eq: Szego_bc1} holds at $w_0$, by a choice of a global sign for $\SKDr$. Let $\gamma$ be a path connecting $w_0$ to one of $w_i$, $1\leq i\leq g$ in $\Dom_\eps$, so that $\gamma\cup\tilde{\gamma}\simeq A_i$ is a smooth simple loop in $\Surfeps$; we assume that $\gamma$ avoids $z$ and the zeros of $\SKDom(z,\cdot), \SKDr(z,\cdot)$. Then, $\wind_{\SpStr(H)}(\gamma\cup\tilde{\gamma})=\wind_{\SpStr(H)}(\gamma)+\wind_{\SpStr(H)}(\tilde{\gamma})$, where the two terms compute the rotation number of $\overline{\SKDom(z,\cdot)}^2$ (respectively, $\overline{\SKDr(z,\overline{\cdot})}^2$) with respect to the tangent vector of $\gamma$ (respectively, $\overline{\gamma}$ traversed backwards); we use the coordinate map $z\mapsto \bar{z}$ for the second term. But since $\SKDom(\cdot,w)$ is a \emph{single-valued function} in a simply connected neighborhood of $\gamma$, we have 
$$
e^{-i\wind_{\SpStr(H)}(\gamma)/2}=\frac{\SKDom(z,w_i)}{\SKDom(z,w_0)}e^{i\wind(\gamma)/2},
$$
where $\wind(\gamma)$ is simply the winding with respect to the flat metric in $\C$. A similar argument applies to $\wind_{\SpStr(H)}(\tilde{\gamma})$, and using $\wind(\overline{\gamma})=-\wind(\gamma)$, we conclude 
$$
e^{-i\wind_{\SpStr(H)}(\gamma\cup\tilde{\gamma})}=\frac{\SKDom(z,w_i)\SKDr(z,w_0)}{\SKDom(z,w_0)\SKDr(z,w_i)}e^{i\wind(\gamma)}=-\frac{\SKDom(z,w_i)\SKDr(z,w_0)}{\SKDom(z,w_0)\SKDr(z,w_i)}\frac{\tau_{w_i}}{\tau_{w_0}}.
$$
Now it follows that \eqref{eq: Szego_bc1} or \eqref{eq: Szego_bc2} holds on the boundary arc containing $w_i$, by combining Lemma \ref{lem: wind} with our choice of the characteristics in Definition \ref{def: marking}. Exactly the same argument applied to $B_{g+i}$ loops, $i=1,\dots,g$, yields \eqref{eq: Szego_bc1} or \eqref{eq: Szego_bc2} on each of the cuts at $v_i$. For a component containing several arcs, we have \eqref{eq: Szego_bc1} on the ``wired" arc $(b_{2k},b_{2j-1})\cap \fixed^\dagger\ni w_i$; applying the above argument to the loops $A_{g+n+j}$ and $A_{g+n+k}$, we get \eqref{eq: Szego_bc2} on the neighboring free arcs $(b_{2k-1},b_{2k})\cap \free^\dagger$ and $(b_{2j-1},b_{2j})\cap \free^\dagger$, and similarly propagate it further. The equations (\ref{eq: Szego_bc3}), (\ref{eq: Szego_bc4}) follow by anti-symmetry.
\end{proof}

\subsection{Limits of the Abelian differentials and the period matrix on \texorpdfstring{$\Surfeps$}{}.}\label{sec:shrink}
Another important ingredient in the proof of Theorem \ref{th:ds} will be taking the limits of Abelian differentials and the period matrix as $\eps\to 0$.

\begin{defn}
\label{def: u_j_all}
        We introduce the following unifying notation for the Abelian differentials of the \emph{first and third kind} on the Schottky double $\Surf$ of the domain $\Dom$: 
    \begin{align}
     u_j(z)&:=u_{\Surf,j}(z),& 1\leq j\leq g,\\
     u_{g+j}(z)&:=\frac{1}{2}\omega_{\Surf,v_{j},v^\star_{j}}(z),& j=1, \dots, n,\\
     u_{g+n+j}(z)&:=\frac{1}{2}\omega_{\Surf,b_{2j-1},b_{2j}}(z),& j=1, \dots, k.
    \end{align}
\end{defn}

\begin{lem}\label{lem:abdifasy} \cite[Proposition 3.7 and Corollary 3.8]{Fay} The Abelian differentials on $\Surfeps$ have the following limits as $\eps \rightarrow 0$:
\label{lem: limit_Abelian}
\begin{align}
u_{\Surfeps,j}(P)&\longrightarrow u_j(P), \quad j=1,\dots, g+n+k,\label{eq: conv_first}\\
\beta_{\Surfeps}(P,Q)&\longrightarrow \beta_{\Surf}(P,Q)\label{eq: conv_second}.
\end{align}
\end{lem}

\begin{rem}
To make sense of this convergence statement 
about Abelian differentials on different surfaces, note that  $\Surfeps,\Surf$ are made up of the same charts $\Dom,\Domr$, except near the cuts like $[v_i-\eps, v_i+\eps]$. Thus, we can evaluate both sides of \eqref{eq: conv_first} and \eqref{eq: conv_second} in the same chart. The limits are in fact locally uniform.
\end{rem}

We also have an expansion for the period matrix:

\begin{lem}
\label{lem: limit_period}
    \cite[Corollary 3.8]{Fay} As $\eps \rightarrow 0$, the period matrix $\tau_\eps=\tau_{\Surfeps}$ of $\Surfeps$ has the following expansion:

\begin{equation}
    \tau_\eps=\frac{1}{2}\log{\eps} \cdot \left(0_{g \times g} \oplus I_{(n+k) \times (n+k)}\right)+\tau_0+O(\eps)
\end{equation}
where the symmetric matrix $\tau_0$ is given by
    \begin{align*}
    (\tau_0)_{ij}&=(\tau_{\Surf})_{ij}, & 1\leq i,j\leq g\\
    (\tau_0)_{jj}&= c_j, & j=g+1,\dots,g+n+k\\
    (\tau_0)_{ij}&= \int^{e_{j-g}}_{e^\star_{j-g}}u_i, & j=g+1,\dots, g+n+k ,\quad 1\leq i< j .
    \end{align*}
    Here $\tau_{\Surf}$ is the period matrix of $\Surf$, $c_j$ are some constants, and 
    $$
    (e_j,e^\star_j)=\begin{cases}(v_{j},v^\star_{j}),&1\leq j\leq n\\
    (b_{2(j-n)-1},b_{2(j-n)}),&n+1\leq j \leq n+k. 
    \end{cases} 
    $$ 
\end{lem}

\begin{rem} In \cite[Prop. 3.7, Cor. 3.8]{Fay}, the result of Lemmas \ref{lem:abdifasy} -- \ref{lem: limit_period} are stated and proven  only for a single pinched handle. But they imply the result for simultaneous pinching of multiple handles as follows: introducing a separate  parameter $\eps_1,\dots,\eps_{n+k}$ for each cut, and denoting $\vec{\eps}=(\eps_1,\dots,\eps_{n+k})$,  the proofs in \cite{Fay} proceed by showing that $u_{\Surf_{\vec{\eps}},j}$ (respectively, $(\tau_{\vec{\eps}})_{ii}-\frac12 \log\eps_{i-g}$ for $i=g+1,\dots,g+n$) are holomorphic in each $\eps_i$ near the origin. But then, by Hartogs' theorem on separate holomorphicity \cite{Fuks}, they are jointly holomorphic in $\eps_1,\dots,\eps_{n+k}$, and those can be taken to zero in any order or simultaneously, producing the same result. We remark here that the arguments below can be alternatively carried out by sending $\eps_i$ to zero sequentially; we choose a single $\eps$ for notational convenience only.
\end{rem}

\subsection{The limit of the right-hand side of (\ref{eq: bosonization}) and the bosonic correlations.} 
\label{sec: limit_right}

In this section, we will obtain a limit of the right-hand side of \eqref{eq: bosonization} and find an interpretation in terms of bosonic correlation functions. Recall from  Definition \ref{def:markchar} the vectors $M,N\in\{0,1\}^{g+n+k}$ that fix the theta characteristics $H$. For $1\leq p\leq g+n+k$, we introduce the relation 
$s_p=2m_p+N_p$ between the parameters $m_p$ and $s_p$. We denote 
\begin{align*}
\setS&:=(2\Z+N_1)\times\dots\times(2\Z+N_{g+n+k})\\
&\; =(2\Z+N_1)\times\dots\times(2\Z+N_{g})\times(2\Z+1)^{n+k}\\
\setSn&:=(2\Z+N_1)\times\dots\times(2\Z+N_{g})\times\{\pm 1\}^{n+k},
\end{align*}
so that, as $m$ runs through $\Z^{g+n+k}$, $s$ runs through $\setS$.  

\begin{lem}
\label{limitlemma}

If $\ccor{\hOp_{v_1}\dots\hOp_{v_n}}\neq 0$, then $\theta_{\tau_\eps}(0;H)\neq 0$ for $\eps$ small enough, and in particular, the identity (\refeq{eq: bosonization}) is valid. As $\eps\to 0$, the right-hand side of this identity converges, locally uniformly, to 

\begin{align}
\beta_{\Surf}(P,Q)+ \sum_{p,q=1}^{g+n+k}\frac{\sum_{s\in\setSn} s_p s_q\exp (Q_{\tau_0}(s)+(\pi \i/2) (s\cdot M)))}{\sum_{s\in\setSn} \exp (Q_{\tau_0}(s)+(\pi \i/2) (s\cdot M)))} u_p(P) u_q(Q).
\label{eq: RHS_limit}
\end{align}
where $u_p$ and $\tau_0$ are as given in Definition \ref{def: u_j_all} and Lemma \ref{lem: limit_period}, and the quadratic form $Q_{\tau_0}$ is given by \eqref{eq: qtau}. 
\end{lem}

\begin{proof}
The claim that $\theta_{\tau_\eps}(0;H)\neq 0$ will be justified in the course of the proof of the current lemma and Lemma \ref{lem: gausscorr}, as follows. We will see below that the leading term of the asymptotics of  $\theta_{\tau_{\eps}}(0;H)$ is proportional to the denominator $\sum_{s\in\setSn} \exp (Q_{\tau_0}(s)+(\pi \i/2) (s\cdot M))$. In \eqref{eq: exp_sines_cosines} below, we identify the latter quantity with $\ccor{\hOp_{v_1}\dots\hOp_{v_n}}$, up to a non-zero factor. Thus $\sum_{s\in\setSn} \exp (Q_{\tau_0}(s)+(\pi \i/2) (s\cdot M))\neq 0$ whenever $\ccor{\hOp_{v_1}\dots\hOp_{v_n}}\neq 0$, and so $\theta_{\tau_{\eps}}(0;H)\neq 0$ for $\eps$ small enough.

Due to Lemma \ref{lem: limit_Abelian}, we have $\beta_{\Surfeps}\to \beta_{\Surf}$ and $u_{\Surfeps,j}\to u_j$. Therefore, we only need to handle the coefficients $\pa_{z_p}\pa_{z_q}\theta_{\tau_\eps}(0;H)/\theta_{\tau_\eps}(0;H)$. Differentiating the definition \eqref{eq: thetachar} of the theta-function, one gets

\begin{equation}
\label{eq: theta ratio}
\frac{\pa_{z_p}\pa_{z_q}\theta_{\tau_\eps}(0;H)}{\theta_{\tau_\eps}(0;H)}=\frac{\sum_{s\in\setS} s_p s_q\exp (Q_{\tau_\eps}(s)+(\pi \i/2) (s\cdot M))}{\sum_{s\in\setS} \exp (Q_{\tau_\eps}(s)+(\pi \i/2)(s\cdot M))},
\end{equation}
where $M$ is as in Definition \ref{def:markchar}. We now apply Lemma \ref{lem: limit_period} to obtain
\begin{equation}
\label{eq: expA}
\exp(Q_{\tau_\eps}(s))=\eps^{d(s)}\exp\left(Q_{\tau_0+O(\eps)}(s)\right).
\end{equation}
where 
\begin{equation}
\label{eq: d(s)}
d(s)=\sum_{p=g+1}^{g+n+k}\frac{s_p^2}{8}.
\end{equation}
Let $d_0=\frac{n+k}{8}$ be the minimal possible value of $d(s)$. We can write 
\[
\frac{\pa_{z_p}\pa_{z_q}\theta_{\tau_\eps}(0;H)}{\theta_{\tau_\eps}(0;H)}=\frac{\sum_{s\in\setS} \eps^{d(s)-d_0}s_p s_q\exp (Q_{\tau_0+O(\eps)}(s)+(\pi \i/2) (s\cdot M))}{\sum_{s\in\setS} \eps^{d(s)-d_0}\exp (Q_{\tau_0+O(\eps)}(s)+(\pi \i/2)(s\cdot M))}
\]
Note that $d(s)=d_0$ if and only if $s\in\setSn$. If we formally take the limit $\eps\to 0$ in the numerator and the denominator, terms with $d(s)\neq d_0$ vanish, leading to \eqref{eq: RHS_limit}. Since the sums are infinite, we have to be a bit more careful. Observe that there exist constants $\alpha,\beta>0$ such that if $d(s)\neq d_0$, then $d(s)-d_0\geq \alpha +\beta d(s)$: indeed, there are only finitely many $d(s)$ with $d(s)<2(d_0+1)$, so we can choose $\alpha<1,\beta<\frac12$ small enough so that the inequality holds for those $s$, but then also $d(s)-d_0\geq 1 +\frac12 d(s)\geq \alpha+\beta d(s)$ whenever $d(s)\geq 2(d_0+1)$. This implies that 
$$
\eps^{d(s)-d_0}\exp (Q_{\tau_0+O(\eps)}(s))\leq \eps^\alpha \exp (Q_{\rho_\eps +\tau_0+O(\eps)}(s)), 
$$
where the matrix $\rho_\eps$ is given by $\frac{\beta}{2}(\log\eps)\cdot \left(0_{g\times g}\oplus I_{(n+k)\times (n+k)}\right)$. Recall also that the block $(\tau_0)_{1\leq i\leq g, 1\leq j\leq g}$ is strictly negative definite, being the period matrix of $\Surf$. It is not hard to see that this implies that for $\eps$ small enough, $\rho_\eps +\tau_0+O(\eps)$ is strictly negative definite, and moreover bounded from above by a fixed strictly negative definite matrix $T$. Therefore, we can write 
\[
\left|\sum_{s\in\setS \setminus \setSn} \eps^{d(s)-d_0}s_p s_q\exp (Q_{\tau_0+O(\eps)}(s)+(\pi \i/2)(s\cdot M))\right|\leq \eps^\alpha\sum_{s\in\setS \setminus \setSn} |s_p| |s_q|\exp (Q_{T}(s))\stackrel{\eps\to 0}{\longrightarrow} 0,
\]
the last sum being finite, and similarly for the denominator. Since, as we explained above, the limit of the denominator is non-zero, the proof is complete.
\end{proof}

In order to re-interpret the right-hand side of (\ref{eq: RHS_limit}), we need to recall the following (well known) relation of Abelian integrals on $\Surf$ with harmonic functions on $\Dom$. Recall the definitions from Section \ref{sec: Riemann surfaces} and Sections \ref{sec: instanton}, \ref{sec: bosonic}.
 
\begin{lem}
\label{lem: harmgreen}
We have  
\begin{align}
\label{eq: green_Abelian}G_\Dom(z,w)&=-\frac{1}{2}\int_{z^\star}^{z} \omega_{\Surf,w,w^\star}&&=-\re \int_{w_0}^z\omega_{\Surf,w,w^\star},\\
\label{eq: harm_Abelian} h_p(z)&=-\frac{1}{\pi \i}\int_{z^\star}^{z}u_{\Surf,p}&&=-\frac{2}{\pi}\im  \int_{w_0}^z u_{\Surf,p},\\
\hat h_p(z)&=-\frac{1}{2\pi \i}\int_{ z^\star}^{z}\omega_{\Surf, b_{2p-1},b_{2p}}&&=-\frac{1}{\pi}\im\int_{w_0}^z\omega_{\Surf, b_{2p-1},b_{2p}}, \label{eq: harm_free}
\end{align}
where  the first integral in \eqref{eq: green_Abelian} (respectively, in \eqref{eq: harm_Abelian}, in \eqref{eq: harm_free} is along any symmetric path connecting $z^\star$ and $z$ on $\Surf$, (respectively, any symmetric path avoiding $p$-th boundary component, $p$-th boundary arc), while the second integral is along any path connecting $w_0$ and $z$ in $\Omega$.
\end{lem}
\begin{proof}
 Consider the involution $\inv:z\mapsto z^\star$ on $\Surf$, and denote by $\inv^* \omega$ the pull-back of an Abelian differential $\omega$. Then, $\inv^* \omega$ is an anti-meromorphic $1$-form, and hence $\overline{\inv^* \omega}$ is again an Abelian differential. Its residues and periods can be readily computed using that $\int_{\gamma\circ\inv} \overline{\inv^* \omega}=\overline{\int_{\gamma\circ\inv} \inv^* \omega}=\overline{\int_\gamma\omega}$ for any curve $\gamma$. We have $A_i\circ \inv=-A_i$, $\gamma_{b_{2k-1}}\circ \inv=-\gamma_{b_{2k-1}}$, and $\gamma_w\circ \inv = -\gamma_{w^\star},$ where $\gamma_z$ denotes a small positively oriented loop around $z$. Combining this with the uniqueness of of Abelian differentials, we conclude that 
$$
\overline{\inv^* u_{\Surf,p}}=u_{\Surf,p}, \quad \overline{\inv^*\omega_{\Surf, b_{2p-1},b_{2p}}}=\omega_{\Surf, b_{2p-1},b_{2p}}, \quad \overline{\inv^*\omega_{\Surf, w,w^\star}}=-\omega_{\Surf, w,w^\star}.
$$
Consider a path  $\gamma$ from $w_0$ to $z$ in $\Omega$, concatenated to its symmetric counterpart $-\gamma\circ \inv$ from $z^\star$ to $w_0$, avoiding $w$. The above symmetries ensure that 
\begin{align*}
\frac12 \int_{-\gamma\circ \inv+\gamma}\omega_{\Surf,w,w^\star}&=\frac12 \int_{\gamma}\omega_{\Surf,w,w^\star}+\frac12 \int_{\gamma\circ \inv}\overline{\inv^\star\omega_{\Surf,w,w^\star}}&&=\re \int_{\gamma}\omega_{\Surf,w,w^\star},\\
\frac{1}{2\i}\int_{-\gamma\circ \inv+\gamma}u_{\Surf,p}&=\frac{1}{2\i}\int_{\gamma}u_{\Surf,p}-\frac{1}{2\i}\int_{\gamma\circ \inv}\overline{\inv^\star u_{\Surf,p}}&&=\im \int_{\gamma}u_{\Surf,p},\\
\frac{1}{2\i}\int_{-\gamma\circ \inv+\gamma}\omega_{\Surf, b_{2p-1},b_{2p}}&=\frac{1}{2\i}\int_{\gamma}\omega_{\Surf, b_{2p-1},b_{2p}}-\frac{1}{2\i}\int_{\gamma\circ \inv}\overline{\inv^\star\omega_{\Surf, b_{2p-1},b_{2p}}}&&=\im \int_{\gamma}\omega_{\Surf, b_{2p-1},b_{2p}}.
\end{align*}
This proves the second identities in (\ref{eq: green_Abelian}--\ref{eq: harm_free}), at least for the contours of integration as chosen. We now note that the right-hand sides $H_1,H_2,H_3$ of (\ref{eq: green_Abelian}--\ref{eq: harm_free}) are real-valued and harmonic in $z$ (except for $H_1$ at $z=w$). They are also locally constant on $\pa\Dom$, except for $H_3$ at $b_{2p-1},b_{2p}$, since if we take a piece of $\gamma$ run along $\pa\Omega$, then the contribution of that piece cancels with the corresponding piece of $\gamma\circ\inv$. It follows that $H_1,H_2,H_3$ are single-valued. Their value at the boundary component (arc) containing $w_0$ is clearly $0$. Their values at other boundary components are given by $A$--periods of $\omega_{\Surf,w,w^\star}$, $u_{\Surf,p}$, $\omega_{\Surf, b_{2p-1},b_{2p}}$, which are all zeros, except for $\int_{A_p}u_{\Surf,p}=\pi \i$. This in particular implies that $H_1,H_2,H_3$ are independent of the choice of the contours of integration in (\ref{eq: green_Abelian}--\ref{eq: harm_free}), subject to the stated restrictions.

Note that if $z$ is on the $p$-th inner boundary component, we can choose $\gamma$ so that $A_p=-\gamma+\gamma\circ\inv$. Hence $H_2\equiv 1$ on the $p$-th inner boundary component while $H_2\equiv 0$ elsewhere on the boundary, and \eqref{eq: harm_Abelian} follows. Using the chart $\Omega$, we can calculate $H_1(z)=-\re\int^z\left(\frac{dz}{z-w}+O(1)\right)=-\log|z-w|+O(1)$ as $z\to w$. Since $G(\cdot,w)$ is the only harmonic function in $\Omega\setminus \{w\}$ satisfying that expansion and vanishing on the boundary, \eqref{eq: green_Abelian} follows. Finally, near $b=b_{2p-1}$, respectively $b_{2p}$, we have 
$$H_3(z)=\pm\frac{1}{\pi}\im\int_{w_0}^z\left(\frac{dz}{z-b}+c_1+o(1)\right)=\pm\frac{1}{\pi}\im \log(z-b)+c_2+o(1),$$ hence the harmonic function $H_3$ is bounded and jumps by $\pm 1$ at $b=b_{2p-1}$ and $b_{2p}$ respectively, so that $H_3\equiv 1$ on the arc $(b_{2p-1},b_{2p})$, proving \eqref{eq: harm_free}.  
\end{proof}

We will also need the following identity:
\begin{lem}
\label{lem: Green_sec_kind}
    We have, for $z,w\in \Omega,$ 
    \begin{align}
    \pa_z\pa_w G_\Omega(z,w)dzdw=-\frac12 \beta_{\Surf,w}(z),
    \end{align}
    or, in other words, $\pa_z\pa_w G(z,w)=-\frac12(\beta_{\Surf,w})_\alpha(z)$ in the chart $\mathcal{U}_\alpha=\Dom$ with the identity coordinate map.
\end{lem}
\begin{proof}
We have $\pa_zG_\Omega(z,w)\,dz=-\frac12\omega_{\Surf,w,w^\star}(z)$ by \eqref{eq: green_Abelian}, and differentiating the right-hand side in $w$ gives $-\frac12 \beta_{\Surf,w}(z)$ by \eqref{eq: omega_beta}.
\end{proof}
We will now compute the coefficients of the forms $Q_\Omega, B_\Omega, \hat{Q}_\Omega$ featuring in the definition of the instanton component, in terms of the period matrix and Abelian integrals on $\Surf$. 

\begin{lem}
\label{lem: prob_instanton}

We have the following identities:

\begin{align}
\langle\nabla h_i,\nabla h_j\rangle_{\Dom}&=-\frac{2}{\pi}(\tau_{\Surf})_{ij}\label{eq: Qform_h_h}\\
\langle\nabla h_i,\nabla \hat h_j\rangle_{\Dom}&=-\frac{2}{\pi}\int_{b_{2j}}^{b_{2j-1}}u_{\Surf,i}\label{eq: Qform_h_hat}\\
\langle\nabla \hat h_i,\nabla \hat h_j\rangle_{\Dom}&=-\frac{1}{\pi}\int_{b_{2j}}^{b_{2j-1}}\omega_{\Surf,b_{2i-1},b_{2i}},\quad i\neq j\label{eq: Qform_hat_hat}
\end{align}
 
\end{lem}
\begin{proof}
    We start with \eqref{eq: Qform_h_h}. Integrating by parts we obtain
    \begin{align}
    \label{eq: by_parts}
    \langle\nabla h_i,\nabla h_j\rangle_{\Dom}=\iint_{\Dom}h_j\Delta h_i|dz|^2+\int_{\pa\Dom}h_j\pa_n h_i |dz|=\int_{B_j}\pa_n h_i|dz|,
    \end{align}
where we have used that $h_i$ is harmonic and $h_j\equiv 0$ on $\pa\Omega\setminus B_j$, $h_j\equiv 1$ on $B_j$. Now, if $h=\re{f}$ is a harmonic function, where $f$ is analytic, we have by Cauchy--Riemann equations $\pa_n h-\i\pa_{\i n}h=(\pa_z f)n_z$, viewing the unit normal vector $n_z$ as a complex number. If in addition $h$ is locally constant on $\pa \Omega$, we get, along the boundary, $\pa_n h|dz|=(\pa_z f)n_z|dz|=\i\pa_z f dz=\i df$. Combining this observation with \eqref{eq: harm_Abelian}, we conclude that 
    \begin{equation}
    \int_{\Gamma_j}\pa_n h_i |dz|=-\frac{2}{\pi}\int_{\Gamma_j}u_{\Surf,i}=-\frac{2}{\pi}(\tau_{\Surf})_{ij}.
    \end{equation}
The proofs of \eqref{eq: Qform_h_hat} and \eqref{eq: Qform_hat_hat} are  similar: instead of \eqref{eq: by_parts}, we get 
 \begin{align}
    \langle\nabla h_i,\nabla \hat{h}_j\rangle_{\Dom}&=\int_{b_{2j}}^{b_{2j-1}}\pa_n h_i|dz|,\\
    \langle\nabla \hat{h}_i,\nabla \hat{h}_j\rangle_{\Dom}&=\int_{b_{2j}}^{b_{2j-1}}\pa_n \hat{h}_i|dz|,
 \end{align}
which gives \eqref{eq: Qform_h_hat} and \eqref{eq: Qform_hat_hat} after using \eqref{eq: harm_Abelian} and \eqref{eq: harm_free} respectively.
\end{proof}

We summarize the above results in the following Lemma:
\begin{lem}
We have the following interpretation for the entries of the matrix $\tau_0$ featuring in Lemma \ref{lem: limit_period}:
\label{lem: tau_0_harmonic}
\begin{align*}
-(\tau_{0})_{ij}=\begin{cases}
\frac{\pi}{2}\langle\nabla h_i,\nabla h_j\rangle,& 1\leq i,j\leq g;
\\
\frac{\pi}{2}\langle\nabla h_i,\nabla \hat{h}_{j-g-n}\rangle,& 1\leq i\leq g,\quad  g+n+1\leq j\leq g+n+k;
\\
\frac{\pi}{2}\langle\nabla \hat{h}_{i-g-n},\nabla \hat{h}_{j-g-n}\rangle,& g+n+1\leq i,j\leq g+n+k,\quad i\neq j ;
\\ 
2G_\Omega(v_{i-g},v_{j-g})& g+1\leq i,j\leq g+n,\quad i\neq j
\\
\pi \i h_{i}(v_{j-g}), &1\leq i\leq g,\quad  g+1\leq j\leq g+n;
\\
\pi \i \hat{h}_{i-g-n}(v_{j-g})&g+n+1\leq i\leq g+n+k,\quad  g+1\leq j\leq g+n;
\end{cases}
\end{align*}
\end{lem}
\begin{proof}
The proof boils down to recalling the definition of $\tau_0$ as in Lemma \ref{lem: limit_period}, and of the involved holomorphic differentials $u_1,\dots,u_{g+n+k}$ as defined in Definition \ref{def: u_j_all}, then applying Lemma \ref{lem: prob_instanton} in the first three cases and Lemma \ref{lem: harmgreen} in the last three cases.
\end{proof}

\begin{lem}

Fix $s\in \setSn$, and let $\xi$ be the corresponding value of the instanton component defined by \eqref{eq: instanton}, where we put $\hat{s}_j=s_{g+n+j}$ for $j=1,\dots,k$. Then, we have the following interpretation of the terms within the sum in \eqref{eq: RHS_limit} in terms of the bosonic field $\Phi=\varphi+\xi$:
\begin{multline}
\label{eq: expD}
\exp (Q_{\tau_0}(s)+(\pi \i/2) (s\cdot M)))
\\=C\cdot \mathbb{P}(\xi) \prod_{i=1}^{n} \exp \left( -\frac{\i s_{g+i}}{\sqrt{2}}\xi(v_{i})\right)\ccor{\prod_{i=1}^{n} \:\nexp{-\frac{\i s_{g+i}}{\sqrt{2}}}{v_i}}
\prod_{i:\Op_{v_i}=\mu_{v_i}} s_{g+i},
\end{multline}
where a constant $C\neq 0$ does not depend on $s$ (but may depend on $\Omega$ and $v_1,\dots,v_n$).
\end{lem}
\begin{proof}

First, note that we have 
\[\exp((\pi \i/2) (s\cdot M))=\prod_{i:\Op_{v_i}=\mu_{v_i}} (\i s_{i+g})
\]
since due to Definition \ref{def: marking}, $M_{g+i}=1$ if $\Op_{v_i}=\mu_{v_i}$ for $i=1,\dots,n$, and $M_i=0$ otherwise. This gives the last term in the product, after absorbing the product of $\i$'s into $C$.

The other terms come from breaking down the quadratic form 
\begin{equation}
Q_{\tau_0}(s)=\frac14\sum_{i,j=1}^{g+n+k}(\tau_0)_{ij}s_is_j
\label{eq: Q_tau_zero}
\end{equation} into various parts, and using Lemma \ref{lem: tau_0_harmonic}: 
\begin{itemize}[leftmargin=*]
\item the diagonal terms $(\tau_0)_{ii}s^2_i=(\tau_0)_{ii}$ for $g+1\leq i\leq g+n+k$ do not depend on $s$, and their exponentials can be absorbed into $C$.
\item Let $\mathcal{I}_1\subset\{1,\dots,g+n+k\}^2$ denote the set of pairs $(i,j)$ of indices such that either $(i,j)$, or $(j,i)$ fits into the first three cases in Lemma \ref{lem: tau_0_harmonic}. We then have
$$
\frac14\sum_{(i,j)\in \mathcal{I}_1}(\tau_0)_{ij}s_is_j=Q_\Omega(s)+B_\Omega (s,\hat{s})+ \hat{Q}_\Omega(\hat{s},\hat{s}),
$$
so that 
$$
\exp\left(\frac14 \sum_{(i,j)\in \mathcal{I}_1}(\tau_0)_{ij}s_is_j\right)=Z\mathbb{P}(\xi).
$$ 
\item when $g+1\leq i\neq j\leq g+n$, we have $(\tau_0)_{ij}=-2G(v_i,v_j)$ and 
\begin{multline}
\exp\left(\frac14 \sum_{g+1\leq i\neq j\leq g+n}(\tau_0)_{ij}s_is_j\right)=\exp\left(- \sum_{g+1\leq i\neq j\leq g+n}\frac{s_i s_j}{2}G(v_i,v_j)\right)\\=\ccor{\nexp{\frac{\i s_1}{\sqrt{2}}}{v_1}\dots\nexp{\frac{\i s_1}{\sqrt{2}}}{v_n}}\exp\left(\sum_{i=1}^n\frac{1}{4}g_\Omega(z_i,z_i)\right),
\end{multline}
according to \eqref{eq: defnexp}. The last factor does not depend on $s$ and can be absorbed into $C$.
\item Let $\mathcal{I}_2\subset\{1,\dots,g+n+k\}^2$ denote the set of pairs $(i,j)$ that fit into the last two cases in Lemma \ref{lem: tau_0_harmonic}. Then, for every $j$ such that $g+1\leq j\leq g+n$, we have 
$$\sum_{i:(i,j)\in\mathcal{I}_2}(\tau_0)_{ij}s_is_j=-\pi 
\i s_j\left(\sum_{i=1}^g s_i h_i(v_{j-g})+\sum_{i=1}^{k}\hat{s}_i\hat{h}_i(v_{j-g})\right)=-\sqrt{2}\i \xi(v_{j-g}),$$
therefore 
$$
\exp\left(\frac14\sum_{(i,j)\in\mathcal{I}_2\text{ or }(j,i)\in\mathcal{I}_2}(\tau_0)_{ij}s_is_j\right)=\exp\left(\sum_{j=g+1}^{j+n}-\frac{\i s_j}{\sqrt{2}}\xi(v_{j-g})\right).
$$
\end{itemize}
These options exhaust the entire range of summation in \eqref{eq: Q_tau_zero}, and putting everything together concludes the proof. 
\end{proof}

\begin{lem}
\label{lem: gausscorr}
The expression \eqref{eq: RHS_limit}, evaluated at $P=z$ and $Q=w$ in the chart $\Omega$ with the identity coordinate map, is equal to
\begin{align}
\label{gausscorr}
\frac{\ccor{\sqrt{2} \i\bdry\Phi(z)\sqrt{2} \i\bdry\Phi(w) \hOp_{v_1}\dots\hOp_{v_n}}}{\ccor{\hOp_{v_1}\dots\hOp_{v_n}}}.
\end{align}
\end{lem}
\begin{proof}

The identification proceeds by breaking the range of summation in \eqref{eq: RHS_limit} into three parts. Let $I_1:=\{1,\dots,g\}$, $I_2:=\{g+1,\dots,g+n\}$ and $I_3:=\{g+n+1,\dots,g+n+k\}$. Our goal is to check the identities, hereinafter evaluating all differentials in the chart $\Omega$:
\begin{align}
\beta(z,w)+\sum_{(p,q)\in I_2\times I_2} \alpha_{p,q}u_p(z)u_q(w)&=-2\frac{\ccor{\bdry\varphi(z)\bdry\varphi(w) \hOp_{v_1}\dots\hOp_{v_n}}}{\ccor{\hOp_{v_1}\dots\hOp_{v_n}}};\label{eq: J_1}
\\
\sum_{(p,q)\in I_2\times (I_1\cup I_3)} \alpha_{p,q}u_p(z)u_q(w)&=-2\frac{\ccor{\pa\varphi(z)\pa\xi(w)\hOp_{v_1}\dots\hOp_{v_n}}}{\ccor{\hOp_{v_1}\dots\hOp_{v_n}}};
\label{eq: J_2}
\\
\sum_{(p,q)\in (I_1\cup I_3)\times (I_1\cup I_3)}\alpha_{p,q}u_p(z)u_q(w)&=-2\frac{\ccor{\bdry\xi(z)\bdry\xi(w) \hOp_{v_1}\dots\hOp_{v_n}}}{\ccor{\hOp_{v_1}\dots\hOp_{v_n}}},
\label{eq: J_3}
\end{align}
where $\alpha_{p,q}=\alpha_{q,p}$ are the coefficients in \eqref{eq: RHS_limit} given by 
$$
\alpha_{p,q}=\frac{\sum_{s\in\setSn} s_p s_q\exp (Q_{\tau_0}(s)+(\pi \i/2) (s\cdot M))}{\sum_{s\in\setSn} \exp (Q_{\tau_0}(s)+(\pi \i/2) (s\cdot M))}.
$$
The Lemma follows by summing \eqref{eq: J_1}--\eqref{eq: J_3} and the symmetric counterpart of \eqref{eq: J_2} with roles of $p,q$ exchanged, using that $\Phi=\varphi+\xi,$ and the linearity of correlations.

We start by rewriting $\alpha_{p,q}$, first by leaving $s_{g+1},\dots,s_{g+n}$ fixed and summing over the other indices. Note that the tuples of indices $s_1,\dots,s_g,\hat{s}_1,\dots\hat{s}_k$, where $\hat{s}_i=s_{g+n+i},$ are in one-to-one correspondence with realizations of $\xi$. Because of the factor $\mathbb{P}(\xi)$ in \eqref{eq: expD}, summing over those amounts to taking an expectation with respect to $\xi$. Therefore, we have 
\begin{multline}
\label{eq: a_p_q}
\alpha_{p,q}=\frac{\sum_{\nu\in \{\pm 1\}^n}\sum_{s\in\setSn:s_{g+1}=\nu_1,\dots,s_{g+n}=\nu_n} s_p s_q\exp (Q_{\tau_0}(s)+(\pi \i/2) (s\cdot M))}{\sum_{\nu\in \{\pm 1\}^n}\sum_{s\in\setSn:s_{g+1}=\nu_1,\dots,s_{g+n}=\nu_n} \exp (Q_{\tau_0}(s)+(\pi \i/2) (s\cdot M))}
\\
=\frac{\sum_{\nu\in \{\pm 1\}^n}\ccor{s_p s_q\Nexp{-\frac{\i\nu_1}{\sqrt{2}}}{v_1}\dots \Nexp{-\frac{\i\nu_n}{\sqrt{2}}}{v_n}}\prod_{i:\Op_{v_i}=\mu_{v_i}} \nu_i}{\sum_{\nu\in \{\pm 1\}^n}\ccor{\Nexp{-\frac{\i\nu_1}{\sqrt{2}}}{v_1}\dots \Nexp{-\frac{\i\nu_n}{\sqrt{2}}}{v_n}}\prod_{i:\Op_{v_i}=\mu_{v_i}} \nu_i}.
\end{multline}
Note that in the numerator, if $p\in\{g+1,\dots,g+n\}$, then $s_p=\nu_{p-g}$ can be taken out of the correlation, otherwise $s_p$ is thought of as a random variable which is a function of $\xi$. 

We can directly relate the denominator above to the denominator in \eqref{gausscorr}: plugging in the expressions for sine and cosine in terms of exponentials, and expanding everything by linearity yields
\begin{equation}
\label{eq: exp_sines_cosines}
\ccor{\hOp_{v_1}\dots\hOp_{v_n}}=\sum_{ \nu\in \{\pm 1\}^n}  2^{-n}\ccor{\Nexp{-\frac{\i\nu_1}{\sqrt{2}}}{v_1}\dots \Nexp{-\frac{\i\nu_n}{\sqrt{2}}}{v_n}}\prod_{i:\Op_i=\mu_i}(\i \nu_i ).
\end{equation}

We are ready to prove \eqref{eq: J_3}. Recalling \eqref{eq: instanton}, we can write
\begin{equation}
\pa \xi(z)=\frac{\pi}{\sqrt{2}}\left(\sum_{i=1}^{g}s_i\pa h_i(z)+\sum_{i=1}^k \hat{s}_i\pa \hat{h}_{i}(z)\right)=-\frac{1}{\sqrt{2}\i}\sum_{p\in I_1\cup I_3}s_p u_p(z).
\label{eq: xi_sum}
\end{equation}
where we have used Lemma \ref{lem: harmgreen} and Definition \ref{def: u_j_all}. Multiplying \eqref{eq: a_p_q} by $u_p(z)u_q(w)$ and summing yields 
$$
\sum_{(p,q)\in\mathcal(I_1\cup I_3)\times (I_1\cup I_3)}\alpha_{p,q}u_p(z)u_q(w)=-2\frac{\sum_{\nu\in \{\pm 1\}^n}\ccor{\pa \xi(z) \pa \xi(w)\prod^n_{i=1}\Nexp{-\frac{\i\nu_i}{\sqrt{2}}}{v_i}}\prod_{i:\Op_{v_i}=\mu_{v_i}} \nu_i}{\sum_{\nu\in \{\pm 1\}^n}\ccor{\prod^n_{i=1}\Nexp{-\frac{\i\nu_i}{\sqrt{2}}}{v_i}}\prod_{i:\Op_{v_i}=\mu_{v_i}} \nu_i},
$$
and simplifying the sums as in \eqref{eq: exp_sines_cosines} yields \eqref{eq: J_3}. 

For \eqref{eq: J_1} and \eqref{eq: J_2}, we will need the ``Gaussian integration by parts" formulae:
\begin{align}
\ccor{\pa\varphi(z)\prod_{i=1}^n\nexp{\gamma_i}{v_i}}&=\sum_{p=1}^n\gamma_p\pa_zG(z,v_p)\ccor{\prod_{i=1}^n\nexp{\gamma_i}{v_i}},
\end{align}
\begin{multline}
\ccor{\pa\varphi(z)\pa\varphi(w)\prod_{i=1}^n\nexp{\gamma_i}{v_i}}\\=\left(\pa_z\pa_w G(z,w)
+\sum_{p,q=1}^n\gamma_p\gamma_q\pa_z G(z,v_p)\pa_w G(w,v_q)\right)\ccor{\prod_{i=1}^n\nexp{\gamma_i}{v_i}},
\end{multline}
which follow by a simple computation from Definition \ref{def: bos_corr}. We multiply these identities by $\mathbb{E}(\pa\xi(w)\prod_{i=1}^ne^{\gamma_i\xi_i})$  and $\mathbb{E}(\prod_{i=1}^ne^{\gamma_i\xi_i})$ respectively, specialize to $\gamma_i=-\frac{\i \nu_i}{\sqrt{2}}$, and plug in $\pa_z\pa_w G(z,w)=-\frac12 \beta(z,w)$, from Lemma \ref{lem: Green_sec_kind} and $\pa_z G(z,v_p)=-u_{p+g}(z)$, from combining Lemma \ref{lem: harmgreen} and Definition \ref{def: u_j_all}. This leads to  
\begin{multline}
\ccor{\pa\varphi(z)\pa\xi(w)\prod_{i=1}^n\Nexp{-\frac{\i\nu_i}{\sqrt{2}}}{v_i}}=\sum_{p\in I_2}\frac{\i\nu_{p-g}}{\sqrt{2}}u_{p-g}(z)\ccor{\pa\xi(w)\prod_{i=1}^n\Nexp{-\frac{\i\nu_i}{\sqrt{2}}}{v_i}}\\
=-\frac12 \sum_{p\in I_2}\sum_{q\in I_1\cup I_3}u_{p}(z)u_q(w)\ccor{\nu_{p-g}s_q\prod_{i=1}^n\Nexp{-\frac{\i\nu_i}{\sqrt{2}}}{v_i}},
\end{multline}
where in the last equality we have used \eqref{eq: xi_sum}, so
\begin{multline}
\ccor{\pa\varphi(z)\pa\varphi(w)\prod_{i=1}^n\Nexp{-\frac{\i\nu_i}{\sqrt{2}}}{v_i}}\\=-\frac12 \beta(z,w)-\frac12\sum_{p,q\in I_2}\nu_{p-g}\nu_{q-g}u_p(z)u_q(z)\ccor{\prod_{i=1}^n\Nexp{-\frac{\i\nu_i}{\sqrt{2}}}{v_i}}.
\end{multline}
Hence, multiplying \eqref{eq: a_p_q} by $u_p(z)u_q(w)$ and summing over the respective ranges (and then simplifying sums over $\nu$'s into correlations involving sines and cosines as before) yields \eqref{eq: J_2} and \eqref{eq: J_1}.
\end{proof}

With the help of the last Lemma, we are ready to prove Lemma \ref{lem: wind}.
\begin{proof}[Proof of Lemma \ref{lem: wind}]
Note that if $H,H'$ are two non-singular characteristics, then, for any $Q\in \mathcal{M}$, 
$$
g_{H,H'}(\cdot,Q):=\frac{\theta_\tau(\Ab(\cdot)-\Ab(Q),H)}{\theta_\tau(\Ab(\cdot)-\Ab(Q),H')}
$$
is a two-valued meromorphic \emph{function} which realizes the co-cycle $s$ discussed in Section \ref{sec: spinors_Szego}, namely, its multiplicative monodromies are $s(A_i)=(-1)^{N_i-N_i'}$ and $s(B_i)=(-1)^{M_i-M_i'}$. Now, if $h_{H'}$ is a meromorphic section of $\mathfrak{c}(H')$, then, from comparing the divisors, we get that $g_{H,H'}(\cdot,Q)h_{H'}(\cdot)$ is a meromorphic section of $\mathfrak{c}(H).$ In particular, if follows that for each $\mathcal{M}$ and each $A$ or $B$ loop, it suffices to prove the Lemma for one characteristic $H$. 

Let $H$ be an odd non-singular characteristics for $\mathcal{M}$ and $h_H$ be the explicitly constructed holomorphic section of $\SpStr_H$, as in \cite[Definition 2.1]{Fay}. This section manifestly depends continuously on $\mathcal{M}$ (as long as $H$ is non-singular, which is an open condition), and thus the windings $\wind_{\SpStr(H)}(A_i)$, $\wind_{\SpStr(H)}(B_i)$ computed using $h_H$ are locally constant. It follows that the set of Torelli marked surfaces for which the conclusion of the Lemma (for a given $A$ or $B$ loop) holds is both closed and open (in the Teichm\"uller space). 

We now let $\mathcal{M}=\Surfeps$, where $\Omega=\HH$ with $n$ distinct marked points $v_1,\dots,v_n$ and the choice of $H$ corresponding to $\Op_{v_i}=\sigma_{v_i}$ for all $i$, and fixed boundary conditions on $\pa\HH$, i.e., $N_i=1$ and $M_i=0$ for all $i$. Combining Theorem \ref{thm: Hejhal} and Lemma \ref{lem: gausscorr}, we see that 
$$\SKDom(z,w)\stackrel{\eps\to 0}{\longrightarrow} \left(\frac{\ccor{\sqrt{2} \i\bdry\varphi(z)\sqrt{2} \i\bdry\varphi(w) \hOp_{v_1}\dots\hOp_{v_n}}}{\ccor{\hOp_{v_1}\dots\hOp_{v_n}}}\right)^\frac12,$$
where, as in Section \ref{sec: Szego on surfeps}, $\SKDom(z,w)$ is the element of $\Lambda_{\Surfeps,\SpStr}(z,w)$  in the chart $\HH$. From the OPE \eqref{eq: fuse_pa_cos}, we see that as $z\to v_i$, the expression in brackets behaves as $c_i(z-v_i)^{-1}$, where the coefficients $c_i$ are generically non-zero. Recalling that $A_i$ can be represented by a small loop in $\HH$ encircling $v_i$, it follows that for $\eps$ small enough, we have, on $\Surfeps$, $\wind_{\SpStr(H)}(A_i)=0\mod 4\pi$. By connectedness of the Teichm\"uller space, this proves the lemma for $A$ loops, for arbitrary surface of genus $n$. To complete the proof, note that by \cite[Eq. 12]{Fay}, the change of homology basis $A_i\mapsto B_i$, $B_i\mapsto -A_i$ transforms $\Delta_H$ into $\Delta_{H'}$, where $M'_i=N_i$ and $N'_i=M_i$. 
\end{proof}

\subsection{The limit of the left-hand side of (\ref{eq: bosonization}) and the Ising correlations.}
\label{szegosection}

\begin{lem}
If $\ccor{\hOp_{v_1}\dots\hOp_{v_n}}\neq 0$ and $\ccor{\Op_{v_1}\dots\Op_{v_n}}_\Dom\neq 0$, then we have the following limit as $\eps\to 0$:
\label{lem: LHS_limit}
\begin{equation}
\label{limitlambda}
\lim_{\eps\to 0}\SKDom(z,w)=
    \frac12\frac{\ccor{\psi_z\psi_w \Op_{v_1}\dots\Op_{v_n}}_\Dom}{\ccor{\Op_{v_1}\dots\Op_{v_n}}_\Dom}.
\end{equation}
for any $z,w\in \Domp$, $z\neq w$.
\end{lem}
\begin{rem}
In the case $k=0$, i.e., every boundary component carries either free or fixed boundary conditions, we can identify $\SKDom$ with the Ising correlation already before passing to a limit: 
$$
\SKDom(z,w)=\frac12\frac{\ccor{\psi_z\psi_w \Op_{\pa_1}\dots\Op_{\pa_n}}}{\ccor{\Op_{\pa_1}\dots\Op_{\pa_n}}},
$$
where $\pa_i=[v_i-\eps_i,v_i+\eps_i]$, and $\Op_{\pa_i}$ stands for the spin on the monochromatic boundary component $\pa_i$ if $\Op_{v_i}=\sigma_{v_i}$, and for a disorder placed (anywhere) on the free boundary component $\pa_i$ if $\Op_{v_i}=\mu_{v_i}$. We do not have such an interpretation when $k\neq 0$. Note also that $\SKDr(z,w)$ and its limit can be expressed in a similar manner with $\psi^\star_w$ in the place of $\psi_w$, but we will not need it.
\end{rem}

\begin{proof}
Fix distinct $w,v_1,\dots,v_n \in \Omega$. Consider, for any $\rho\in \mathbb{C},$ the linear combination 
$$
f_\eps^{[\rho]}(z)=\bar{\rho}\SKDom(z,w)+\rho \SKDr(z,w).
$$

We know that as $\eps\to 0$, we have $f_\eps^{[\rho]}(\cdot)\to f^{[\rho]}(\cdot)$ uniformly on compact subsets of $\bar{\Omega}\setminus \{v_1,\dots,v_n,w\}$, where 
$$
f^{[\rho]}(z)=\bar{\rho}\SKDomLim(z,w)+\rho\SKDrLim(z,w),
$$
and $\SKDomLim=\lim_{\eps\to 0}\SKDom$, $\SKDrLim=\lim_{\eps\to 0}\SKDr$; the limits exist by Lemma \ref{limitlemma}. We will prove the lemma by checking that $f^{[\rho]}(z)$ solves the same boundary value problem as 
$$\frac{\ccor{\psi_z\psi^\rho_w \Op_{v_1}\dots\Op_{v_n}}_\Dom}{\ccor{\Op_{v_1}\dots\Op_{v_n}}_\Dom}=\bar{\rho}\frac{\ccor{\psi_z\psi_w \Op_{v_1}\dots\Op_{v_n}}_\Dom}{\ccor{\Op_{v_1}\dots\Op_{v_n}}_\Dom}+\rho\frac{\ccor{\psi_z\psi^\star_w \Op_{v_1}\dots\Op_{v_n}}_\Dom}{\ccor{\Op_{v_1}\dots\Op_{v_n}}_\Dom},$$
cf. \cite[Def. 3.3., Prop. 5.18 and Thm. 6.1--6.2]{CHI2} and that the solution to that boundary value problem is unique.

Clearly, $z\mapsto f^{[\rho]}(z)$ is a holomorphic two-valued function in $\Dom\setminus\{v_1,...,v_n,w\}$ with monodromy $-1$ around $v_1,\dots,v_n$ and $1$ around other boundary components, and has the same singularity at $z=w$ as each of $f_\eps^{[\rho]}$, namely, 
\begin{equation}
  f^{[\rho]}(z)=\frac{\bar{\rho}}{z-w}+O(1),\quad z\to w\label{eq: pole_t_eta}.  
\end{equation}
Let $z\in\pa\Dom\setminus \bpoints$. We also have 
\begin{align}
\tau^\frac12_z f^{[\rho]}(z)&\in \R, &z\in \pa_{\mathrm{fixed}}\Omega;\label{eq: bc_fixed}\\
\tau^\frac12_z f^{[\rho]}(z)&\in  i\R, &z\in \pa_{\mathrm{free}}\Omega \label{eq: bc_free}.
\end{align}
since by Lemma \ref{lem: Szego_elements}, we have these inclusions for $f^{[\rho]}_\eps$ whenever $\eps$ is so small that $z$ does not belong to a cut.

We turn to describing the behavior of $z\mapsto f^{[\rho]}(z)$ at $v_1,\dots, v_n$. Define 
$$
h_\eps(t)=\im\int^t_{t_0}(f_{\eps}^{[\rho]}(z))^2\,dz,
$$
a (for a moment, multi-valued) harmonic function in $\Dom$. In terms of $h_\eps$, the boundary conditions (\ref{eq: Szego_bc1}--\ref{eq: Szego_bc4}) mean that $h_\eps$ is constant along each cut $l_i=[v_i-\eps,v_i+\eps]$, and $\pa_n h_\eps\geq 0$ on $l_i$ with $\Op_{v_i}=\sigma_{v_i}$ (respectively, $\pa_n h_\eps\leq 0$ on $l_i$ with $\Op_{v_i}=\mu_{v_i}$), where $\pa_n$ is the inner normal derivative; cf. \cite[Proposition 3.7]{CHI2}.

It follows from the convergence of $f_\eps^{[\rho]}$ that $h_\eps(t)$ converges to $h(t)=\im\int^t_{t_0}(f^{[\rho]}(z))^2\,dt$, uniformly on compact subsets of $\Omega\setminus\{w,v_1,\dots,v_n\}$. Note that for a fixed small $r>0$, any $i$ such that $\Op_{v_i}=\sigma_{v_i}$ (respectively, $\Op_{v_i}=\mu_{v_i}$), and $\eps$ small enough, we have $\max_{B_r(v_i)\cap \Dom_\eps}h_\eps=\max_{\pa B_r(v_i)}h_\eps$ (respectively, $\min_{B_r(v_i)\cap \Dom_\eps}h_\eps=\min_{\pa B_r(v_i)}h_\eps$). Indeed, by the maximum principle, the maximum (resp. minimum) cannot be attained inside $B_r(v_i)\cap \Dom_\eps$, and the boundary conditions for $h_\eps$ mean that it cannot be attained on the inner component $B_r(v_i)\cap \pa\Dom_\eps$ either. Passing to the limit, we see that $h$ is harmonic and bounded in $B_r(v_i)\setminus\{v_i\}$ from above by $\max_{\pa B_r(v_i)}h$ (respectively, from below by $\min_{\pa B_r(v_i)}h$). We conclude that it has an expansion $h(t)=c\log|t-v_i|+O(1)$ with $c\geq 0$ (respectively, $c\leq 0$). Differentiating and taking the square root, we deduce that 
 \begin{align}
f^{[\rho]}(z)&=\alpha_i(z-v_i)^{-\frac12}(1+o(1)),\quad \alpha_i\in e^{\i\pi/4}\mathbb{R} &&\text{ if }\Op_{v_i}=\sigma_{v_i}\label{eq: bc_spin}\\
f^{[\rho]}(z)&=\alpha_i(z-v_i)^{-\frac12}(1+o(1)),\quad \alpha_i\in e^{- \i\pi/4}\mathbb{R} &&\text{ if }\Op_{v_i}=\mu_{v_i}.\label{eq: bc_mu}
\end{align}
As a final ingredient, we derive the boundary conditions for $f^{[\rho]}(\cdot)$ at the endpoints of free arcs, namely, if $(b_{2i-1},b_{2i})$ is a free arc, we will prove that $(f^{[\rho]})^2\,dz$ has at worst simple poles at $b_{2i-1},b_{2i}$, and 
\begin{equation}
\res_{b_{2i-1}} f^{[\rho]}(z)\psi_i(z)\,dz=-\res_{b_{2i}} f^{[\rho]}(z)\psi_i(z)\,dt,
\label{eq: bc_free_arc_full}
\end{equation}
where $\psi_i$ is given by 
$$
\psi_i(z)=\frac{\varphi_i'(z)^\frac12(\varphi_i(z)-\frac12)}{\sqrt{\varphi_i(z)(\varphi_i(z)-1)}},
$$
and $\varphi_i$ is a Möbius map that maps $[b_{2i-1},b_{2i}]$ to $[0,1]$ so that a neighborhood of $[b_{2i-1},b_{2i}]$ in $\Omega$ is mapped into the upper half plane. Note that $\psi_i$ has two well defined branches in $\mathcal{U}\setminus [b_{2i-1},b_{2i}]$, where $\mathcal{U}$ is a simply-connected neighborhood of $[b_{2i-1},b_{2i}]$; it does not matter which branch is chosen in \eqref{eq: bc_free_arc_full}, but we insist that the same branch is used in both sides of the equality. Since $\res_{b_{2i-1}}\psi^2_i(z)\,dz=-\res_{b_{2i}}\psi^2_i(z)\,dz$, the condition (\ref{eq: bc_free_arc_full}) is equivalent to $\res_{b_{2i-1}}f^{[\rho]}(z)^2\,dz=-\res_{b_{2i}}f^{[\rho]}(z)^2\,dz$ together with an additional bit of information regarding sign choices. It is easy to see that (\ref{eq: bc_free_arc_full}) is equivalent to \cite[Eq. (3.4)]{CHI2}.

The fact that $(f^{[\rho]})^2$ has at most simple poles at $b_{2i-1},b_{2i}$ follows from Lemmas \ref{limitlemma} and \ref{lem: limit_Abelian}. To prove (\ref{eq: bc_free_arc_full}), recall that in the third step of the construction of $\Surfeps$ in Section \ref{sec: gluing}, to define the cuts at $b_{2i-1},b_{2i}$ and then glue them together, we used a local coordinate map from a neighborhood $\mathcal{U}$ of $(b_{2i-1},b_{2i})$ in $\Surf$ to a neighborhood $U$ of an interval in the real line. We may assume that the interval is $[0,1]$ and the map is given by  $$\phi(z):= \begin{cases}\varphi_i(z),&z\in \Dom_\eps\\ \overline{\varphi_i(z)},&z\in \Domr_\eps.\end{cases}$$ 
Let us pull back $f^{[\rho]}_\eps(z)$ into the coordinate domain $U\cap \HH$, i.e.,  consider the spinor $\tilde{f}_{\eps}^{[\rho]}$ defined in $U\cap\HH$ by $f^{[\rho]}_\eps(z)=\varphi_i'(z)^\frac12 \tilde{f}_{\eps}^{[\rho]}(\varphi_i(z))$. Then, since $\tilde{f}_{\eps}^{[\rho]}$ can be viewed as an element of the spinor $\bar \rho \Lambda_{\Surfeps,\SpStr}(\cdot ,w)+\rho \Lambda_{\Surfeps,\SpStr}(\cdot ,w^\star)$,
 it extends to a two-valued holomorphic function on the surface $U_\eps$ obtained by gluing $U$ along straight cuts $[-\eps,\eps]$ and $[1-\eps,1+\eps]$. Due to our choice of the spin structure $\SpStr$ as per Definition \ref{def: marking} and Lemma \ref{lem: wind}, this function is ramified at the loop $\tilde{A}:=\varphi_i(A_{g+n+i})$ but not at $\tilde{B}:=\varphi_i(B_{g+n+i})$.

Consider now the Szeg\H{o} kernel $\Lambda_{\hat{\C}_{\eps},\tilde{\SpStr}}$ on the genus 1 surface $\hat{\C}_{\eps}$ obtained from the Riemann sphere $\hat{\C}$ by cutting and gluing along the same straight cuts, with the spin structure $\tilde{\SpStr}$ assigning the same winding as $\SpStr$ to the loops $\tilde{A}, \tilde{B}$, i.e., $0$ and $2\pi$. Let $\tilde{\psi}_\eps(z)$ denote the element of $\Lambda_{\hat{\C}_{\eps},\tilde{\SpStr}}(z,\infty)$ with respect to the identity coordinate for the first argument and $1/w$-coordinate for the second one.  Then, $\tilde{\psi}_\eps(z)\tilde{f}_{\eps}^{[\rho]}(z)\,dz$ is an Abelian differential on $U_\eps$, in particular, for a small fixed $r>\eps$, we have 
$$\oint_{\pa B_r(0)} \tilde{\psi}_\eps(z)\tilde{f}_{\eps}^{[\rho]}(z)\,dz=-\oint_{\pa B_r(1)} \tilde{\psi}_\eps(z)\tilde{f}_{\eps}^{[\rho]}(z)\,dz,$$ because $\pa B_{r}(0)$ and $\pa B_{r}(1)$ can be contracted to the loops traversing upper and lower parts of the respective cuts, and under our gluing, this is the same loop with opposite orientations. 

We have the asymptotics, which we prove below,
\begin{equation}
\tilde{\psi}_\eps(z)\stackrel{\eps \to 0}{\longrightarrow} \tilde{\psi}(z):= \i\frac{z-\frac12}{\sqrt{z(z-1)}}\label{eq: model_Szego_asymp}
\end{equation}
 uniformly on compact subsets of $\C\setminus\{0,1\}.$ Therefore, $\tilde{f}^{[\rho]}=\lim_{\eps\to 0} \tilde{f}_{\eps}^{[\rho]}$ 
 satisfies \begin{equation}
 \label{eq: residue_psi}
     \res_0 \left(\tilde{f}^{[\rho]}(t)\tilde{\psi}(t)\,dt\right)=-\res_1\left(\tilde{f}^{[\rho]}(t)\tilde{\psi}(t)\,dt\right),
 \end{equation}
and pulling back by  $\varphi_i$ yields (\ref{eq: bc_free_arc_full}).

Now, we know from \cite[Prop. 5.18 and Thm. 6.1--6.2]{CHI2} that if $\ccor{\Op_{v_1}\dots\Op_{v_n}}\neq 0$, then 
$$
\frac{\ccor{\psi_z\psi^\rho_w \Op_{v_1}\dots\Op_{v_n}}}{\ccor{\Op_{v_1}\dots\Op_{v_n}}}
$$
solves the same boundary value problem (\ref{eq: pole_t_eta}--\ref{eq: bc_free_arc_full}).  Therefore, to complete the proof, it suffices to show that there does not exist a non-trivial spinor satisfying (\ref{eq: bc_fixed}--\ref{eq: bc_free_arc_full}) and holomorphic in $\Dom\setminus\{v_1,\dots,v_n \}$. Assuming $g$ is such a spinor, we compute 
\begin{multline}
\bar{\rho}g(w)=\res_w (f^{[\rho]}g)=\frac{1}{2\pi \i}\mathrm{v.p.}\int_{\pa \Dom\setminus\{b_1,\dots,b_{2k}\}}f^{[\rho]}(z)g(z)\,dz\\-\sum^n_{j=1}\res_{v_j}(f^{[\rho]}g)-\frac12\sum^k_{i=1}(\res_{b_{2i-1}}(f^{[\rho]}g)+\res_{b_{2i}}(f^{[\rho]}g)).
\end{multline}
The first (respectively, the second) term is purely imaginary since $f^{[\rho]}$, $g$ both satisfy (\ref{eq: bc_fixed}--\ref{eq: bc_free}) (respectively, (\ref{eq: bc_spin}--\ref{eq: bc_mu})). To see that the last sum vanishes, observe that since  $f^{[\rho]}$ and $g$ both satisfy $(\ref{eq: bc_free_arc_full})$, so do $f^{[\rho]}+ g$. As noted above, this implies
$$
\res_{b_{2i-1}}(f^{[\rho]}+ g)^2=-\res_{b_{2i-1}}(f^{[\rho]}+ g)^2.
$$ and by expanding,  $\res_{b_{2i-1}}(f^{[\rho]}g)=-\res_{b_{2i}}(f^{[\rho]}g).$ 

We conclude that $\rho g(w)\in \i\R$ for any $\rho\in\mathbb{C}$ and $w\in \Dom$, hence $g\equiv 0$, and the proof is complete.
\end{proof}

\begin{proof}[Proof of (\ref{eq: model_Szego_asymp})]
Let's note that $\hat{\C}_\eps$  is a particular case of $\Surfeps$ in the setup of Section \ref{sec: gluing}, with $\Omega=\HH,$ no insertions, and one free boundary arc $(b_1,b_2)=(0,1)$, i.e., $g=0$, $n=0$ and $k=1$; the genus of $\hat{\C}_\eps$ is 1. We use the bosonization identity (\ref{eq: bosonization}):
\[\Lambda_{\hat{\C}_{\eps},\SpStr} (z,w)^2=\beta_{\hat{\C}_{\eps}}(z,w)+\frac{\theta_{\tau_\eps}''(0;H)}{\theta_{\tau_\eps}(0;H)}u_{\hat{\C}_{\eps}}(z)u_{\hat{\C}_{\eps}}(w)\]
and apply Lemma \ref{limitlemma}: in this case, by Definition \ref{def: marking}, we have $M_1=0$ and $N_1=1$, so that $$\frac{\sum_{s_1=\pm 1}s^2_1\exp \left(Q_{\tau_0}(s_1)+(\pi\i/2)M_1s_1\right)}{\sum_{s_1=\pm 1}\exp \left(Q_{\tau_0}(s_1)+(\pi\i/2)M_1s_1\right)}=1.$$ 
Consequently, as $\eps\to 0$, using Example \ref{exa: riemann_sphere}, we have in the chart $\C$, 
\begin{multline}
\Lambda_{\hat{\C}_{\eps},\SpStr} (z,w)^2\to \beta_{\hat{\C}}(z,w)+\frac12\omega_{\hat{\C},0,1}(z)\frac12\omega_{\hat{\C},0,1}(w)\\
=\left(\frac{1}{(z-w)^2}+\frac{1}{4z(z-1)w(w-1)}\right)dzdw.
\end{multline}
We now put $w=\infty$, using the coordinate chart $\hat\C\setminus\{0\}$, with the coordinate $w'=\frac{1}{w}$, so that $dw=-\frac{1}{(w')^2}dw'$; this gives 
$$
\lim_{\eps\to 0}\Lambda_{\hat{\C}_{\eps},\SpStr} (z,w)^2=-\left(\frac{1}{(zw'-1)^2}+\frac{1}{4z(z-1)(1-w')}\right)dzdw'.
$$
Putting $w=\infty$ corresponds to putting $w'=0$, i.e., we obtain 
$$
\lim_{\eps\to 0}\Lambda_{\hat{\C}_{\eps},\SpStr} (z,\infty)^2=-\left(1+\frac{1}{4z(z-1)}\right)dzdw'=-\frac{(z-\frac12)^2}{z(z-1)}dzdw'.
$$
Taking the square root yields \eqref{eq: model_Szego_asymp}.
\end{proof}

We observe that we have now proven \eqref{eq: for_log_der}, under the assumptions $\ccor{\hOp_{v_1}\dots\hOp_{v_n}}\neq 0$ and $\ccor{\Op_{v_1}\dots\Op_{v_n}}_\Dom\neq 0$: 
\begin{proof}[Proof of \eqref{eq: for_log_der}]
    Start with (\ref{eq: bosonization}) for $\M=\Surfeps$ and $H$ as described in Definition \ref{def: marking}, and pass to the limit as $\eps\to 0$. By Lemma \ref{lem: LHS_limit}, the left-hand side of (\ref{eq: bosonization}), evaluated in the chart $\Omega$, converges to $\frac14$ times the left-hand side of \eqref{eq: for_log_der}. The limit of the right-hand side of (\ref{eq: bosonization}) is given by \eqref{eq: RHS_limit}, which is identified with $\frac14$ times the right-hand side of \eqref{eq: for_log_der} in Lemma \ref{lem: gausscorr}.
\end{proof}

\section{Concluding the proof of Theorem \ref{th:main}.}
\label{sec: conclusion}
We now derive \eqref{eq: log_der} from \eqref{eq: for_log_der}, by applying the OPE on both sides as $w\to v_1$ and then $z\to v_1$.
\subsection{Matching the logarithmic derivatives}
\begin{lem}
\label{thm: log}
If $\ccor{\hOp_{v_1}\dots\hOp_{v_n}}\neq 0$ and $\ccor{\Op_{v_1}\dots\Op_{v_n}}\neq 0$, then \eqref{eq: log_der} holds . 
\end{lem}

\begin{proof}
We will use fusion rules on (\ref{eq: for_log_der}). We start with the Ising side. 

Let's assume we have a spin at $v_1$, i.e. $\Op_{v_1}=\sigma_{v_1}$. By properly normalizing the fusion rules from \cite[Thm. 6.2]{CHI2} , we obtain

\begin{equation}
\left(\frac{\ccor{\psi_z\psi_w\sigma_{v_1}\Op_{v_2}\dots\Op_{v_n}}}{\ccor{\sigma_{v_1}\Op_{v_2}\dots\Op_{v_n}}}\right)^2=\frac{\bar\lambda^2}{z-v_1}\left(\left(\frac{\ccor{\psi_w\mu_{v_1}\Op_{v_2}\dots\Op_{v_n}}}{\ccor{\sigma_{v_1}\Op_{v_2}\dots\Op_{v_n}}}\right)^2+O(z-v_1)\right)
\end{equation}
where $\lambda=\exp(- \i \pi/4)$. If we again use fusion rules, we get

\begin{multline}
\left(\frac{\ccor{\psi_w\mu_{v_1}\Op_{v_2}\dots\Op_{v_n}}}{\ccor{\sigma_{v_1}\Op_{v_2}\dots\Op_{v_n}}}\right)^2=\frac{\lambda^2}{w-v_1}\big(1+8\bdry_{v_{1}}\log\ccor{\sigma_{v_1}\Op_{v_2}\dots\Op_{v_n}}(w-v_1)+\\
+O(w-v_1)^2\big)
\end{multline}
The equations are similar if we instead put a disorder at $v_1$. We will just have $\bar\lambda$ instead of $\lambda$, and $\mu_{v_1}$ instead of $\sigma_{v_1}$ on the right side. Thus, if we collect all of the fusion rules, we obtain

\begin{equation}
\bdry_{v_{1}}\log\ccor{\Op_{v_1}\dots\Op_{v_n}}=\lim_{w\rightarrow v_1}\left(\lim_{z\rightarrow v_1}\left(\frac{\ccor{\psi_z\psi_w\Op_{v_1}\dots\Op_{v_n}}}{\ccor{\Op_{v_1}\dots\Op_{v_n}}}\right)^2 \frac{z-v_1}{8}-\frac{1}{8(w-v_1)}\right)
\end{equation}

And thus, we can plug in (\ref{eq: for_log_der}) to obtain

\begin{equation}
\label{log_ising}
\bdry_{v_{1}}\log\ccor{\Op_{v_1}\dots\Op_{v_n}}=-\lim_{w\rightarrow v_1}\left(\lim_{z\rightarrow v_1}\frac{\ccor{\pa\Phi(z)\pa\Phi(w)\hat{\Op}_{v_1}\dots\hat{\Op}_{v_n}}}{\ccor{\hat{\Op}_{v_1}\dots\hat{\Op}_{v_n}}} \left(z-v_1\right)+\frac{1}{8(w-v_1)}\right)
\end{equation}

Now, for the GFF side we can use \eqref{eq: fuse_pa_cos} and \eqref{eq: fuse_pa_sin} to obtain
\begin{equation}
\lim_{z\rightarrow v_1}\frac{\ccor{\pa\Phi(z)\pa\Phi(w)\hat{\Op}_{v_1}\dots\hat{\Op}_{v_n}}}{\ccor{\hat{\Op}_{v_1}\dots\hat{\Op}_{v_n}}} (z-v_1)=\frac{1}{2\sqrt{2}}\frac{\ccor{\pa\Phi(w)\hat{\Op}'_{v_1}\hat{\Op}_{v_2}\dots\hat{\Op}_{v_n}}}{\ccor{\hat{\Op}_{v_1}\dots\hat{\Op}_{v_n}}}
\end{equation}
where $$\hat{\Op}'_{v_1}=\begin{cases}:\sin(\Phi (v_{1})/\sqrt{2}):&\text{ if } \hat{\Op}_{v_1}=:\cos(\Phi (v_{1})/\sqrt{2}):\\
-:\cos(\Phi (v_{1})/\sqrt{2}):&\text{ if } \hat{\Op}_{v_1}=:\sin(\Phi (v_{1})/\sqrt{2}):.
\end{cases}$$  
Applying \eqref{eq: fuse_pa_cos} and \eqref{eq: fuse_pa_sin} one more time, we get 
\begin{multline}
\frac{\ccor{\pa\Phi(w)\hat{\Op}'_{v_1}\hat{\Op}_{v_2}\dots\hat{\Op}_{v_n}}}{\ccor{\hat{\Op}_{v_1}\dots\hat{\Op}_{v_n}}}=-\frac{1}{2\sqrt{2}(w-v_1)}\big(1+4\bdry_{v_{1}}\log\ccor{\hat\Op_{v_1}\dots\hat\Op_{v_n}}(w-v_1)\\
+O(w-v_1)^2\big)
\end{multline}
Hence, for any $\hat{\Op}_{v_1}$, we have

\begin{equation}
\label{log_gff}
\bdry_{v_{1}}\log\ccor{\hat\Op_{v_1}\dots\hat\Op_{v_n}}=\lim_{w\rightarrow v_1}\left(\lim_{z\rightarrow v_1}\frac{\ccor{\pa\Phi(z)\pa\Phi(w)\hat{\Op}_{v_1}\dots\hat{\Op}_{v_n}}}{\ccor{\hat{\Op}_{v_1}\dots\hat{\Op}_{v_n}}}(-2) (z-v_1)-\frac{1}{4(w-v_1)}\right)
\end{equation}

Therefore, if we put together (\ref{log_ising}) and (\ref{log_gff}), we obtain \eqref{eq: log_der}.

\end{proof}

\subsection{Concluding the proof}
\begin{proof}[Proof of Theorem \ref{th:ds}]
As a first step, we handle the case $\Op_{v_i}=\sigma_{v_i}$ all $i$. Indeed, we then always have $\ccor{\hOp_{v_1}\dots\hOp_{v_n}}>0$, as the cosines expand into exponentials with positive coefficients. Moreover, we also have $\ccor{\Op_{v_1}\dots\Op_{v_n}}>0$, as defined in \cite[Section 5.1]{CHI2}; note that $n$ is even. Therefore, \eqref{eq: log_der} is valid for \emph{all} configurations of distinct points $v_1,\dots,v_n$. Integrating and taking exponentials implies 
\begin{equation}
\label{eq: up_to_const}
\ccor{\Op_{v_1}\dots\Op_{v_n}}^2=C_{\Omega,n}\cdot\ccor{\hat\Op_{v_1}\dots\hat\Op_{v_n}}
\end{equation}
for a constant $C_{\Omega,n}$ independent of $v_1,\dots,v_n$. The value $C_{\Omega,n}=1$ can now be fixed by induction, letting $v_1\to v_2$ and applying 
the OPEs \eqref{eq: fuse_cos_cos} and \cite[Theorem 6.3]{CHI2}. 

Notice that in particular, this implies that $\ccor{\sigma_{v_1}\dots\sigma_{v_n}}$ is a real analytic function of (distinct) points $v_1,\dots,v_n$, since, as noted in Section \ref{sec: bosonic}, $\ccor{\hOp_{v_1}\dots\hOp_{v_n}}$ is always manifestly real analytic. Furthermore, we have \eqref{eq: for_log_der}; thus using \cite[Theorem 6.2]{CHI2}, we can write 
$$
\ccor{\mu_{v_1}\mu_{v_2}\sigma_{v_3}\dots\sigma_{v_n}}=\ccor{\sigma_{v_1}\dots\sigma_{v_n}}\oint\oint\left(\frac{R(z,w,v_1,\dots,v_n)}{(z-v_1)(w-v_2)}\right)^\frac12\,dzdw,
$$
where $R$ is a real-analytic function of $v_1,\dots,v_n$, and the small contours of integration around $v_{1,2}$ can be chosen to avoid its zeros. Therefore, $\ccor{\mu_{v_1}\mu_{v_2}\sigma_{v_3}\dots\sigma_{v_n}}$ is a real-analytic function of $v_1,\dots,v_n$. Using the Pfaffian idenity \cite[Eq 5.19]{CHI2}, we conclude that $\ccor{\Op_{v_1}\dots\Op_{v_n}}$, $\Op_{v_i}\in\{\sigma_{v_i},\mu_{v_i}\}$ is always real analytic. 

We are ready to conclude the proof in the general case by induction on $n$. Indeed, assume without loss of generality that $\Op_{v_1}=\mu_{v_1}$ and $\Op_{v_2}=\mu_{v_2}$. Fix $v_3,\dots, v_n$ such that $\ccor{\Op_{v_3}\dots\Op_{v_n}}\neq 0$, then by induction hypothesis also $\ccor{\hOp_{v_3}\dots\hOp_{v_n}}\neq 0$. Applying once again the OPE \eqref{eq: fuse_sin_sin} and \cite[Theorem 6.3]{CHI2}, we see that $\ccor{\Op_{v_1}\dots\Op_{v_n}}\neq 0$ and $\ccor{\hOp_{v_1}\dots\hOp_{v_n}}\neq 0$ for $v_1$ in a small neighborhood of $v_2$. In particular, \eqref{eq: log_der} is true in that neighborhood, and by sending $v_1\to v_2$ to fix the multiplicative normalization, we conclude that the result of Theorem \eqref{th:ds} holds true in that neighborhood. But since both sides of the desired identity are real-analytic functions of $v_1$, it holds true whenever $v_1\neq v_2,\dots,v_n$ and $\ccor{\Op_{v_3}\dots\Op_{v_n}}\neq 0$. Since the latter condition is open, we can drop it by applying the real analyticity once again.
\end{proof}


\begin{thebibliography}{99}
\bibitem{Adler} Adler, Robert J., Taylor, Jonathan E., Random fields and Geometry. Springer, 2007.
\bibitem{Ahlfors} Ahlfors, L. V., Sario, L. Riemann Surfaces: (PMS-26) (Vol. 58). Princeton university press, 2015.
\bibitem{AmKytetc} Ameen, T., Kytölä, K., Park, S. C., and Radnell, D. (2022). Slit-strip Ising boundary conformal field theory 1: Discrete and continuous function spaces. Mathematical Physics, Analysis and Geometry, 25(4), 30.
\bibitem{AmKytetc2} Ameen, T., Kytölä, K., and Park, S. C. "Slit-strip Ising boundary conformal field theory 2: Scaling limits of fusion coefficients." arXiv preprint arXiv:2108.05105 (2021)
\bibitem{Atiyah} Atiyah, Michael F. Riemann surfaces and spin structures. In Annales scientifiques de l'École Normale Supérieure Vol. 4, No. 1, pp. 47-62, 1971.
\bibitem{Basok} Basok, Mikhail. Dimers on Riemann surfaces and compactified free field. Preprint arXiv:2309.14522.
\bibitem{BPZ} Belavin, Alexander A., Polyakov, Alexander M. and Zamolodchikov, Alexander B. "Infinite conformal symmetry in two-dimensional quantum field theory." Nuclear Physics B 241.2 (1984): 333-380.
\bibitem{Bergman}  Bergman, Stefan. The kernel function and conformal mapping. Second, revised edition. Mathematical Surveys, No. V. American Mathematical Society, Providence, R.I., 1970. x+257 pp.
\bibitem{BG} Burkhardt, Theodore W., and Guim, Ihnsouk. "Conformal theory of the two-dimensional Ising model with homogeneous boundary conditions and with disordred boundary fields." Physical Review B 47.21 (1993)
\bibitem{CHI1} Chelkak, Dmitry; Hongler, Cl\'ement; Izyurov, Konstantin. Conformal invariance of spin correlations in the planar Ising model. Ann. of Math. (2) 181 (2015), no. 3, 1087--1138.
\bibitem{CHI2} Chelkak, Dmitry; Hongler, Cl\'ement; Izyurov, Konstantin. Correlations of primary fields in the critical Ising model. Preprint arXiv:2103.10263.
\bibitem{Cohn}  Cohn, Harvey. Conformal mapping on Riemann surfaces. Reprint of the 1967 edition. Dover Books on Advanced Mathematics. Dover Publications, Inc., New York, 1980. xv+325 pp. 
\bibitem{Conway} Conway, John B. Functions of one complex variable II. Vol. 159. Springer Science and Business Media, 2012.
\bibitem{DFMS} Di Francesco, Philippe; Mathieu, Pierre; S\'en\'echal, David Conformal field theory. Graduate Texts in Contemporary Physics. Springer-Verlag, New York, 1997. xxii+890 pp.
\bibitem{DiFSZ} Di Francesco, Ph, H. Saleur, and J. B. Zuber. "Critical Ising correlation functions in the plane and on the torus." Nuclear Physics B 290 (1987): 527-581.
\bibitem{DLMF} NIST Digital Library of Mathematical Functions. https://dlmf.nist.gov/, Release 1.1.11 of 2023-09-15. F. W. J. Olver, A. B. Olde Daalhuis, D. W. Lozier, B. I. Schneider, R. F. Boisvert, C. W. Clark, B. R. Miller, B. V. Saunders, H. S. Cohl, and M. A. McClain, eds.

\bibitem{DotsenkoFateev} Dotsenko, Vl S., and Vladimir A. Fateev. "Conformal algebra and multipoint correlation functions in 2D statistical models." Nuclear Physics B 240.3 (1984): 312-348.
\bibitem{Dubedat09} Dubédat, Julien. "SLE and the free field: partition functions and couplings." Journal of the American Mathematical Society 22.4 (2009): 995-1054.
\bibitem{Dubedat15} Dubédat, Julien. "Dimers and families of Cauchy-Riemann operators I." Journal of the American Mathematical Society 28.4 (2015): 1063-1167.
\bibitem{Dubedat} Dub\'edat, Julien. Exact bosonization of the Ising model. Preprint arXiv:1112.4399.
\bibitem{Hugoetc} Duminil-Copin, Hugo, and Lis, Marcin. "On the double random current nesting field." Probability Theory and Related Fields 175 (2019): 937-955.
\bibitem{Evans}  Evans, Lawrence C. Partial differential equations. Second edition. Graduate Studies in Mathematics, 19. American Mathematical Society, Providence, RI, 2010. xxii+749 pp.
\bibitem{FK}  Farkas, Hershel M.; Kra, Irwin. Riemann surfaces. Graduate Texts in Mathematics, 71. Springer-Verlag, New York-Berlin, 1980. xi+337 pp.
\bibitem{Fay}  Fay, John D. Theta functions on Riemann surfaces. Lecture Notes in Mathematics, Vol. 352. Springer-Verlag, Berlin-New York, 1973. iv+137 pp.
\bibitem{Felder} Felder, Giovanni. "BRST approach to minimal models." Nuclear Physics B 317.1 (1989): 215-236.
\bibitem{Flores-Kleban} Flores, Steven M., and Peter Kleban. "A solution space for a system of null-state partial differential equations: Parts I--IV." Communications in mathematical physics 333 (2015): 389--715
\bibitem{FloresPeltola} Flores, Steven M., and Eveliina Peltola. "Monodromy invariant CFT correlation functions of first column Kac operators," in preparation.
\bibitem{Fuks} Fuks, B. A.  Theory of analytic functions of several complex variables, American Mathematical Soc. (1963).
\bibitem{HS} Hawley, N.S. and Schiffer, M. Half-order differentials on Riemann surfaces, Acta Math. 115 (1966) 199-236.
\bibitem{Hejhal} Hejhal, Dennis A. Theta functions, kernel functions, and Abelian integrals. Memoirs of the American Mathematical Society, No. 129. American Mathematical Society, Providence, R.I., 1972. iii+112 pp. 
\bibitem{HonKytVik} Hongler, Clément, Kalle Kytölä, and Fredrik Viklund. "Conformal field theory at the lattice level: discrete complex analysis and Virasoro structure." Communications in Mathematical Physics 395.1 (2022): 1-58.
\bibitem{Johnson} Johnson, Dennis. Spin structures and quadratic forms on surfaces. Journal of the London Mathematical Society, 2(2), (1980),  365-373. 
\bibitem{JSW}  Junnila, Janne; Saksman, Eero; Webb, Christian. Imaginary multiplicative chaos: moments, regularity and connections to the Ising model. Ann. Appl. Probab. 30 (2020), no. 5, 2099--2164.
\bibitem{KM} Kang, Nam-Gyu; Makarov, Nikolai G. Gaussian free field and conformal field theory. Ast\'erisque No. 353 (2013), viii+136 pp.
\bibitem{OK} Kaufman, B., and Onsager, L. (1949). Crystal statistics. III. Short-range order in a binary Ising lattice. Physical Review, 76(8), 1244.
\bibitem{KytolaPeltola} Kytölä, Kalle, and Eveliina Peltola. "Conformally covariant boundary correlation functions with a quantum group." Journal of the European Mathematical Society 22.1 (2019): 55-118.
\bibitem{LRV}  Lacoin, Hubert; Rhodes, R\'emi; Vargas, Vincent Complex Gaussian multiplicative chaos. Comm. Math. Phys. 337 (2015), no. 2, 569--632.
\bibitem{ML}  Mattis, Daniel C.; Lieb, Elliott H. Exact solution of a many-fermion system and its associated Boson field. J. Mathematical Phys. 6 (1965), 304--312.
\bibitem{MW} McCoy, Barry M., and Tai Tsun Wu. The two-dimensional Ising model. Harvard University Press, 1973.
\bibitem{Onsager} Onsager, L. (1944). Crystal statistics. I. A two-dimensional model with an order-disorder transition. Physical Review, 65(3-4), 117.
\bibitem{Palmer} Palmer, John. Planar Ising Correlations. Vol. 49. Springer Science and Business Media, 2007.
\bibitem{SS} Schiffer, Menahem; Spencer, Donald C. Functionals of finite Riemann surfaces. Princeton University Press, Princeton, N. J., 1954. x+451 pp.
\bibitem{Sheffield} Sheffield, Scott.``Gaussian free fields for mathematicians." Probability theory and related fields 139.3-4 (2007): 521-541.
\bibitem{Sussman} Sussman, Ethan, The regularization of Dotsenko--Fateev integrals. Lett. Math. Phys. 113.100 (2023).
\bibitem{Yang} Yang, C. N. (1952). The spontaneous magnetization of a two-dimensional Ising model. Physical Review, 85(5), 808.


\end{thebibliography}
\end{document}